%% file: main.tex
\documentclass{article}

\usepackage{microtype}
\usepackage{graphicx}
\usepackage{caption}
\usepackage{subcaption}
\usepackage{booktabs} %

\DeclareGraphicsExtensions{.pdf}

\usepackage{soul}
\usepackage{hyperref}
\usepackage[dvipsnames]{xcolor}

\usepackage[accepted]{icml2021}

\usepackage{amsmath,amssymb,amsthm,textcomp}

\newtheorem{theorem}{Theorem}[section]
\newtheorem{lemma}[theorem]{Lemma}

\newtheorem{corollary}[theorem]{Corollary}

\newtheorem{assumption}[theorem]{Assumption}

\newtheorem{definition}[theorem]{Definition}
\usepackage[suppress]{color-edits}
\addauthor{vs}{red}

\newcommand{\CI}{\mathrm{CI}}

\newcommand{\E}{\ensuremath{\mathbb{E}}}
\newcommand{\Var}{\ensuremath{\mathbb{V}}}
\newcommand{\noise}{\ensuremath{\text{noise}}}
\newcommand{\adj}{\ensuremath{\text{adj}}}
\newcommand{\new}{\ensuremath{\text{new}}}

\newcommand{\R}{\ensuremath{\mathbb{R}}}
\newcommand{\mcS}{\ensuremath{{\cal S}}}
\newcommand{\mcT}{\ensuremath{{\cal T}}}
\newcommand{\mcF}{\ensuremath{{\cal F}}}
\newcommand{\mcD}{\ensuremath{{\cal D}}}
\newcommand{\indep}{\perp \!\!\! \perp}
\newcommand{\fdim}{d}
\newcommand{\mean}{\ensuremath{\text{mean}}}
\newcommand{\Cov}{\ensuremath{\mathtt{Cov}}}

\icmltitlerunning{Estimating the Long-Term Effects of Novel Treatments}
\onecolumn

\begin{document}

\icmltitle{Estimating the Long-Term Effects of Novel Treatments:\\ The Dynamically Adjusted Surrogate Index}

\icmlsetsymbol{equal}{*}

\begin{icmlauthorlist}
\icmlauthor{Keith Battocchi}{msr}
\icmlauthor{Eleanor Dillon}{msr}
\icmlauthor{Maggie Hei}{msr}
\icmlauthor{Greg Lewis}{msr}
\icmlauthor{Miruna Oprescu}{msr}
\icmlauthor{Vasilis Syrgkanis}{msr}
\end{icmlauthorlist}

\icmlaffiliation{msr}{Microsoft Research, New England}

\icmlcorrespondingauthor{Vasilis Syrgkanis}{vasy@microsoft.com}

\icmlkeywords{long-term effects, dynamic treatment effects, surrogates, high-dimensional, double machine learning}

\vskip 0.3in

\printAffiliationsAndNotice{}  %

\begin{abstract}
Policy makers typically face the problem of wanting to estimate the long-term effects of novel treatments, while only having historical data of older treatment options. We assume access to a long-term dataset where only past treatments were administered and a short-term dataset where novel treatments have been administered. We propose a surrogate based approach where we assume that the long-term effect is channeled through a multitude of available short-term proxies. Our work combines three major recent techniques in the causal machine learning literature: surrogate indices, dynamic treatment effect estimation and double machine learning, in a unified pipeline. We show that our method is consistent and provides root-n asymptotically normal estimates under a Markovian assumption on the data and the observational policy. We use a data-set from a major corporation that includes customer investments over a three year period to create a semi-synthetic data distribution where the major qualitative properties of the real dataset are preserved. We evaluate the performance of our method and discuss practical challenges of deploying our formal methodology and how to address them.
\end{abstract}

\section{Introduction}

Businesses frequently invent new ways of interacting with their customers. Marketing departments frequently devise new marketing campaigns. Pharmaceutical companies typically roll out trials of new drugs. In a multitude of domains, policy makers want to understand the long-term effects of novel treatments, recently deployed, while only having access to short-term data following their deployment. Such policy makers only have access to long-term historical data where other treatments had been deployed.

We propose an estimation methodology that leverages historical data to derive estimates of the long-term treatment effects of novel treatments. The seminal paper of \citet{athey2020estimating} proposes a surrogate index methodology as one solution to this problem. The fundamental assumption of the surrogate index is that there exist short-term proxies that are observed in the short-term dataset and that causal effects on long-term outcomes have to be primarily channeled through these short-term signals. In other words, the treatment has a long-term causal effect, if and only if it has an effect on a variety of short terms signals. This allows one to use the historical long-term data set to learn a mapping from short-term signals to a projected long-term reward --- referred to as the surrogate index --- and subsequently, estimate the causal effect of novel treatments on the surrogate index.

But often historical treatment policies are \emph{dynamic}: treatments are assigned repeatedly, and their assignments depend on past treatments and short-term outcomes. This can easily break the assumptions needed for the surrogate index to work and introduce bias.
For example, suppose that a firm offers multiple investments to a particular customer in the historical data and these investments are auto-correlated, i.e. if a customer receives an investment this month, then they will receive an investment with high probability in one of the subsequent months. These future investments can substantially increase the long-term outcome of interest, and this increase will be attributed to the short-term proxies.  The surrogate index thus formed will tend to over-predict long-run outcomes, and so when it is used to measure the treatment effect of some new treatment in the short-run data, the estimated treatment effect will be bigger (in absolute magnitude) than the truth.  

The main methodological innovation of this paper is to suggest using the dynamic treatment effect analysis of \cite{lewis2020doubledebiased} on the historical data in order to create an unbiased dynamically adjusted surrogate index. Our dynamically adjusted surrogate index takes the interpretation of the projected long-term reward in the absence of any future treatments. Applying this dynamically adjusted surrogate model to the short-term dataset leads to unbiased causal effect estimates of the long-term effects of the novel treatments.

A second contribution of the paper is to generalize the causal analysis step that uses the experimental sample to allow for multiple continuous treatments rather than a single binary treatment. 
We do this by proposing a new estimator for this expanded surrogate approach which allows for the construction of valid confidence intervals, even when using flexible machine learning models at both stages of the estimation to deal with high-dimensional data.
In short, we show how one can combine three recently developed techniques, i.e. i) the surrogate index approach of \cite{athey2020estimating}, ii) the double machine learning approach of \cite{chernozhukov2018} and iii) the dynamic treatment effect estimation approach of \cite{lewis2020doubledebiased}, in a single data analysis pipeline to estimate treatment effects in the presence of dynamic treatment policies.

Our work lies in the broader field of estimating causal effects with machine learning and Neyman orthogonality \cite{Neyman:1979,robinson:88,Ai2003,Chernozhukov2016locally,chernozhukov2018double}. Moreover, it relates to the work on machine learning estimation of treatment effects in the dynamic treatment regime \cite{nie2019learning,Thomas2016,petersen2014targeted,kallus2019efficiently,kallus2019double,lewis2020doubledebiased,bodoryevaluating,singh2020kernel} and on structural nested models in biostatistics \cite{robins1986new,robins1992g,robins1994correcting,robins1997toward,robins2000marginal,lok2012impact,vansteelandt2014structural,Vansteelandt2016}. Finally, it relates to the surrogacy literature in causal inference
\cite{prentice1989surrogate,begg2000use,frangakis2002principal,freedman1992statistical}. 
Our work builds on insights in these works and proposes the first complete method that combines all three lines of work, so as to estimate the long-term effect of novel treatments from observational data that stem from a dynamic observational policy and in a manner that allows for high-dimensionality.

\section{Problem and Methodology}
Though our methodology applies in many domains, for concreteness we consider the running example of a firm making investments in its customers in order to increase subsequent purchases by those customers.

A firm has a number of distinct investments/treatments $T_1, T_2 \ldots T_k$ it offers to its customers. At each period $t$ (e.g. the period of a month), a vector of treatments $T_{i,t}=(T_{i, t,1}, \ldots, T_{i,t,k})\in \R^k$, is applied to each customer $i$. We also observe a vector of $p$ characteristics $X_{i, t}=(X_{i, t, 1}, \ldots, X_{i,t,p})\in \R^p$, some of which are constant within customer (e.g. industry) and some of which vary over time and customer (e.g. last month's revenue). We observe the outcome of interest $Y_{i, t}\in \R$ (e.g. monthly revenue).

We are interested in identifying the average effect of each treatment at some period $t$, on the cumulative outcome in the subsequent $M$ periods, i.e. $\bar{Y}_{i, t} = \sum_{\kappa=0}^{M} Y_{i,t + \kappa}$. 
We assume that for most customers and periods we do not observe the subsequent $M$ periods following a treatment.
This makes direct inference from the historical data impractical.
Instead, we assume that 
for every customer $i$ and period $t$, we have access to a vector of $d$ short-term proxies/surrogates, $S_{i, t}=(S_{i, t, 1}, \ldots, S_{i, t, d})$. In practice, this can for instance include the next few months of revenue and other measures that are indicative of a customer's trajectory.

Given these surrogate measures, \citet{athey2020estimating} propose the following estimation strategy:

\begin{enumerate}
    \item Begin with a long-term \emph{observational data set} that, for each period and customer, consists of customer characteristics $X_{i,t}$, customer surrogates for growth $S_{i,t+1}$ and realized $M$-period outcomes $\bar{Y}_{i,t}$. Notably, this data-set does not need to contain the treatments whose causal effect we want to measure and hence can use a longer span of historical data even if treatments are only measured recently. Based on this data-set we construct a model that tries to predict the target outcome $\bar{Y}_{i,t}$ from the surrogates $S_{i,t}$ and customer characteristics. Mathematically, this model $\hat{g}(S_{i,t}, X_{i,t})$ estimates the conditional expectation:
    \begin{equation}
    g_0(S_{i,t}, X_{i,t}):=\E[\bar{Y}_{i,t} \mid S_{i,t}, X_{i,t}].  \tag{surr. model}
    \end{equation}
    and corresponds to a simple regression task. Any machine learning estimation algorithm can be used to construct $\hat{g}$. 
    \item Next, use a second short-term, \emph{experimental data set} that includes customer features, $X_{i,t}$, surrogates, $S_{i,t}$, and treatments $T_{i,t}$.\footnote{We note that the set of customers does not need to be the same in the two data-sets.} This data set is restricted to the period where we can measure treatments, but only requires a few months of leading surrogates rather than $M$ periods of leading outcomes. For any such sample we calculate the surrogate index:
    \begin{equation}
    I_{i,t} := \hat{g}(S_{i,t}, X_{i,t}), \tag{predicted surrogate index}    
    \end{equation}
    which is the predicted long-term outcome from the observed short-term proxies. This predicted index becomes the outcome in the final causal model.
\end{enumerate}

This standard approach will be consistent for the treatment effects under two assumptions. The first is that the relationship between the surrogate variables and the target outcome remains unchanged over our historical sample period and the periods into which we project future outcomes. We follow \citet{athey2020estimating} in maintaining this assumption. The model is able to account for steady growth in a natural way: as the surrogates grow, the predicted outcomes grow too.\footnote{It is not robust to changes in the mapping between the surrogates and outcomes over time, and so we would recommend refitting the model periodically to account for this.}

The second main assumption is that the only causal path from the treatment to the outcome goes through the surrogates. In other words, the treatment has an effect on $M$-period outcomes if and only if it affects these surrogates. This critical ``sufficiency'' assumption typically requires that a few periods of post-treatment data are available to build a relatively good understanding of the $M$-period path. 

Serially correlated treatment policies will often violate this second assumption.
A customer who receives a treatment in period $t$ may be more likely to receive a second treatment over the next year, even controlling for observed features $X_{it}$.  
This positive correlation could be either because a customer that is treated subsequently becomes more salient to the firm, or because the treatment successfully drives surrogate metrics higher, thereby improving the perceived return on subsequent treatments of this customer (an example could be a firm that targets its fast-growing customers).

In either case, the standard estimated surrogate index model $\hat{g}$ will now be biased; we would wrongly estimate that even small increases in our proxy measures forecast strong growth in $M$-period outcomes on average. If we subsequently estimate the causal effect of novel treatments on the surrogate index, we will over-estimate the causal effect on $M$-period outcomes. Essentially, causal effect that should be attributed to the future treatment gets doubly attributed to the current treatment. 
To remove this bias we include a dynamic adjustment step based on \cite{lewis2020doubledebiased} that alters our observational data-set to remove the effect of future treatments from the target outcome before estimating our surrogate index model.
To show this why this method is both necessary and sufficient for removing bias, we frame the problem and solution in approach in a two-period example before formally presenting the general case in section \ref{sec:math}.

\subsection{A two-period example}

Assume that $\bar{Y}_{i,t} = Y_{i,t} + Y_{i, t+1}$ (i.e. the long-term outcome contains two periods). 
We assume a linear Markovian model of how the random variables evolve over time. 
We collapse $X_{i,t}$ and $S_{i,t-1}$ so that the surrogates at each period and the controls in the next period are all denoted by $S_{i,t-1}$. 
In other words, all customer observable characteristics in the next period are candidate surrogates and become controls for the period after next. 
This is without loss of generality. Finally, we focus on a single scalar investment. 

With these simplifications, the structural model that describes how all the random variables evolve can be written in three equations:
\begin{align}
    S_{i,t} =~& A T_{i,t} + B S_{i,t-1} + \epsilon_{i,t} \tag{evolution of controls}\\
    Y_{i,t} =~& \gamma' S_{i,t} + \zeta_{i,t} \tag{outcome equation}\\
    T_{i,t+1} =~& \kappa T_{i,t} + \lambda' S_{i,t} + \eta_{i,t} \tag{treatment policy}
\end{align}
where $A$ is an $p\times 1$ matrix, $B$ is a $p\times p$ matrix, $\gamma, \lambda$ are $p$-dimensional vectors and $\kappa$ is a scalar. The terms $\epsilon_{i,t}, \zeta_{i,t}, \eta_{i,t}$ are exogenous mean-zero independent noise terms. For conciseness we drop the customer index $i$ in the remaining equations.

\paragraph{Target effect estimand}
The effect of treatment $T_{t}$ on the long-term outcome $\bar{Y}_{t}$ can be derived as follows:
\begin{align*}
    Y_{t} =~& \gamma' S_{t} + \noise
          =~\gamma' A T_t + \gamma'B S_{t-1} + \noise\\
    Y_{t+1} =~& \gamma' S_{t+1} + \noise
    =~ \gamma'A T_{t+1} + \gamma'B S_{t} + \noise\\
    =~& \gamma'A T_{t+1} + \gamma' B A T_{t} + \gamma' B^2 S_{t-1} + \noise
\end{align*}
Thus we see that the effect of $T_t$ on $Y_{t+1} + Y_{t}$, keeping all other variables fixed, is:
\begin{equation}
\theta_0 :=\gamma' (B + I)A.
\end{equation}
A one-unit increase in investment $T_t$, leads to a total of $\theta_0$ units more revenue in the next two periods assuming future treatments are held constant.

\paragraph{Surrogate index without dynamic adjustment} 
If we train a surrogate index by regressing $\bar{Y}_{t}=Y_{t} + Y_{t+1}$ on $S_{t}$, as is the method proposed in \citet{athey2020estimating}, then this would result in the following surrogate index:
\begin{align}
\begin{split}
    g_0(S_{t}) =~& \E[\bar{Y}_t \mid S_{t}]
    = \E[Y_{t} \mid S_{t}] + \E[Y_{t+1} \mid S_{t}]\\
    =~& \gamma' S_{t} + \E[\gamma'S_{t+1} \mid S_{t}]\\
    =~& \gamma' S_{t} + \E[\gamma'AT_{t+1} + \gamma'B S_{t} \mid S_{t}]\\
    =~& \gamma' (I + B) S_{t} + \gamma'A \E[T_{t+1}\mid S_{t}]
\end{split}
\end{align}
Subsequently, if we estimate the causal effect of $T_t$ on $g_0(S_{t})$, then this effect would be:
\begin{align}
\begin{split}
    \theta_* =~& \gamma' (I + B) A +  \gamma' A \frac{\E[\E[T_{t+1}\mid S_{t}]\, T_{t}]}{\E[T_t^2]}\\
    =~& \theta_0 +  \gamma' A \frac{\E[\E[T_{t+1}\mid S_{t}]\, T_{t}]}{\E[T_t^2]} \\
    =~& \theta_0 +  \gamma' A \kappa \frac{\E[\E[T_{t}\mid S_{t}]\, T_{t}]}{\E[T_t^2]} + \gamma'A\, \lambda'A
\end{split}
\end{align}

The standard estimate contains a bias stemming from the fact that investment today can lead to higher investment tomorrow, either directly (i.e. that $\kappa \geq 0$) or indirectly through the surrogates (i.e. that $\lambda'A \geq 0$). The standard surrogate approach is valid only when  the investment policy is not adaptive, i.e. $\kappa=0$ and $\lambda=0$, since then $T_{t+1}$ would be independent of $S_{t}$. We see here that the bias is larger if the investment policy is highly auto-correlated, which is typical in practice. The two possible channels for bias, through direct auto-correlation or surrogate-dependent treatment policies, can result in substantially biased estimates.

\paragraph{Dynamic adjustment} 
The goal of the dynamic adjustment is to remove the effect of the next-period treatment from the long-term outcome $\bar{Y}_{i, t}$. We achieve this by estimating a separate causal effect of $T_{i,t+1}$ on $Y_{i,t+1}$, controlling for $X_{i,t+1}, S_{i,t}$. Observe that this conditional expectation is equal to:
\begin{equation}
    \E[Y_{t+1} \mid T_{t+1}, S_{t}] = \gamma'A T_{t+1} + \gamma'BS_{t}.
\end{equation}
We then subtract the causal effect, $\alpha_{t+1}$, from $Y_{i, t+1}$ to get a \emph{dynamically adjusted outcome}, $Y_{i, t+1}^{\adj} := Y_{i, t+1} - \alpha_{t+1}' T_{t+1}$. Finally, we can create an \emph{adjusted long-term outcome}, $\bar{Y}_{i, t}^{\adj} := Y_{i, t} + Y_{i, t+1}^{\adj}$.

We then build a new \emph{dynamically adjusted surrogate index}, which is the projected adjusted long-term outcome, conditional on the observed surrogates:
\begin{align}
\begin{split}
    g_0^{\adj}(S_{t}) =~& \E[\bar{Y}_t^{\adj} \mid S_{t}]\\
    =~& \E[Y_{t}\mid S_{t}] + \E[Y_{t+1}^{\adj}\mid S_{t}]\\
    =~& \gamma'S_{t} + \E[Y_{t+1} - \gamma'A T_{t+1} \mid S_{t}]\\
    =~& \gamma'S_{t} + \E[\gamma'BS_{t} \mid S_{t+1}]\\
    =~& \gamma' (I + B) S_{t}
\end{split}
\end{align}
This new index captures the projected $M$-period outcomes as if the customer was offered no future treatments.

When we estimate the effect of the treatment based on this adjusted surrogate index we recover:
\begin{align}
\begin{split}
    \E[g_0^{\adj}(S_{t}) \mid T_{t}, S_{t-1}] =~& \gamma' (I + B) \E[S_{t}\mid T_t, S_{t-1}]\\ 
    =~& \gamma' (I+B) A T_t + \gamma'(I+B) B S_{t-1}
\end{split}
\end{align}
With the dynamic adjustment, the coefficient in front of $T_t$ that we recover is the true causal effect $\theta_0=\gamma'(I+B)A$. \citet{lewis2020doubledebiased} show how this adjustment approach can be extended to many periods, via a recursive peeling process, and also to high-dimensional surrogates and controls via a dynamic double machine learning approach.

\paragraph{Estimating causal effects of new treatments} 
With this adjusted surrogate index in hand, we can also estimate the long-term effect of any other treatment that was introduced more recently. In this example, the stationarity assumption of both proposed surrogate approaches requires that $B$, which governs how the surrogate evolves, and $\gamma$, which governs how surrogates translate to per-period outcomes, do not change between the observational and experimental data-sets. These two parameters govern how surrogates today relate to future outcomes in the absence of any treatment. 

Under such a condition, if we introduce a new treatment $T_t^{\new}$, which has a different effect $A^{\new}$ on the surrogates (and hence on the long-term outcome), then the effect of this treatment on the long-term outcome is $\theta_0^{\new}=\gamma'(I+B)A^{\new}$. This $\theta_0^{\new}$ is exactly the outcome of estimating the causal effect of $T_t^{\new}$ on $g_0^{\adj}(S_{t+1})$ controlling for $S_t$.

\section{Formal Setting}\label{sec:math}

In this section we present the general problem formulation and in the subsequent sections we present our formal main results.
There are a number of innovations beyond the basic strategy presented in the two period example above.
First, we develop a generalization of the doubly robust estimation method in \citet{athey2020estimating} to the case of multiple continuous treatments, under a semi-parametric assumption (c.f. Section~\ref{ssec:ddsurr_noadj}) and in the presence of a dynamic treatment policy in the observational data (c.f. Section~\ref{ssec:ddsurr_wadj}). Second, we make use of orthogonal machine learning techniques \cite{chernozhukov2018} throughout, to allow for a rich set of potential confounders and valid analytic confidence intervals.

We assume that we have access to two sample populations: an experimental population, denoted as $e$, and an observational population, denoted as $o$. A sample from each population consists of a finite horizon time-series $(S_0, T_1, S_1, Y_1, T_2, S_2, Y_2, \ldots, T_M, S_M, Y_M)$. We observe the full $M$-period time series for each sample in the observational population, but we only observe $(S_0, T_1, S_1)$ for each sample from the experimental population. \vsedit{Moreover, the random variables in the two populations could be distributed differently and even have different support (e.g. treatments in the population $e$ could be different from treatments in population $o$).} As in the two period example, we simplify without loss of generality by merging the control variables in period $t$ and the surrogates for period $t-1$. In other words, all next-period control variables serve as potential surrogates and vice versa. We assume that the data obey the Markovian assumptions depicted in the causal graph in Figure~\ref{fig:simplified_cg}.

\begin{figure*}[thb!]
    \centering
   \includegraphics[height=2in]{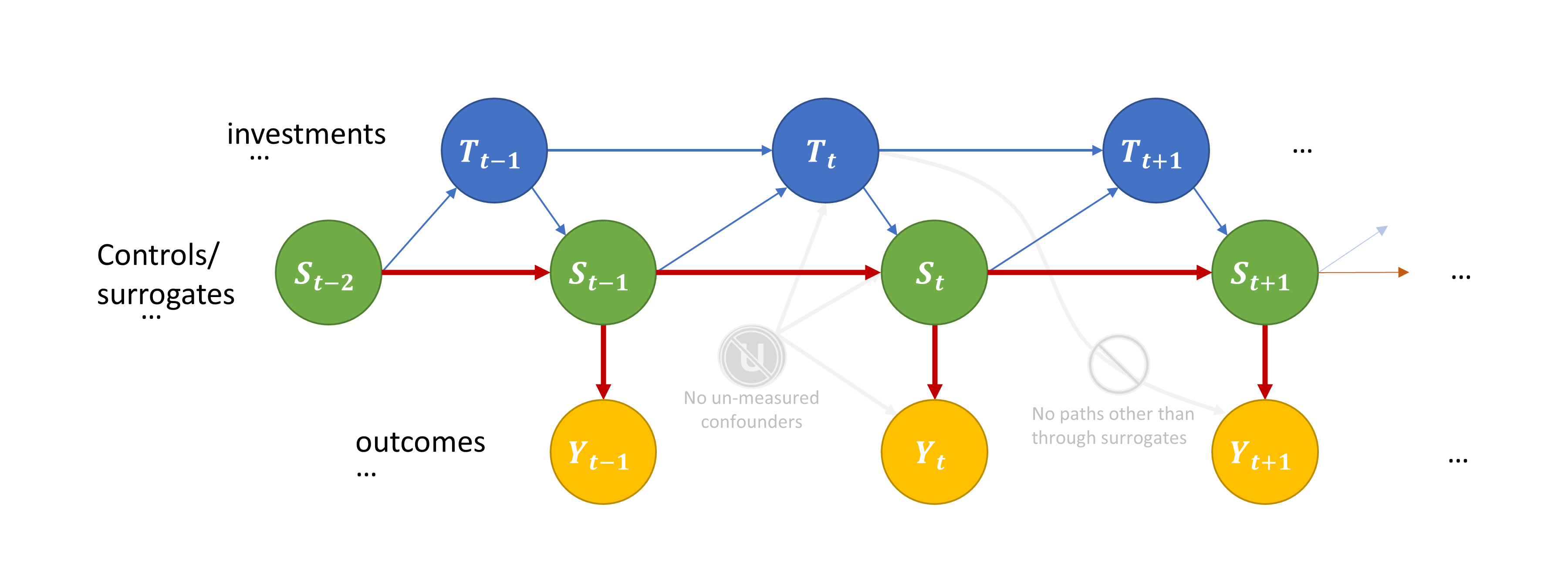}
    \caption{\footnotesize{Causal graph representation of the main assumptions of our causal analysis, which are:  i) there are no paths to future outcomes from past investments that do not go through surrogates, ii) the relationships designated with a red arrow have to remain un-changed between the long-term data-set and the short-term data-set, iii) there are no un-measured confounders at each period, i.e. variables that are not observed and which affect the treatment assignment and the outcome directly.}}
\label{fig:simplified_cg}
\end{figure*}

Our goal is to estimate the causal effect of treatment vector $T_1$ on the long term outcome:
\begin{align}
\bar{Y} := \sum_{t=1}^M Y_t \tag{long-term outcome}
\end{align}
in the experimental/short-term sample, \vsedit{assuming all future treatments take some baseline value. For simplicity we use the $0$ value as the baseline, but this could be replaced by any baseline treatment vector value. In other words, if we set the future treatments $T_{>1}:=(T_1,\ldots, T_M)$ that each sample receives to the baseline level of $0$ and we change the treatment $T_1$ from some value $t_0$ to some other value $t_1$, then what is the change on the long term outcome $\bar{Y}$, i.e.:
\begin{align}\label{eqn:target-do}
    \tau(t_1, t_0) :=~& \E_e[\E[\bar{Y}\mid do(T_1=t_1, T_{>1}=0)] - \E[\bar{Y}\mid do(T_1=t_0, T_{>1}=0]] \tag{target}
\end{align}
}

We present our theoretical results in two steps. In the first setting, we assume that treatments happen only at period $1$ (Section~\ref{ssec:ddsurr_noadj}). This is the setting analyzed in \citet{athey2020estimating}, albeit only for the case of a single binary treatment $T_1$.\footnote{We note that the work of \citet{athey2020estimating} also allowed for estimation of average treatment effects, even in the case when there is arbitrary treatment effect heterogeneity. In this work, we assume that treatment effects are constant. A generalization to the case of arbitrary treatment effect heterogeneity is feasible, but would require the estimation of conditional covariance matrices, which would make the estimation algorithm more brittle and the exposition much more complex.} We then show how this approach can be modified to incorporate a dynamic treatment policy in the observational and experimental sample (Section~\ref{ssec:ddsurr_wadj}).

\paragraph{Notation} Throughout, we will denote with $\E_e[\cdot]$ the expectation conditional on the experimental population and $\E_o[\cdot]$ the expectation conditional on the observational population. Moreover, for any vector-valued function $f$ that takes as input a random variable $Z$, we denote with:
\begin{equation}
    \|f\|_{2,o} = \sqrt{\E_o\left[\|f(Z)\|_2^2\right]}
\end{equation}
and analogously $\|f\|_{2,e}$. We denote with $\E_n[\cdot]$, the empirical expectation over all the samples, i.e. for any random variable $Z$, $\E_n[Z]:=\frac{1}{n}\sum_i Z_i$, and with $\E_{e,n}$ and $\E_{o,n}$ the empirical expectation over the experimental and observational samples correspondingly, i.e. $\E_{e,n}[Z] = \frac{1}{n_e} \sum_{i\in e} Z_i$ and $\E_{o,n}[Z]=\frac{1}{n_o} \sum_{i\in o} Z_i$.

\section{Surrogates without Dynamic Adjustment and Continuous Treatments}\label{ssec:ddsurr_noadj}

For expository purposes, we begin by analyzing the setting where \vsedit{$T_{>1}=0$ almost surely, i.e. treatments occur only in period $1$, but are multi-dimensional and potentially continuous.} We will further assume a \textit{partially linear relationship} between the treatment $T_1$ and the long-term outcome in the experimental sample:
\begin{align}\label{eqn:ass:plr}
    \E_e[\bar{Y}\mid T_1, S_0] =~& \theta_0^\top \phi(T_1, S_0) + b_0(S_0) \tag{PLR}
\end{align}
for some known feature map $\phi(\cdot, \cdot)$, but arbitrary function $b_0(\cdot)$. 

Formally, the \textit{invariance of the surrogate-outcome} relationship requires that that the mean-relationship between the surrogates $S_1$ and the long-term outcome does not change between the observational and the experimental sample:
\begin{align}\label{eqn:ass:ir}
    g_0(S_1) :=~& \E_o[\bar{Y}\mid S_1] = \E_e[\bar{Y} \mid S_1]. \tag{IR}
\end{align}
We denote with $g_0$ the surrogate index model and with $g_0(S_1)$ the surrogate index.

\vsedit{Finally, for the surrogate approach to be valid we need that the long-term outcome $\bar{Y}$ is independent of $S_0, T_1$, conditional on $S_1$. In fact we simply need \emph{conditional mean independence}:
\begin{align}\label{eqn:ass:meanid}
    \bar{Y} \indep_{\mean} (T_1, S_0) \mid S_1 \tag{MeanID}
\end{align}
i.e. $\E[\bar{Y} \mid S_1, T_1, S_0] = \E[\bar{Y}\mid S_1]$. The latter is satisfied under the causal graph assumption of Figure~\ref{fig:simplified_cg}, when $T_{>0}=0$ a.s..}

\vsedit{Under the PLR assumption and the causal graph governing our data, we have by a standard $g$-formula (see e.g. \cite{hernan2010causal}) that:
\begin{equation}
    \tau(t_1, t_0) = \E[\E[\bar{Y}\mid T_1=t_1, S_0] - \E[\bar{Y} \mid T_1=t_0, S_0]] = \theta_0^\top \E[\phi(t_1, S_0) - \phi(t_0, S_0)]
\end{equation}}
To estimate our treatment effect of interest, it suffices to find an estimate $\hat{\theta}$ of the parameter vector $\theta_0$. Subsequently, we can also estimate:
\begin{align}
    \hat{\tau}(t_1, t_0) = \hat{\theta}^\top \E_{e,n}[\phi(t_1, S_0) - \phi(t_0, S_0)]
\end{align}

We establish three valid estimands for $\theta_0$ that follow similar intuition as \citet{athey2020estimating}, adapted to consider linear effects of continuous treatments rather than a single binary treatment. A graphical depiction of the different identification strategies is depicted in Figure~\ref{fig:static_id}
\vsedit{\begin{theorem}[Identification]\label{thm:id}
Denote the residual surrogate index and the residual featurized treatment with:
\begin{align}
    \tilde{g}_0(S_1, S_0) :=~& g_0(S_1)  - \E_e[g_0(S_1)\mid S_0] &
    \Phi_1 :=~& \phi(T_1, S_0) &
    \tilde{\Phi}_1 :=~& \Phi_1 - \E_e[\Phi_1\mid S_0]
\end{align}
Then under assumptions~\eqref{eqn:ass:plr}, \eqref{eqn:ass:ir} and \eqref{eqn:ass:meanid}:
\begin{align*}
    \theta_0 =~& \E_e[\tilde{\Phi}_1\, \tilde{\Phi}_1^\top]^{-1} \E_e[\tilde{\Phi}_1\, \tilde{g}_0(S_1, S_0)] \tag{surrogate index rep.}\\
    \theta_0 =~& \E_e[\tilde{\Phi}_1 \tilde{\Phi}_1^\top]^{-1} 
    \E_o\left[\frac{\Pr(e\mid S_1)}{\Pr(o\mid S_1)} \frac{\Pr(o)}{\Pr(e)} \E_e[\tilde{\Phi}_1\mid S_{1}]\, \bar{Y} \right] \tag{surrogate score rep.}\\
    \theta_0 =~& \E_e[\tilde{\Phi}_1 \tilde{\Phi}_1^\top]^{-1} \bigg(\E_e[\tilde{\Phi}_1\, \tilde{g}_0(S_1, S_0)] + \E_o\bigg[\frac{\Pr(e\mid S_1)}{\Pr(o\mid S_1)} \frac{\Pr(o)}{\Pr(e)}\, \E_e[\tilde{\Phi}_1\mid S_{1}] (\bar{Y} - g_0(S_1))\bigg]\bigg) \tag{orthogonal rep.}
\end{align*}
\end{theorem}
}

\begin{figure}[htpb]
    \centering
    \includegraphics[scale=.5]{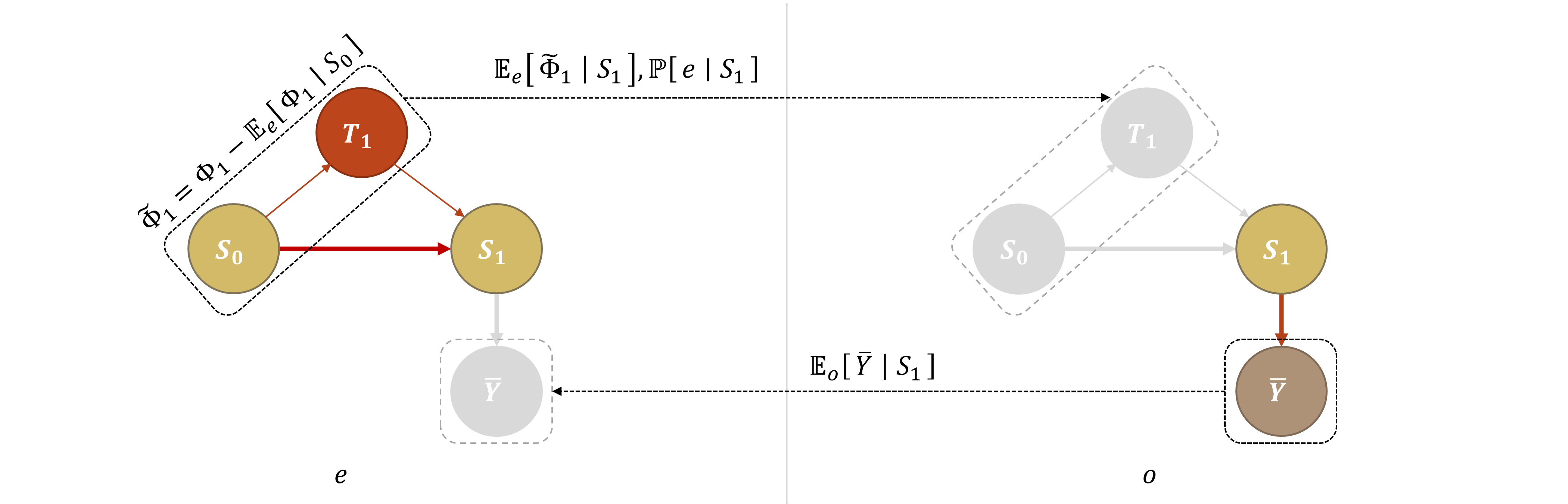}
    \caption{Graphical depiction of surrogate index and surrogate score identification strategies (c.f. Theorem~\ref{thm:id}) for the  case when treatments occur only at period $1$. Grey elements are unobserved and dashed arrows in between the two settings portray the quantities that are being learned in one setting and transferred to the other so as to enable identification of the target quantity.}
    \label{fig:static_id}
\end{figure}

The core estimation challenge that the surrogate approach resolves is that the treatments and outcome of interest are not observed in a single dataset. Intuitively, the first \textbf{surrogate index representation} approaches this challenge by using realized treatments from the experimental sample and, in place of realized outcomes, substitutes the expected outcome conditional on the surrogates, $g_0\left(S\right)$, which can be identified from the observational sample and then constructed in the experimental sample. 

The second \textbf{surrogate score representation} reverses this substitution. The second term pairs an expectation of the featurized treatment conditional on surrogates, $\E_e[\tilde{\Phi}_1\mid S_{1}]$, with the realized outcomes from the observational sample. This representation requires an added ratio of probabilities of appearing in each sample, $\Pr(e\mid S_1)$ and $\Pr(o\mid S_1)$, to adjust for variation of the marginal surrogate distribution across the two datasets.

The third \textbf{orthogonal representation} blends the first two representations and satisfies \textit{Neyman orthogonality}, which allows the construction of confidence intervals and double robustness.\footnote{One practical difficulty with this third doubly robust approach is that it is less transparent and requires access to the raw historical dataset whenever estimating a new treatment option. In contrast, the surrogate index representation allows for segmentation: one can estimate the surrogate index in the observational data once and store only the $g_0$ model. Treatment effects can then be estimated using only these stored parameters and the experimental dataset, or multiple experimental datasets. This explicit construction of expected outcomes in the experimental data also makes the first approach particularly easy to interpret. However, one must then be careful to account for the additional uncertainty stemming from estimating outcomes in the first step, as standard confidence intervals in the second step will not account for this pre-estimated component.}

We note that the parameter identified in the equations in Theorem~\ref{thm:id} is interpretable even if the partially linear assumption is violated. In this case, the equations are identifying the best linear projection of the variation in the long-term outcome that is not explained by the initial state, i.e. $\bar{Y} - \E_e[\bar{Y}\mid S_0]$, on the variation in the feature map $\phi(T_1, S_0)$, that is also un-explained by the initial state, i.e. $\tilde{T}_1$. That is the quantity:
\vsedit{
\begin{equation}
    \theta_0 = \E_e[\tilde{\Phi}_1\tilde{\Phi}_1^\top]^{-1} \E_e[\tilde{\Phi}_1\, \left(\bar{Y} - \E_e[\bar{Y}\mid S_0]\right)]
\end{equation}
}
We formulate the estimation of $\theta_0$ based on the orthogonal representation as a $Z$-estimator based on a vector of moment equations that depends on a vector of nuisance functions $f_0$, i.e.:
\begin{equation}
    m(\theta_0; f_0) := \E[\psi(Z; \theta_0, f_0))] = 0
\end{equation}
and such that it satisfies the Neyman orthogonality condition:
\begin{definition}[Neyman orthogonal moments] A vector of moment conditions $m(\theta; f)$ that depends on a target parameter $\theta$ and a nuisance parameter $f$, whose true values are $\theta_0, f_0$ respectively, satisfies Neyman orthogonality if:
\begin{equation*}
    D_f m(\theta_0; f_0)[f-f_0] := \frac{\partial}{\partial t} m(\theta_0; f_0 + t\, (f-f_0)) \bigg|_{t=0} = 0
\end{equation*}
\end{definition}
Subsequently, this will allow us to invoke the results in \cite{chernozhukov2018double}, to derive an asymptotic normal estimator, even when high-dimensional, regularized approaches are used to estimate the nuisance functions $f_0$.

\vsedit{\begin{theorem}[Orthogonal Moment Formulation]\label{thm:ortho} Denote with:
\begin{align*}
    q_0(S_1) :=~& \frac{\Pr(e\mid S_1)}{\Pr(o\mid S_1)}\, \E_e[\tilde{\Phi}_1\mid S_{1}] &
    p_0(S_0) :=~& \E_e[\Phi_1\mid S_0] &
    h_0(S_0) :=~& \E_e[g(S_1)\mid S_0]
\end{align*}
and let $f=(g, q, p, h)$, denote the full set of nuisance functions. Then $\theta_0$ is the solution to the moment equation:
\begin{align*}
    m(\theta_0; f_0):=~& m_e(\theta_0; f_0) + m_o(\theta_0; f_0) = 0 \\
    m_e(\theta; f) :=~& \E\bigg[1\{e\}\, (g(S_1) - h(S_0) -
    \theta^\top(\Phi_1 - p(S_0)))\cdot (\Phi_1 - p(S_0))\bigg]\\
    m_o(\theta_0; f):=~& \E\left[1\{o\}\, q(S_1) (\bar{Y} - g(S_1))  \right]
\end{align*}
Moreover, the moment $m$ satisfies the Neyman orthogonality property with respect to $f$. Furthermore, it satisfies a stronger double robustness property with respect to $g$ and $q$, i.e. if we denote with $\hat{\theta}$ the solution to $m(\theta; \hat{g}, \hat{q}, p_0, \hat{h})=0$, then:
\begin{align*}
    \left\|\E_e[\tilde{T}_1 \tilde{T}_1^\top] \left(\hat{\theta} - \theta_0\right)\right\|_2 \leq \frac{\Pr(o)}{\Pr(e)}\|g_0 - \hat{g}\|_{2,o}\, \|q_0 - \hat{q}\|_{2,o} 
\end{align*}
\end{theorem}
}

\vsedit{
Given the latter orthogonal moment formulation of the target parameter of interest, one can achieve a root-$n$ asymptotically normal estimate and accompanied asymptotically valid confidence intervals by invoking the results
in \cite{chernozhukov2018double} and verifying that the general conditions required by the main theorems in \cite{chernozhukov2018double} are satisfied. Given that the estimate that we present in the next section (see Theorem~\ref{thm:dyn-estimation}) is a generalization of the setting presented in this section, we omit this result and refer the reader to the more general theorem of the next section.
}
\vsdelete{\begin{theorem}[Estimation and Asymptotic Normality]\label{thm:estimation}
Consider the estimator based on the empirical version of the orthogonal moment, with plug-in nuisance estimates $\hat{f}=(\hat{g}, \hat{q}, \hat{p}, \hat{h})$ trained on a separate sample, i.e. $\hat{\theta}$ solves the empirical moment equation:
\vsedit{\begin{align*}
    m_n(\hat{\theta}; \hat{f}) :=~& m_{e,n}(\hat{\theta}; \hat{f}) + m_{o,n}(\hat{\theta}; \hat{f}) = 0 \\
    m_{e,n}(\theta; f) :=~& \E_n\bigg[1\{e\}\, (g(S_1) - h(S_0) -
    \theta^\top (\Phi_1 - p(S_0))) (\Phi_1 - p(S_0))\bigg]\\
    m_{e,n}(\theta; f) :=~& \E_n\left[1\{o\}\, q(S_1, S_0) (\bar{Y} - g(S_1))  \right]
\end{align*}}
\vscomment{Need to work out the exact requirement from nuisances. Potentially this will also give rise to the "domain adapted MSE rate for $g$.}
If $\|\hat{h}-h_0\|_{2,e}, \|\hat{p}-p_0\|_{2,e}=o_p(n^{-1/4})$ and $\|g_0 - \hat{g}\|_{2,o}\cdot \|q_0 - \hat{q}\|_{2,o} = o_p(n^{-1/2})$, then:
\begin{equation*}
    \sqrt{n} \left(\hat{\theta} - \theta_0\right) \to_{d} N(0, J_e^{-1}\,(\Sigma_e + \Sigma_o)\, J_e^{-1})
\end{equation*}
where:
\begin{align*}
    J_e :=~& \E_e[\tilde{T}_1\, \tilde{T}_1^\top] &
    \Sigma_e :=~& \frac{1}{\Pr(e)} \E_e\left[(\tilde{g}_0(S_1)-\theta_0^\top \tilde{T}_1)^2\, \tilde{T}_1 \tilde{T}_1^\top\right] \\
    & & \Sigma_o :=~& \frac{\Pr(o)}{\Pr(e)^2}\, \E_o\left[(\bar{Y} - g_0(S_1))^2\, q_0(S_1, S_{0})\, q_0(S_1, S_{0})^\top\right]
\end{align*}
\end{theorem}
Asymptotically valid confidence intervals can be constructed using empirical analogues of the co-variance matrix.
}

\section{Surrogates with Dynamic Adjustment: Non-Parametric Identification}\label{ssec:ddsurr_wadj}

In this section we deal with the case where the treatment policy in the observational and the experimental data is dynamic and we want to estimate only the effects of the treatment at period $1$, under zero future treatments, i.e. the part of the effect that does not go through future treatments but solely through the surrogates/control variables.

\paragraph{Preliminary definitions.} To present the identification and estimation strategy we will need to introduce some notation from the dynamic treatment regime literature. Consider an arbitrary time-series process $\{S_{t-1}, T_t, Y_t\}_{t=1}^{M}$, with $S_t\in \mcS_t$ and $T_t\in \mcT_t$. For any time $t$, let $\bar{S}_t=\{S_1,\ldots, S_t\}$ and $\bar{T}_t=\{T_1,\ldots, T_t\}$ denote the sequence of the variables up until time $t$ and similarly, let $\underline{S}_t = \{S_t, \ldots, S_M\}$ and $\underline{T}_t=\{T_t,\ldots, T_M\}$. We will also denote with $\bar{s}_t, \bar{\tau}_t, \bar{y}_t, \underline{s}_t, \underline{\tau}_t, \underline{y}_t$, corresponding realizations of the latter random sequences. Moreover, we will be denoting with $(\bar{\tau}'_t, \underline{\tau}_{t+1})$, the sequences of treatments that follows $\tau'$ up until time $t$ and then continues with $\tau$. We let $0\in \mcT_t$ denote a baseline policy value, which could be appropriately instantiated based on the context.

\paragraph{Target quantity.} For any sequence of treatment $\tau=(\tau_1,\ldots, \tau_M)$, let $Y_{t}^{(\tau)}$ denote the counterfactual outcome at period $t$ under such a sequence of interventions, equivalently in do-calculus notation $Y_{t} \mid do(\bar{T}_M=\bar{\tau}_M)$. Note that $Y_t^{(\tau)}$ is only a function of $\bar{\tau}_t$, i.e. $Y_t^{(\tau)} \equiv Y_t^{(\bar{\tau}_t)}$, but for simplicity of notation we use the overall vector of treatments. We will also denote with $\bar{Y}_t^{(\tau)}:=\sum_{j=t}^M Y_j^{(\tau)}$, the counterfactual cumulative outcome from period $t$ and onwards, and with $\bar{Y}^{(\tau)} = \sum_{j=1}^M Y_j^{(\tau)}$, the total counterfactual cumulative outcome. 
Under this counterfactual notation, we can re-write our target quantity of interest from Equation~\eqref{eqn:target-do} as:
\begin{align}\label{eqn:target-cntf}
    \tau(t_1, t_0) := \E_e\left[\bar{Y}^{(t_1, \underline{0})} - \bar{Y}^{(t_0, \underline{0})}\right]
\end{align}

\subsection{Non-Parametric Identification}\label{ssec:nonparamid}

We show that the target quantity of interest is non-parametrically identified if the data generating processes adhere to the causal graph depicted in Figure~\ref{fig:dag} and satisfy a regularity condition on overlap, as well a a dynamic analogue of an invariance relationship between the observational and experimental setting. 
We first present a set of high-level conditions that lead to non-parametric identification and then present the main identification result. 

\begin{figure}[htpb]
    \centering
    \begin{subfigure}{.45\textwidth}
    \includegraphics[scale=.4]{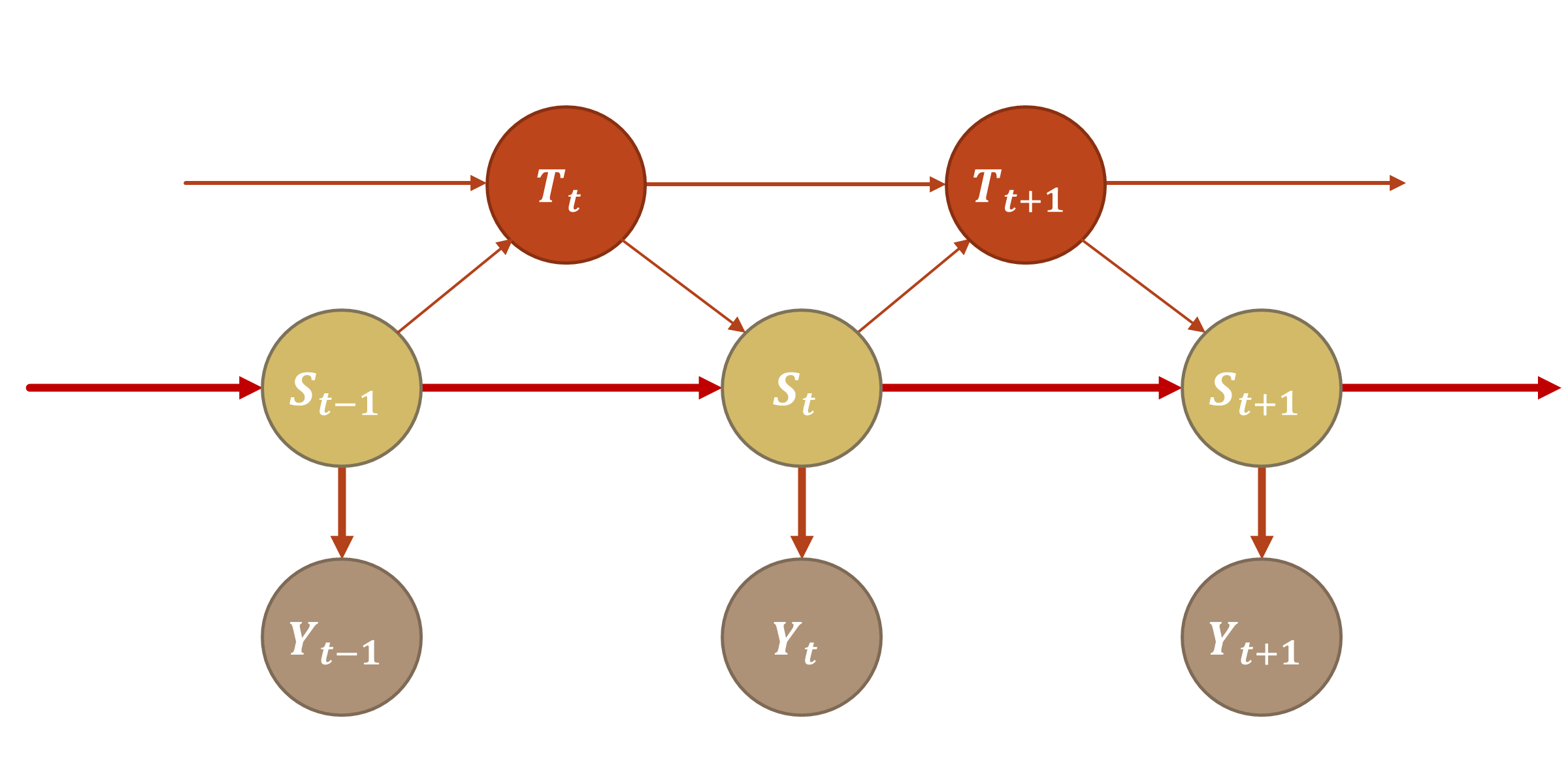}
    \caption{Causal diagram describing the causal relationships of the random variables in both settings $\{e,o\}$.}
    \label{fig:dag}
    \end{subfigure}
    ~~~~~~~~~~
    \begin{subfigure}{.45\textwidth}
    \includegraphics[scale=.4]{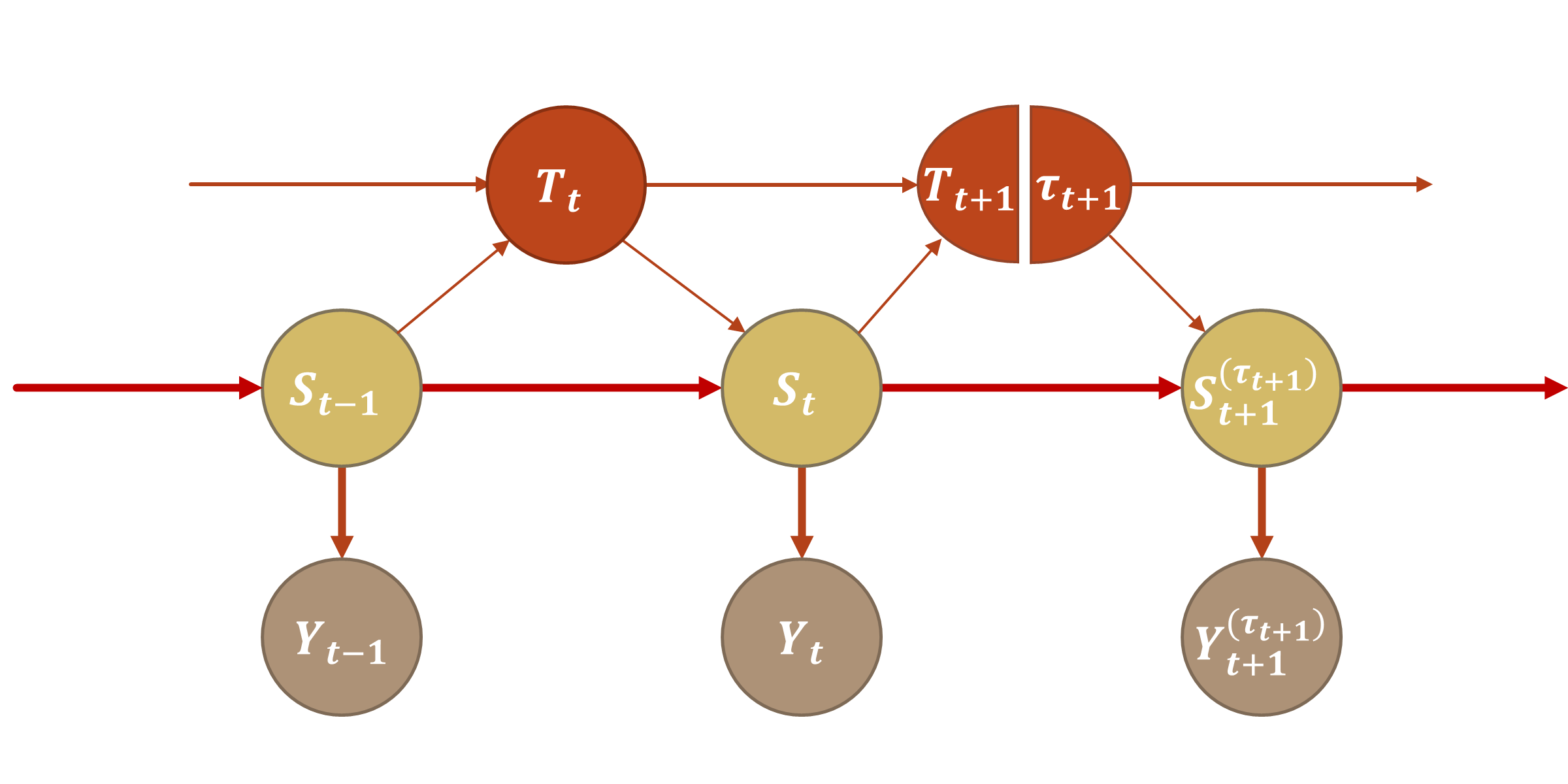}
    \caption{Single world intervention diagram from intervening and setting the treatments at periods $t+1$ and on-wards to $\underline{\tau}_{t+1}$.}
    \label{fig:swig}
    \end{subfigure}
    \caption{Causal graph and single world intervention graph that govern the data generating processes in both the experimental (e) and observational (o) setting.}
\end{figure}

We assume that the data generating process satisfies the following sequential conditional randomization assumption:
\begin{assumption}[\emph{Sequential Conditional Exogeneity}]\label{ass:cond-ex} The data generating process satisfies the following conditional independence conditions:
\begin{equation}
    \forall 1\leq t\leq j \leq M \text{ and } \forall \underline{\tau}_{t+1}\in \times_{k=t+1}^M \mcT_k: \left\{Y_{j}^{(\bar{T}_{t}, \underline{\tau}_{t+1})}, \underline{S}_t^{(\bar{T}_{t},\underline{\tau}_{t+1})}\right\} \indep T_{t+1} \mid S_{t} \tag{dynExog}
\end{equation}
\end{assumption}
This condition is for instance satisfied if the data generating process adheres to the causal graph presented in Figure~\ref{fig:dag}, as can be easily verified from the single-world-intervention graph (SWIG) in Figure~\ref{fig:swig}. Moreover, we will assume a surrogacy assumption, that under a zero future treatment policy, the effect of $T_t$ on future outcomes only goes through $S_t$. This is again satisfied if the data generating process adheres to the causal graph presented in Figure~\ref{fig:dag}, as can be easily verified from the single-world-intervention graph (SWIG) in Figure~\ref{fig:swig}. In fact, we will only require a conditional mean-independency assumption.
\begin{assumption}[\emph{Sequential Surrogacy}]\label{ass:dyn-surrogacy}
The data generating process satisfies the following conditional mean-independence conditions:
\begin{equation}
    \forall 1\leq t\leq j \leq M \text{ and } \forall \underline{\tau}_{t+1}\in \times_{k=t+1}^M \mcT_k: \left\{Y_{j}^{(\bar{T}_{t}, \underline{\tau}_{t+1})}, \underline{S}_t^{(\bar{T}_{t}, \underline{\tau}_{t+1})}\right\} \indep_{\mean} (T_{t}, S_{t-1}) \mid S_{t} \tag{dynSurr}
\end{equation}
\end{assumption}

Since we do not observe long-term outcomes from the experimental setting, we will need to assume a dynamically adjusted analogue of the invariance property, so that we can use long-term outcomes from the observational dataset to ``impute'' long-term outcomes in the experimental dataset.
\begin{assumption}[Dynamic Invariance]\label{ass:dynIR}
The two settings $\{e,o\}$ satisfy the following invariance property:
\begin{align}
    \E_e\left[\bar{Y}^{(\underline{0}_2)}\mid S_1\right] = \E_o\left[\bar{Y}^{(\underline{0}_2)}\mid S_1\right] \tag{dynIR}
\end{align}
\end{assumption}
Observe that the dynamic invariance Assumption~\eqref{ass:dynIR} is much more permissive in practice than the standard invariance assumption as we no longer require that the dynamic treatment policy in the observational data be the same as in the experimental data, but simply that the adjusted outcomes under baseline treatment levels retain the same relationship with the surrogates. Moreover, for conveniency to reader's more familiar with do-calculus notation, we can equivalently express this assumption as:
\begin{align*}
\E_e\left[\bar{Y}\mid do(T_{>1}=0), S_1\right] = \E_o\left[\bar{Y}\mid do(T_{>1}=0), S_1\right].
\end{align*}

\begin{figure}[htpb]
    \centering
    \includegraphics[scale=.47]{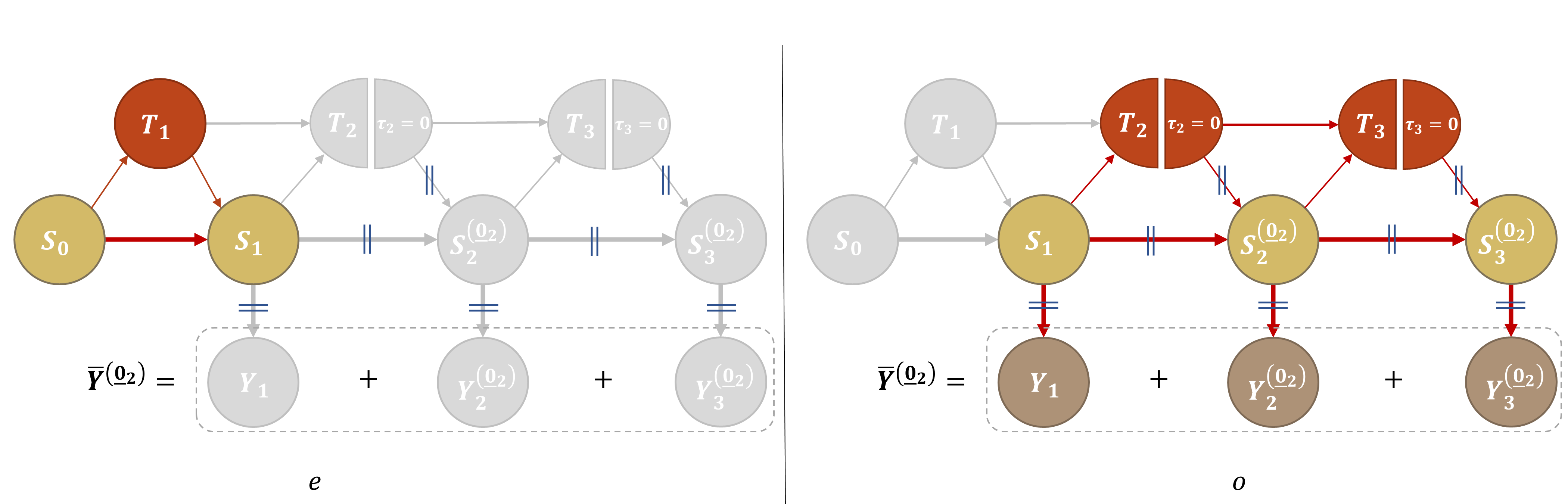}
    \label{fig:dynir}
    \caption{A sufficient condition for the dynamic invariance property to hold is that the relationships designated with a double line in the SWIGs of settings $\{e,o\}$, under the intervention $\text{do}(\underline{T}_2=\underline{0}_2)$, remain invariant (at least in terms of means) between the two settings. These relationships govern how the outcome evolves in the absence of treatments and do not depend on the treatment policy that governs the two settings (which is prescribed by the arrows that are not designated with a double line).}
\end{figure}

Finally, we also require a regularity condition of sequential positivity (aka overlap), which essentially states that the density of treatment is bounded away from zero a.s.. To define sequential positivity, we will denote with $\pi_d(\tau_t, s_{t-1})$ the marginal densities of the random variables $(T_t, S_{t-1})$, for any setting $d\in \{e, o\}$ and period $t\in [1,M]$. Then sequential positivity is defined as:
\begin{assumption}[Sequential Positivity]\label{ass:dynPos}
The densities $\pi_e, \pi_o$ of the data generating processes satisfy that: $\pi_e(t_1, S_0), \pi_e(t_0, S_0)>0$ almost surely and that for all $t \in [2, M]$: $\pi_o(0, S_t) >0$ almost surely.
\end{assumption}

Under these high-level assumptions, we can show that the target outcome of interest is non-parametrically identified using a variant of the $g$-formula, based on a recursively defined estimand.
\begin{theorem}[Non-Parametric Identification]\label{thm:non-param-id}
Suppose that the data generating processes adhere to the causal graph depicted in Figure~\ref{fig:simplified_cg}. Then sequential conditional exogeneity Assumption~\ref{ass:cond-ex} and sequential surrogacy Assumption~\ref{ass:dyn-surrogacy} are satisfied.

If the data generating processes satisfy Assumption~\ref{ass:cond-ex}, Assumption~\ref{ass:dyn-surrogacy}, Assumption~\ref{ass:dynIR} and Assumption~\ref{ass:dynPos}, then the target quantity $\tau(t_1, t_0)$ is non-parametrically identified via the following recursively defined estimands:
\begin{align}
    \tau(t_1, t_0) =~& \E_e\left[\E_e\left[g_{\adj}^*(S_1)\mid S_0, T_1 = t_1\right] - \E_e\left[g_{\adj}^*(S_1)\mid S_0, T_1 = t_0\right]\right] \nonumber\\
    g_{\adj}^*(S_1) :=~& \E_o\left[\bar{Y}^{(\underline{0}_2)} \mid S_1\right] = \E_o\left[Y_1\mid S_1\right] + \sum_{j=2}^M f_{2,j}(S_1) \tag{adjusted surrogate index}\\
    \forall j\in [2, M]: f_{j, j}(S_{j-1}) =~& \E_o\left[Y_j\mid S_{j-1}, T_j=0\right] \tag{base estimand}\\
    \forall 2\leq t < j \leq M: f_{t, j}(S_{t-1}) =~& \E_o\left[f_{t+1,j}(S_t)\mid S_{t-1}, T_t=0\right] \tag{recursive estimand}
\end{align}
\end{theorem}

The non-parametric identification argument of Theorem~\ref{thm:non-param-id} requires the estimation of quantities of the form $\E[f_{t+1,j}(S_t)\mid S_{t-1}, T_t=0]$. When treatment $T$ is binary or discrete, then such quantities can be estimated in a relatively accurate manner by fitting nested regression models on the sub-population for which $T_t=0$. Moreover, we can also employ the great variety of doubly robust estimators for the quantity $\E_o[\bar{Y}^{(\underline{0}_2)}\mid S_1]$ (see e.g. \cite{tran2019double}) combined with a doubly robust estimator for the surrogate part, to arrive at an overall doubly robust estimator. For instance, we can adapt the efficient influence function (EIF) representation of dynamic treatment effects proposed in \cite{scharfstein1999adjusting,van2011targeted,robins2000comment}, to the case of a surrogate index setting as follows: 
\begin{theorem}[Double Robust Representation for Discrete Treatments]
Suppose that treatments are supported on a discrete set of values. For $2\leq t \leq m \leq M$, let $f_{t, m}(S_{t-1})$ as defined in Theorem~\ref{thm:non-param-id}, and let:
\begin{align*}
f_{1,m}(T_1, S_0) :=~& \E_e\left[f_{2,m}(S_1)\mid S_0, T_1\right], & 
f_{m+1,m}(S_m) :=~& Y_m
\end{align*}
Moreover, consider the sequential inverse propensity weights defined as follows, for any $2\leq t\leq M$:
\begin{align*}
W_0(\tau) :=~& \frac{1\{T_1=\tau\}}{\Pr_e(T_1=\tau\mid S_0)}, & 
W_1(\tau) :=~& \frac{\Pr(S_1\mid e)}{\Pr(S_1\mid o)} \E_e\left[W_0(\tau)\mid S_1\right],\\
W_2 :=~& \frac{1\{T_2=0\}}{\Pr_o(T_2=0\mid S_{1})}, &
W_t :=~& \frac{1\{T_t=0\}}{\Pr_o(T_t=0\mid S_{t-1}, T_{t-1}=0)}, &
\bar{W}_t :=~& W_1(\tau) \prod_{j=2}^t W_j
\end{align*}
Then for any $\tau\in \mcT_1$:
\begin{align}\label{eqn:dr-dyn-non-param}
    \E_e\left[Y_m^{(\tau, \underline{0}_2)}\right] = \E_e\left[f_{1,m}(\tau, S_0) + W_0(\tau)\, (f_{2,m}(S_1) - f_{1,m}(\tau, S_0))\right] + \E_o\left[\sum_{t=2}^m \bar{W}_t\, \left(f_{t+1,m}(S_{t}) - f_{t,m}(S_{t-1})\right)\right]
\end{align}
Moreover, the above equation holds if either i) for all $t\in [m]$, $f_{t,m}$ take their correct values, or ii) for all $t\in [m]$, $W_{t}$ take their correct values.
\end{theorem}

Since, we have that $\tau(t_1,t_0)=\sum_{m=1}^M \E_e\left[Y_m^{(t_1, \underline{0})}\right] - \E_e\left[Y_m^{(t_0, \underline{0})}\right]$, we can combine the doubly robust representations prescribed by Equation~\eqref{eqn:dr-dyn-non-param}, for each $m\in [1, M]$ and $\tau\in \{t_1, t_0\}$, to get an overall doubly robust representation of the target quantity. Similar to existing augmented inverse propensity methods in the dynamic treatment regime, the latter representation will lead to a consistent estimation if either all the models that go into the inverse propensity weights $\{W_{t}\}_{t=0}^M$ are consistent, or if all the nested regression functions $\{f_{t,j}\}_{1\leq t\leq m\leq M}$ are consistent. Moreover, this variant of the double robustness property also implies Neyman orthogonality (local robustness) of the moment implicitly defined by Equation~\eqref{eqn:dr-dyn-non-param}. Thus using the general results in \cite{chernozhukov2018double}, we can devise an estimation strategy that enables valid inference while using machine learning, adaptive and regularized estimators for the auxiliary regression and classification models required by the above identification strategy. One could also adapt and apply alternative adaptive estimation frameworks, that also allow for the use of machine learning, adaptive estimators for the auxiliary models, such as the longitudinal targeted minimum loss estimation approach \cite{rotnitzky2012improved,van2011targeted}, based on the latter representation of the target quantity.\footnote{We omit these details for succinctness and since the main estimation algorithm we propose in this work, which applies to both discrete and continuous treatments, appears in Section~\ref{sec:dsurr_semipar} under a semi-parametric assumption.}

However, when treatment $T$ is continuous and potentially multi-dimensional, then non-parametric estimation rates for the quantities described in Theorem~\ref{thm:non-param-id} are required, without further assumptions, and can be potentially very slow and prohibitive. Moreover, finite sample performance will heavily depend on the number of samples observed in a region around the baseline treatment level at each period $T_t=0$, which could be very small and impact statistical power. Since our main application of interest (return-on-investments) involves multiple continuous treatments, being able to handle this setting is of primary practical importance.

\section{Surrogates with Dynamic Adjustment: Semi-Parametric Identification and Inference}\label{sec:dsurr_semipar}

To achieve parametric estimation rates, with valid confidence intervals, and more stable finite sample performance for the quantities of interest, even in the case of multiple continuous treatments, we will make further semi-parametric assumptions on the data-generating processes, i.e. that some parts of the data-generating process adhere to a known parametric form. One option for instance, would be to assume that the regression functions $\E[f_{t+1,j}(S_t)\mid S_{t-1}, T_t]$ adhere to some known parametric form, e.g. $\theta^\top\phi(T_t, S_{t-1})$, for a known feature map $\phi$. However, this essentially assumes a fully parametric model: even in the absence of any treatment, the world behaves in a simple manner. Unlike, for instance, in the classic partially linear model, where the baseline behavior under no-treatment is left non-parametric and only the effect of the treatment on the baseline behavior is modeled in a parametric manner. Instead, we could only model how these regression functions behave as the treatment $T_t$ deviates from the baseline, i.e. 
\begin{align*}
\E[f_{t+1,j}(S_t)\mid S_{t-1}, T_t] - \E[f_{t+1,j}(S_t)\mid S_{t-1}, T_t=0] = \theta^\top \phi(T_t, S_{t-1}).
\end{align*}
Hence, analogous to the partially linear model, we are leaving un-modeled, the baseline behavior at each period, conditional on the past (the nested conditional mean). This is exactly the approach taken in the line of work on \emph{structural nested mean models (SNMMs)}, which we explore in the subsequent sections. As it will be shown below, the structural parameters $\theta$ of these nested means, can be identified without the need to estimate local non-parametric regression quantities of the form $\E[f_{t+1,j}(S_t)\mid S_{t-1}, T_t=0]$ and hence wont suffer from low sample sizes near the baseline treatment. Moreover, the target quantity of interest can be expressed in terms of these structural parameters $\theta$ of the SNMM.

\subsection{A Primer on Structural Nested Mean Models (SNMMs)}\label{ssec:ddsurr_wadj:semi_id}

The aforementioned semi-parametric assumption, can be expressed in terms of primitive counterfactual quantities, using the notion of a \emph{blip} function. 
\begin{definition}[Blip Functions] For any $\forall 1\leq t\leq j\leq M$ and any $\tau_t\in \mcT_t$, we define the blip function as:
\begin{align}
     \gamma_{t,j}(\tau_t, s_{t-1}) := \E\left[Y_j^{(\bar{T}_{t-1}, \tau_t, \underline{0}_{t+1})} - Y_j^{(\bar{T}_{t-1}, \underline{0}_{t})} \mid T_{t}=\tau_{t}, S_{t-1}=s_{t-1}\right],
\end{align}
which corresponds to the mean change in outcome $Y_j$, if we go to all units which received treatment $\tau_t$ at time $t$ and had observed surrogate state ${s}_{t-1}$ and we remove their last treatment, while subsequently continue with a zero treatment.
\end{definition}
These functions are a variant of what are known as the \emph{blip} functions \citep{Chakraborty2013,Robins2004} and can be shown to be non-parametrically identifiable, assuming sequential conditional exogeneity and a sequential analogue of the positivity (aka overlap) assumption \citep{Robins2004}. Theorem~3.1 of \citep{Robins2004} combines a telescoping sum argument and the sequential randomization condition to express counterfactual outcomes in terms of blip functions. We restate this result here, adapting it to our notation and our variant of sequential conditional exogeneity and blip function definition and providing a proof for completeness:
\begin{lemma}[Identification of Counterfactual Outcomes via Blip Functions]\label{lem:cntf-char-rho} For any treatment sequence $\tau$ and under the sequential conditional exogeneity assumption, the following identity holds about the counterfactual outcomes:
\begin{align}
    \forall 1\leq t\leq j\leq M: \E\left[Y_j^{(\bar{T}_{t-1}, \underline{0}_t)}\mid {S}_{t-1}, {T}_t={\tau}_t\right] = \E\left[Y_j - \sum_{q=t}^{j} \gamma_{q, j}({T}_q, {S}_{q-1}) \mid {S}_{t-1}, {T}_t={\tau}_t\right]\label{eqn:cntf-rho}
\end{align}
\end{lemma}
Intuitively, each term $\gamma_j$, removes from the outcome the \emph{blip effect} of the observed action $T_j$. Consider any target outcome $Y_j$. When we remove $\gamma(T_j, S_{j-1})$ from $Y_j$, then what remains is, in-expectation (and crucially, even conditional on $S_{j-1}, T_{j}$), equal to the counterfactual outcome, where the sample received zero-treatment at period $j$. Subsequently, removing $\gamma(T_{j-1}, S_{j-2})$ from this remnant, then what remains is in-expectation (and crucially, even conditional on $S_{j-2}, T_{j-1}$), equal to the counterfactual outcome, where the sample received zero-treatment at periods $\{j-1,j\}$, and so on and so forth.

Note that if we denote with $\gamma_{o, t,j}$ the blip functions of the observational setting, then the latter lemma immediately gives an alternative identification strategy to the one presented in Theorem~\ref{thm:non-param-id}, since we can write:
\begin{align*}
    g_{\adj}^*(S_1) := \E_o\left[\bar{Y}^{(\underline{0}_2)} \mid S_1\right]= \sum_{j=1}^M \E_o\left[Y_j - \sum_{q=t}^{j} \gamma_{o, q, j}({T}_q, {S}_{q-1}) \mid {S}_{1}\right]
\end{align*}
Thus if we can identify the blip functions, then the target quantity is also immediately identified without further assumptions.

One strategy for identifying the blip functions is to assume that they obey some known parametric form and then identify the parameters via a set of moment restrictions that the blip functions need to satisfy. In particular, by Lemma~\ref{lem:cntf-char-rho}, we know that the quantity $H_{t,j}(\theta^*)$ is equal in-expectation, and conditional on $S_{t-1}, T_t$ to the counterfactual outcome $Y_j^{(\bar{T}_{t-1}, \underline{0}_t)}$. However, this counterfactual outcome, by the sequential conditional exogeneity implied by the causal graph assumption, is independent of the treatment $T_t$, conditional on $S_{t-1}$, i.e. $Y_j^{(\bar{T}_{t-1}, \underline{0}_t)}\indep T_t\mid S_{t-1}$. Thus for any function $f$ of $T_{t}, S_{t-1}$:
\begin{align*}
    \E\left[Y_j^{(\bar{T}_{t-1}, \underline{0}_t)} f(T_t, S_{t-1})\mid S_{t-1}\right] = \E\left[Y_j^{(\bar{T}_{t-1}, \underline{0}_t)}\mid S_{t-1}\right] \E\left[f(T_t, S_{t-1})\mid S_{t-1}\right]
\end{align*}
Moreover, by the conditional mean equivalence of this counterfactual outcome and the ``remnant of the blip effects'' $H_{t,j}(\theta^*)$, the same conditional mean independence property needs to hold for $H_{t,j}(\theta^*)$.
\begin{align*}
    \E\left[H_{t,j}(\theta^*) f(T_t, S_{T-1})\mid S_{t-1}\right] = \E\left[H_{t,j}(\theta^*)\mid S_{t-1}\right] \E\left[f(T_t, S_{T-1})\mid S_{t-1}\right]
\end{align*} 
This leads to the following lemma:
\begin{lemma}[Moment Restrictions for Blip Functions]\label{lem:moment-restrictions} For any parameterization of the blip functions $\gamma_{t,j}({\tau}_t, {s}_t; \theta_{t,j})$, and for any $j\geq t$, if we let the random variable $H_{t,j}(\theta):= Y_j - \sum_{q=t}^{j} \gamma_{q,j}(T_q, {S}_{q-1}; \theta_{q,j})$, then the true parameter vector $\theta^*$ must satisfy the moment restrictions:
\begin{align}\label{eqn:snmm-moments}
    \forall 1\leq t\leq j, \forall f \in \mcF: \E\left[ H_{t,j}(\theta^*)\, \left(f(T_t, {S}_{t-1}) - \E[f(T_t, {S}_{t-1})\mid {S}_{t-1}]\right) \mid S_{t-1}\right] = 0
\end{align}
where $\mcF$ contains all functions mapping lag surrogates ${s}_{t-1}$ and current period treatments ${\tau}_t$ to $\R$.
\end{lemma}
Hence, if we have found the right $\theta^*$, then the infinite set of conditional moment restrictions in Equation~\eqref{eqn:snmm-moments} need to be satisfied. Lemma~\ref{lem:moment-restrictions} is an adaptation of Theorem~3.2 of \citep{Robins2004} to our notation and we include its proof for completeness. Methods that estimate the structural parameters by utilizing such conditional mean independence moment restrictions are typically referred to in the literature on dynamic treatment effects as $g$-estimation methods.\footnote{$g$-estimation is a different term than $g$-computation, which typically refers to using the $g$-formula for dynamic treatment effects and estimating effects in a plug-in manner by estimation all conditional densities, and conditional means.}

One approach to operationalize Lemma~\ref{lem:moment-restrictions} would be to perform a grid search over some discretization of the parameter space and check that this set of conditional moment restrictions holds. In the full generality of structural nested mean models, without any further assumptions on the blip functions, such an exhaustive grid search could be inevitable, and renders the method impractical from a computational perspective.

For this reason, a typical approach in structural nested mean models, to render the methodology practical, is to assume a linear parametric form for the blip functions, leading to the class of linear structural nested mean models.\footnote{We note that the literature on $g$-estimation has also analyzed other forms of generalized linear parametric forms and provided practical methods (see e.g. \cite{Robins2004,Chakraborty2013,vansteelandt2014structural}).}
\begin{assumption}[Linear SNMM]\label{ass:linear-snmm}
The blip functions admit a linear parametric form:\footnote{We could also allow more flexibility and allow the feature map $\phi_{t}$ to also depend on the target period $j$, i.e. $\phi_{t,j}$, but it makes exposition more cumbersome.}
\begin{align}
    \gamma_{t,j}({\tau}_t, {s}_{t-1};\theta_{t,j}) := \theta_{t,j}^\top\phi_{t}(\tau_t, {s}_{t-1})  
\end{align}
for a known $\fdim$-dimensional feature vector maps $\phi_{t}$, satisfying $\phi_{t}(0, {s}_{t-1})=0$, such that for some $\theta_{t,j}^*$, $\gamma_{t,j}(\cdot, \cdot;\theta_{t,j}^*)=\gamma_{t,j}(\cdot,\cdot)$.
\end{assumption}
Assuming that the expected conditional covariance matrix $\E[\Cov(\phi_t(T_t, S_{t-1})\mid S_{t-1})]$ of the feature map $\phi_t(T_t, S_{t-1})$ conditional on $S_{t-1}$, is full rank, then we can uniquely identify $\theta^*$ by finding a parameter vector $\theta$ that satisfies a small subset of the moment restrictions of the form:
\begin{align}\label{eqn:dyn-params-moments}
    \forall 1\leq t\leq j \leq M: \E[ H_{t,j}(\theta)\, \left(\phi_{t}({T}_t, {S}_{t-1}) - \E[\phi_{t}({T}_t, {S}_{t-1})\mid {S}_{t-1}]\right)] = 0
\end{align}
What is most appealing about linear SNMMs is that the latter system of moment equations has a recursive closed form solution. In particular, we can express parameter $\theta_{t,j}$ as a function of parameters $\theta_{\tau,j}$ for $\tau>t$, in a closed form manner:
\begin{align*}
    \theta_{t,j} = \E\left[\Cov(\phi_t(T_t, S_{t-1})\mid S_{t-1})\right]^{-1} \E\left[\phi(T_t, S_{t-1})\,\left(Y_j - \sum_{\tau=t+1}^j \theta_{\tau,j}^\top \phi(T_\tau, S_{\tau-1})\right)\right]
\end{align*}
This immediately portrays the practicality of the method and the sufficiency of this subset of moment restrictions.

\subsection{Semi-parametric identification of Long-Term Effects via Dynamically Adjusted Surrogates}

We will assume that both the data generating processes that generated the observational dataset and the experimental dataset obey a SNNM model with linear blip functions. Albeit, we allow both the treatments to change in between the two environments, as well as the blip function parameterizations to be different. We will denote with $\gamma_{e,t,j}, \gamma_{o,t,j}$, the blip functions in the two settings, with $\theta_{e,t,j},\theta_{o,t,j}$ the structural parameters of the blip functions in the two settings and with $\phi_{e,t},\phi_{o,t}$ the corresponding feature maps. 

We start by presenting an identification argument for the target quantity of interest, as a function of the blip functions $\gamma_{o,t,j}$ in the observational dataset. Subsequently, in Theorem~\ref{thm:dyn-ortho}, we combine it with a separate identification argument for the structural parameters $\theta_{o,t,j}$ of the blip functions, to arrive at a complete identification strategy.

\begin{theorem}[Semi-parametric Identification in Dynamic Regime]\label{thm:dyn-id}
Suppose that the data generating processes in the experimental and observational setting adhere to the causal graph in Figure~\ref{fig:simplified_cg} and that the blip functions in the experimental setting satisfy Assumption~\ref{ass:linear-snmm}. Moreover, suppose that the two settings $\{e,o\}$ satisfy Assumption~\ref{ass:dynIR}.
Consider the dynamically adjusted outcomes and the dynamically adjusted surrogate index, from the observational setting:
\begin{align}\label{eqn:defns-main-thm}
    Y_j^{o,\adj} :=~& Y_j - \sum_{t=2}^{j} \gamma_{o, t, j}({T}_t, {S}_{t-1}) & \bar{Y}^{o,\adj} :=~& \sum_{j=1}^M Y_j^{o,\adj} & g_{\adj}^*(S_1) :=~& \E_o[\bar{Y}^{o,\adj}\mid S_1]
\end{align}
Denote the residual adjusted surrogate index and the residual feature map with:
\begin{align}
    \tilde{g}_{\adj}^*(S_1, S_0) :=~& g_{\adj}^*(S_1)  - \E_e[g_{\adj}^*(S_1)\mid S_0] &
    \Phi_{1} :=~& \phi_{e,1}(T_1, S_0) &
    \tilde{\Phi}_1 :=~& \Phi_1 - \E_e[\Phi_1\mid S_0]
\end{align}
and assume that $\E_e\left[\tilde{\Phi}_1\, \tilde{\Phi}_1^\top\right]$ is full rank. Then the target quantity of interest can be expressed as:
\begin{align}\label{eqn:dyn-target-param}
\tau(t_1, t_0) = \theta_0^\top\, \E[\phi_{e,1}(t_1, S_0) - \phi_{e,1}(t_0,S_0)],
\end{align}
where parameter $\theta_0$ can be represented by any of the following estimands:
\begin{align*}
    \theta_0 =~& \E_e[\tilde{\Phi}_1\, \tilde{\Phi}_1^\top]^{-1} \E_e[\tilde{\Phi}_1\, \tilde{g}_{\adj}^*(S_1, S_0)] \tag{surrogate index rep.}\\
    \theta_0 =~& \E_e[\tilde{\Phi}_1 \tilde{\Phi}_1^\top]^{-1} 
    \E_o\left[\frac{\Pr(e\mid S_1)}{\Pr(o\mid S_1)} \frac{\Pr(o)}{\Pr(e)} \E_e[\tilde{\Phi}_1\mid S_{1}]\, \bar{Y}^{o,\adj} \right] \tag{surrogate score rep.}\\
    \theta_0 =~& \E_e[\tilde{\Phi}_1 \tilde{\Phi}_1^\top]^{-1} \bigg(\E_e[\tilde{\Phi}_1\, \tilde{g}_{\adj}^*(S_1,S_0)] + \E_o\bigg[\frac{\Pr(e\mid S_1)}{\Pr(o\mid S_1)} \frac{\Pr(o)}{\Pr(e)}\, \E_e[\tilde{\Phi}_1\mid S_{1}] (\bar{Y}^{o,\adj} - g_{\adj}^*(S_1))\bigg]\bigg) \tag{orthogonal rep.}
\end{align*}
\end{theorem}

\begin{figure}[htpb]
    \centering
    \includegraphics[scale=.47]{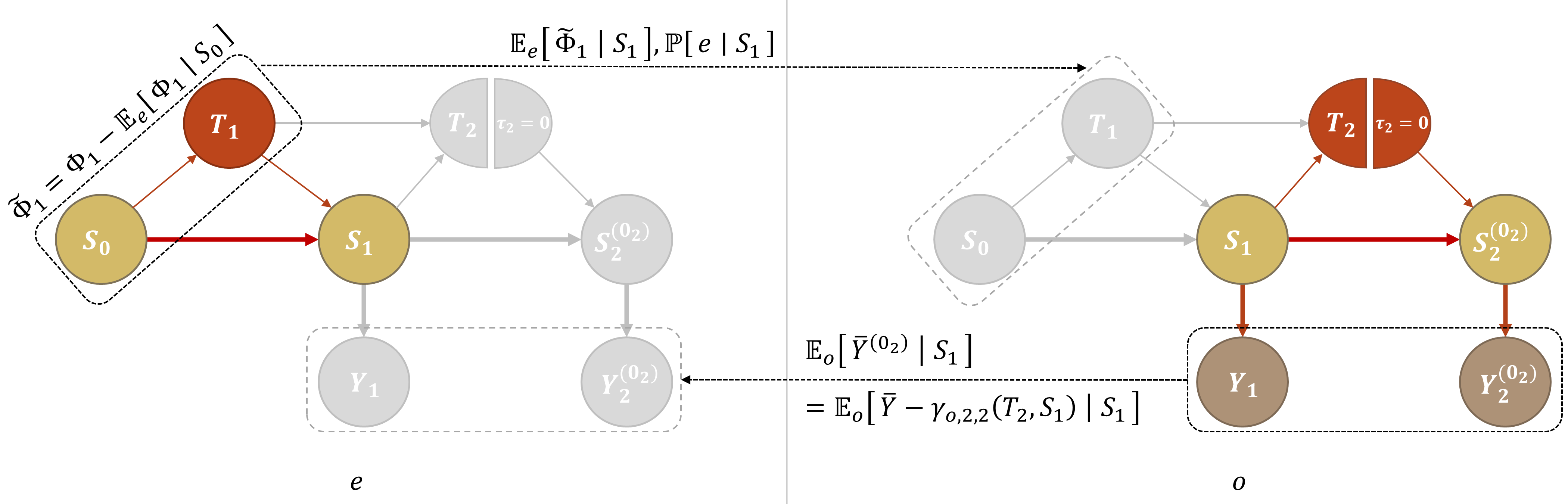}
    \caption{Identification strategies via surrogates in dynamic treatment regime (c.f. Theorem~\ref{thm:dyn-id}) for the special case of two periods. Grey elements are unobserved and dashed arrows in between the two settings portray the quantities that are being learned in one setting and transferred to the other so as to enable identification of the target quantity.}
    \label{fig:dyn_id}
\end{figure}

\subsection{Semi-Parametric Orthogonal Moment Restrictions}\label{ssec:ddsurr_wadj:moments}
One caveat of Theorem~\ref{thm:dyn-id} is that $\bar{Y}^{o,\adj}$ and $g_{\adj}^*$ are defined in terms of the dynamic effects $\{\theta_{o,t,j}\}_{2\leq t\leq j\leq M}$ of the observational setting, which are parameters that also need to be estimated. In particular, if we denote with $\Phi_{o,t} := \phi_{o,t}(T_t, S_{t-1})$, then we can write:
\begin{align*}
    \theta_{o,t} :=~&\sum_{j=t}^M \theta_{o,t,j} & 
    \bar{Y}^{o,\adj} =~&  \bar{Y} - \sum_{t=2}^M \theta_{o, t}^\top \Phi_{o,t} &
    g_{\adj}^*(S_1) =~& \E_o\left[\bar{Y}\mid S_1\right] - \sum_{t=2}^M \theta_{o, t}^\top \E\left[\Phi_{o,t}\mid S_1\right]
\end{align*}
However, we can combine the Neyman orthogonal moment equations developed in \cite{lewis2020doubledebiased} (which are an orthogonal variant of the moment equations in Equation~\eqref{eqn:dyn-params-moments} and a variant of the doubly robust version of this equation introduces by \cite{Robins2004}), with the Neyman orthogonal moment equation from Theorem~\ref{thm:dyn-ortho} to arrive at an overall Neyman orthogonal strategy for simultaneously identifying $\theta_0$ and these auxiliary dynamic effects. In particular, the parameters $\theta_{o,t}$ are identified recursively by the moment restrictions:
\begin{align}
    \E_o\left[ \left(\bar{Y}_t - \E_o[\bar{Y}_t\mid S_{t-1}] - \sum_{\tau=t}^M \theta_{o,\tau}^{\top}\, \left(\Phi_{o,\tau} - \E_o\left[\Phi_{o,\tau} \mid S_{t-1}\right]\right)\right)\, \left(\Phi_{o,t} - \E_o\left[\Phi_{o,t} \mid S_{t-1}\right]\right) \right]
\end{align}
where $\bar{Y}_t := \sum_{j=t}^M Y_j$.

Collecting all the aforementioned discussion, we find that in order to identify the structural parameters of interest we need to estimate the following auxiliary nuisance models:
\begin{definition}[Nuisance Functions]\label{defn:nuisances}  For $2\leq t\leq \tau\leq M$, denote with:
\begin{align*}
    \Phi_{1} :=~& \phi_{e,1}(T_1, S_0) &
    \Phi_{o, t} :=~& \phi_{o,t}(T_2, S_{t-1}) & 
    \tilde{\Phi}_1 :=~& \Phi_1 - \E_e[\Phi_1\mid S_0]\\
    g^*(S_1) :=~& \E_o\left[\bar{Y}\mid S_1\right] &
    g_t^*(S_1) :=~& \E_o\left[\Phi_{o,t}\mid S_1\right] &
    q^*(S_1) :=~& \frac{\Pr(e\mid S_1)}{\Pr(o\mid S_1)}\, \E_e[\tilde{\Phi}_1\mid S_{1}] \\
    h^*(S_0) :=~& \E_e[g^*(S_1)\mid S_0] &
    p_{e, t}^*(S_0) :=~& \E_e\left[g_t^*(S_1)\mid S_0\right] &
    p_{e, 1}^*(S_0) :=~& \E_e[\Phi_1\mid S_0]\\
    \bar{Y}_{t} :=~& \sum_{j=t}^M Y_j &
    b_{o, t}^*(S_{t-1}) :=~& \E_o\left[\bar{Y}_t \mid S_{t-1}\right] &
    p_{o, \tau, t}^*(S_{t-1}) :=~& \E_o\left[\Phi_{o,\tau}\mid S_{t-1}\right]
\end{align*}
Then we define as $f=\{g, q, h, p_{e,1}, \{g_t, p_{e,t}, b_{o, t}, p_{o,t,\tau}\}_{2\leq \tau\leq t\leq M}\}$ the nuisance functions for identifying the target parameter $\tau(t_1, t_0)$ and with $f^*$ their true values.\footnote{We note that the nuisance functions $g_t$ and $p_{o,t,2}$ are modeling the exact same quantity. However, for notational convenience we overload notation and give two symbols that will be used when the nuisance function enters at different parts of the moment equations.}
\end{definition}

Note that all the nuisance functions $f$ are estimable from the observed data. All nuisance functions except $q$ correspond to a regression problem and $q$ can be decomposed into a classification problem for estimating the odds ratio $\frac{\Pr(e\mid S_1)}{1 - \Pr(e\mid S_1)}$ and a regression problem for estimate $\E_e\left[\tilde{\Phi}_1\mid S_1\right]$. Given these nuisance models we can define the parameter $\theta_0$ of interest as the solution to a set of moment restrictions that are Neyman orthogonal with respect to all the nuisance functions. To state our theorem we first define the vector of orthogonal scores.
\begin{definition}[Orthogonal Scores]\label{defn:scores}
Let $\theta = (\theta_0; \theta_{o,2};\ldots; \theta_{o, M})$ denote the target parameters of interest and $Z$ denote the vector of all random variables. The \emph{orthogonal score vector} $\psi=(\psi_1;\ldots; \psi_M)$ for the target parameter is defined as: 
\begin{align*}
    \psi_{1}(Z; \theta, f) :=~& \psi_{e,1}(Z; \theta, f) + \psi_{e,1}(Z; \theta, f)\\
    \psi_{e,1}(Z; \theta, f) :=~& 1\{e\}\, (g(S_1) - h(S_0) - \sum_{t=2}^M \theta_{o,t}^\top \left(g_t(S_1) - p_{e,t}(S_0)\right) - \theta_0^\top (\Phi_1 - p_{e,1}(S_0)))\cdot (\Phi_1 - p_{e,1}(S_0))\\
    \psi_{o,1}(Z; \theta, f) :=~& 1\{o\}\, q(S_1) \left(\bar{Y} - g(S_1) - \sum_{t=2}^M \theta_{o,t}^\top (\Phi_{o,t} - g_t(S_1))\right)\\
    \forall t\in [2, M]: \psi_{t}(Z; \theta, f) :=~& 1\{o\} \left(\bar{Y}_{t} - b_{o,t}(S_{t-1}) - \sum_{\tau=t}^{M} \theta_{o,\tau}^\top\left(\Phi_{o,\tau} - p_{o,\tau, t}(S_{t-1})\right)\right)\, \left(\Phi_{o,t} - p_{o,\tau, t}(S_{t-1})\right)
\end{align*}
\end{definition}

We are now ready to state our main semi-parametric identification theorem via Neyman orthogonal moment restrictions:

\begin{theorem}[Semi-Parametric Dynamic Orthogonal Moment Restrictions]\label{thm:dyn-ortho}
Let $\theta = (\theta_0; \theta_{o,2};\ldots; \theta_{o, M})$ denote the target parameters of interest and $\theta^*$ their true values. Then $\theta^*$ is a solution to the system of moment equations:
\begin{align}
    m(\theta; f^*):=  \E\left[\psi(Z; \theta, f^*)\right] = 0
\end{align}
Moreover, if for all $t\in [2,M]$:
\begin{align}
    \E_o\left[\Cov_o(\Phi_{o,t}, \Phi_{o,t}\mid S_{t-1})\right] \succ~& 0 
    &
    \E_e\left[\Cov_e(\Phi_{1}, \Phi_{1}\mid S_{0})\right] \succ~& 0 
\tag{average overlap}
\end{align}
then $\theta^*$ is the unique solution. Finally, the moment $m(\theta; f)$ satisfies the Neyman orthogonality property with respect to $f$.
\end{theorem}

\subsection{Semi-Parametric Adaptive Estimation and Inference}\label{ssec:ddsurr_wadj:inference}

Given that we have formulated the target structural parameters of interest as the solution to a vector of Neyman orthogonal moment equations, we can now easily transfer this identification argument to an estimation strategy, by invoking standard approaches. In particular, our estimation strategy will first estimate and apply the nuisance functions in a cross-fitting manner and subsquently solve a plug-in empirical analogue of the moment equations. Algorithm~\ref{alg:main} provides a formal description of the process.

\begin{algorithm}[htpb]
   \caption{Double/Debiased Machine Learning for Dynamically Adjusted Surrogates in SNMMs}
   \label{alg:main}
\begin{algorithmic}
   \STATE {\bf Input:} An experimental/short-term data set consisting of $n_e$ samples: $\left\{\left(S_0^i, T_1^i, S_1^i\right)\right\}_{i=1}^{n_e}$
   \STATE {\bf Input:} An observational/long-term data set consisting of $n_o$ samples: $\left\{\left(S_1^i, Y_1^i, T_2^i, S_2^i, Y_2^i,\ldots, S_{M-1}^i, T_M^i, Y_M^i\right)\right\}_{i=1}^{n_o}$
   \STATE {\bf Input:} $\fdim$-dimensional feature vector maps $\{\phi_{o,t}\}_{t=2}^M$ and $\phi_{e,1}$, which parameterize the blip functions
   \STATE Randomly split in half the $n_e$ and $n_o$ samples in $S, S'$
   \STATE Estimate nuisance models $\hat{f}_O$ of $f=\{g, q, h, p_{e,1}, \{g_t, p_{e,t}, b_{o, t}, p_{o,t,\tau}\}_{2\leq \tau\leq t\leq M}\}$ on each half-sample $O\in \{S, S'\}$.
   \STATE Using all the data $S\cup S'$, estimate structural parameters $\hat{\theta}$ by solving the empirical plug-in vector of moment equations:
   \begin{align}\label{eqn:z-estimator-alg}
     \sum_{i\in S} \psi(Z^{i}; \theta, f_{S'}) + \sum_{i\in S'} \psi(Z^{i}; \theta, f_{S}) = 0
   \end{align}
   \STATE {\bf Return:} Structural parameter estimate $\hat{\theta} = (\hat{\theta}_0; \hat{\theta}_{o,2};\ldots; \hat{\theta}_{o, M})$
\end{algorithmic}
\end{algorithm}

To guarantee that our estimator is root-$n$ consistent and asymptotically normal, we need to assume that our first stage estimates of the nuisance functions are sufficiently accurate. In particular, we need to make the following nuisance rate assumptions:

\begin{assumption}[Nuisance Rates]\label{ass:nuisance-rate} For any vector of nuisance estimates $\hat{f}=\{\hat{g}, \hat{q}, \hat{h}, \hat{p}_{e,1}, \{\hat{g}_t, \hat{p}_{e,t}, \hat{b}_{o, t}, \hat{p}_{o,t,\tau}\}_{2\leq \tau\leq t\leq M}\}$, let $\nu_{\hat{f}} = \hat{f}-f^*$. A vector of nuisance estimates $\hat{f}$ satisfies the \emph{sufficient rate assumption} if $\|\nu_{\hat{f}}\|_{2} = o_p(1)$ and:
\begin{align*}
    \|\nu_{\hat{p}_{e,1}}\|_{2,e}\cdot \left\|\nu_{\hat{g}}-\nu_{\hat{h}} - \sum_{t=2}^M \left(\nu_{\hat{g}_t} - \nu_{\hat{p}_{e,t}}\right)^\top \theta_{o,t}^* + \nu_{\hat{p}_{e,1}}^\top\theta_{0}^*\right\|_{2,e}=~& o_p(n^{-1/2})\\
    \|\nu_{\hat{q}}\|_{2,o}\cdot \left\|\nu_{\hat{g}} - \sum_{t=2}^M \nu_{\hat{g}_t}^\top\,\theta_{o,t}^*\right\|_{2,o} =~& o_p(n^{-1/2})\\
    \forall t\in [2,M]:  \|\nu_{\hat{p}_{o,t,t}}\|_{2,o}\cdot \left\|\nu_{\hat{b}_{o,t}}+ \sum_{j=t}^M \nu_{\hat{p}_{o,j,t}}^\top \theta_{o,j}^*\right\|_{2,o} =~& o_p(n^{-1/2})
\end{align*}
\end{assumption}

Note that these nuisance rate assumptions possess almost a doubly robust flavor. With the exception of the nuisance quantities $\hat{p}_{o,t,t}$ and $\hat{p}_{e,1}$, which need to admit $o_p(n^{-1/4})$ root-mean-squared-error (RMSE) rates, for the remainder of the nuisance functions it suffices that the product of their RMSE rates with some other nuisance function be $o_p(n^{-1/2})$ and not that they individually satisfy $o_p(n^{-1/4})$ rates. For instance, if we knew the treatment policy in the experimental sample (captured by the propensity $p_{e,1}$) and the dynamic treatment policy in the observational sample (captured by the dynamic porpensity $p_{o, t,t}$), then we don't need any rates for $\hat{h}, \hat{p}_{e,t}, \hat{b}_{o,t}, \{\hat{p}_{o,j,t}\}_{t<j}$. Moreover, it suffices that the product of the surrogate score $\hat{q}$ error and the surrogate indices $\hat{g}, \hat{g}_t$ error, be small. Subject to these nuisance rate conditions we can show asymptotic normality of our estimate and provide asymptotically valid confidence intervals.

\begin{theorem}[Estimation and Inference for Structural Parameters]\label{thm:dyn-estimation}
Let $\mcD_n$ be a sequence of families of data generating processes for the experimental and observational sample, which adhere to the causal graph presented in Figure~\ref{fig:simplified_cg} and which satisfy Assumption~\ref{ass:linear-snmm} and Assumption~\eqref{ass:dynIR}, for a constant feature map dimension $\fdim$ and such that all random variables and the ranges of all nuisance functions are bounded by a constant a.s.. Moreover, for some $\lambda >0$, for any $D\in \mcD_n$:
\begin{align}
    \E_o\left[\Cov_o(\Phi_{o,t}, \Phi_{o,t}\mid S_{t-1})\right] \succeq~& \lambda I
    &
    \E_e\left[\Cov_e(\Phi_{1}, \Phi_{1}\mid S_{0})\right] \succeq~& \lambda I
\tag{strict average overlap}
\end{align}
Assume that the nuisance function estimates $\hat{f}_S, \hat{f}_{S'}$, estimated on each half-sample in the first stage of Algorithm~\ref{alg:main}, satisfy the sufficient rate Assumption~\ref{ass:nuisance-rate}. 

Let $\psi$ denote the orthogonal scores in Definition~\ref{defn:nuisances}. Let $J=\E[\nabla_{\theta} \psi(Z;\theta^*, f^*)]$ and $\Sigma = \E\left[\psi(\theta^*; f^*)\, \psi(\theta^*; f^*)^\top\right]$ and $V=J^{-1} \Sigma (J^{-1})^\top$. 
Then the estimate $\hat{\theta}$ returned by Algorithm~\ref{alg:main} satisfies:
\begin{align}
    \sqrt{n} V^{-1/2} (\hat{\theta} - \theta^*) = -\frac{1}{\sqrt{n}}\sum_{i=1}^n V^{-1/2} J^{-1}\, \psi(Z^i; \theta^*, f^*) + o_p(1) \to_d N(0, I_{\fdim\cdot M})
\end{align}
Let $\hat{J}=\frac{1}{n}\sum_{i=1}^n \nabla_{\theta} \psi(Z^i;\hat{\theta}, \hat{f})$ and $\hat{\Sigma}=\frac{1}{n} \sum_{i=1}^n \psi(Z^i; \hat{\theta}, \hat{f})\, \psi(Z^{i}; \hat{\theta}, \hat{f})^\top$ and $\hat{V}=\hat{J}^{-1} \hat{\Sigma} (\hat{J}^{-1})^\top$.
Let ${\cal F}$ be the CDF of the standard normal distribution. Then for any vector $\nu \in \R^{d\,m}$, the confidence interval:
\begin{align}
    \CI := \left[\nu^\top \hat{\theta} \pm {\cal F}^{-1}(1-\alpha/2)\, \sqrt{\frac{\nu'\hat{V} \nu}{n}}\right]
\end{align}
is asymptotically uniformly valid:
    $\sup_{D\in \mcD_n} \left|\Pr_D(\nu^\top \theta^* \in \CI) - (1-\alpha)\right| \to 0$.
\end{theorem}

\vscomment{Add potential discussion on variance and jacobian.}
\vsdelete{denote a $(m\,\fdim)\times (m\,\fdim)$ upper triangular matrix consisting of $\fdim\times \fdim$ blocks, such that the $(t,j)$ block, for $t\geq 2$ is defined as:
\begin{align*}
    J_{t,j}:=~& 1\{t\leq j\}\Pr(o)\E_o[\Cov_o(\Phi_{o,t}, \Phi_{o,j}\mid {S}_{t-1})]\\
    J_{1, j} :=~& 1\{j=1\} \Pr(e)\E_{e}\left[\Cov_e(\Phi_1, \Phi_1\mid S_0)\right] + 1\{j>1\} \Pr(e)\E_e\left[\Cov_e(\Phi_1, \E_o[\Phi_{o,t}\mid S_1] \mid S_0)\right]
\end{align*}
be an estimate of $J$ such that the $(t,j)$ block, for $t\geq 2$ is defined as:
\begin{align*}
    \hat{J}_{t,j}:=~& 1\{t\leq j\} \frac{1}{n}\sum_{i\in o} \left(\Phi_{o,t} - \hat{p}_{o,t,t}(S_{t-1})\right)\, \left(\Phi_{o,j} - \hat{p}_{o, j, t}(S_{t-1})\right)^\top\\
    \hat{J}_{1, j} :=~& 1\{j=1\} \frac{1}{n} \sum_{i\in e} \left(\Phi_1 - \hat{p}_{e,1}(S_0)\right)\, \left(\Phi_1 - \hat{p}_{e,1}(S_0)\right)^\top + 1\{j>1\} \frac{1}{n} \sum_{i\in e} \left(\Phi_1 - \hat{p}_{e,1}(S_0)\right)\, \left(\hat{g}_j(S_1) - \hat{p}_{e,j}(S_0)\right)^\top
\end{align*}
}

\begin{corollary}[Estimation and Inference for Treatment Effects]
Under the assumptions and definitions of Theorem~\ref{thm:dyn-estimation}, the following is an estimate of the target value $\tau(t_1, t_0)$:
\begin{align}
    \hat{\tau}(t_1, t_0) := \hat{\theta}_0^\top \left(\frac{1}{n_e} \sum_{i\in e} (\phi_{e,1}(t_1, S_0^i) - \phi_{e,1}(t_0, S_0^i))\right)
\end{align}
If we let $Q=\phi_{e,1}(t_1, S_0) - \phi_{e,1}(t_0, S_0)$ and $\E_{n,e}[\cdot], \Var_{n,e}$ the empirical average and empirical variance over samples from setting $e$, then:
\begin{align}
    \frac{\sqrt{n}}{\sqrt{\gamma + \mu}} \left(\hat{\tau}(t_1, t_0) - \tau(t_1, t_0)\right) = \frac{1}{\sqrt{n}} \sum_{i=1}^n \frac{1}{\sqrt{\gamma + \mu}}f(Z^i;\theta^*,h^*) + o_p(1) \to_d N(0, 1)
\end{align}
with $f(Z; \psi^*, h^*) := \frac{1\{e\}}{\Pr(e)}\left(Q - \E_e[Q]\right)^\top\theta_0^* - \E_e[Q]^\top e_{1:\fdim}^\top\,  J^{-1}\psi(Z;\theta^*, f^*)$ and $\gamma=\frac{1}{\Pr(e)}\Var_e(Q^\top\theta^*)$ and $\mu=\E_e[Q]^\top V_{1:d,1:d}\E_e[Q]$, where we denoted with $e_{1:d}\in \R^{\fdim\cdot M}$ the vector with entries of $1$ on the first $\fdim$ coordinates and zero otherwise, $\phi, V$ as defined in Theorem~\ref{thm:dyn-estimation} and $V_{1:\fdim,1:\fdim}$ is the submatrix of $V$ consisting of its first $\fdim$ rows and columns. Moreover, if we let $\hat{\gamma}=\frac{n}{n_e}\Var_{n,e}(Q^\top \hat{\theta}_0)$ and $\hat{\mu}=\E_{n,e}[Q]^\top\hat{V}_{1:d,1:d}\E_{n,e}[Q]$, with $\hat{V}$ as in Theorem~\ref{thm:dyn-estimation}, then the confidence interval:
\begin{align}
    \CI := \left[\hat{\tau}(t_1, t_0) \pm \mcF^{-1}(1-\alpha/2)\, \sqrt{\frac{\hat{\gamma}+\hat{\mu}}{n}}\right]
\end{align}
is asymptotically uniformly valid:
    $\sup_{D\in \mcD_n} \left|\Pr_D(\tau(t_1, t_0) \in \CI) - (1-\alpha)\right| \to 0$.
\end{corollary}

The asymptotic linearity of our estimate also allows for alternative computationally convenient resampling methods for the construction of confidence intervals, with potentially better finite sample properties. For instance, constructing intervals by running the Bootstrap on the final stage estimation (keeping the nuisance estimates fixed), will be asymptotically valid. Moreover, the computationally even more convenient multiplier Bootstrap can also be used \cite{Chatterjee2005,Chernozhukov2013,Chernozhukov2014,Spokoiny2015,Zhilova2020}, which can also be used for joint inference on multiple parameters, such as for constructing uniform confidence bands on dose response curves, i.e. the curve of the form $t \to \tau(t, 0)$, for $t$ in some bounded range $[U, L]$.

\section{A High-Dimensional Linear Markovian Data Generating Process Example}

As a simple example where the linear SNMM assumption holds, consider the following linear Markovian (albeit high-dimensional) data generating process:
\begin{align}
    S_{t} =~& A T_{t} + B S_{t-1} + \epsilon_t \nonumber\\
    Y_{t} =~& C S_{t} + \eta_t \label{eqn:dgp}\\
    T_{t+1} =~& D T_{t-1} + G S_{t-1} + \zeta_t \nonumber
\end{align}
where $\epsilon_t, \eta_t, \zeta_t$ are i.i.d. random shocks. Our assumptions are satisfied if the quantities $B, C$ remain unchanged between the experimental and the observational setting, while the quantities $A, D, G$, as well as the distributions of mean-zero random shocks, can change arbitrarily, in the two settings, denoted as $A_d, D_d, G_d$ for $d\in\{o, e\}$.

In this case, the blip functions take the simple form: $\psi(\tau_t, s_{t-1})=\tau_t$ and $\theta_{t,j}=C B^{j-t} A$. Moreover, note that in this case, for any non-adaptive sequence of treatments $\tau_{>1}$, we have that:
\begin{align}
    \E_e\left[Y^{(t_1, \tau_{>1})} - Y^{(t_0, \tau_{>1})}\right] = \E_e\left[Y^{(t_1, 0_{>1})} - Y^{(t_0, 0_{>1})}\right] = \sum_{j=1}^m \theta_{1, j}^\top (t_1 - t_0)
\end{align}
Thus the quantity that our algorithm estimates is valid, irrespective of the baseline future policy that one considers and is a universal effect quantity that holds under any non-adaptive future sequence of treatments. This is practically convenient, as the causal effect derived is not heavily dependent on the future treatment sequence that a sample will receive in the short-term data set.
Finally, note, that even though the surrogates/controls can be high-dimensional objects and hence the matrices $B, C, G$ are high-dimensional objects, our estimation strategy allows to estimate the target parameter $\theta_0 = \sum_{j=1}^m \theta_{1,j}$, which is low-dimensional at parametric root-$n$ rates and with asymptotically normal distributional limits. The intuition is that our analysis and estimation strategy, never really identifies or argues about estimation errors of these intermediate high-dimensional quantities.

\vspace{-.1in}
\section{Semi-Synthetic Experimental Evaluation}\label{app:semidata}

We evaluate the performance of our proposed estimation strategy on a semi-synthetic dataset. The semi-synthetic data retain qualitative characteristics of data on real-world incentive investments in customers at a major corporation, although all data series and relationships have been perturbed to retain confidentiality.

The semi-synthetic dataset, like the real-world dataset on which it is based, displays several patterns that are common across many potential applications. The treatments, in this case incentive investments, are lumpy: in most periods most customers get no investments. Proxies, which include single period values of the outcome of interest, are highly auto-correlated over time. Treatments are also auto-correlated, and correlated with past values of proxies. Finally, we include a set of time-invariant controls that affect both proxies/outcomes and treatments.

To build the semi-synthetic data we estimate a series of moments from a real-world dataset: a full covariance matrix of all proxies, treatments, and controls in one period and a series of linear prediction models (lassoCV) of each proxy and treatment on a set of 6 lags of each treatment, 6 lags of each proxy, and time-invariant controls. Using these values, we draw new parameters from distributions matching the key characteristics of each family of parameters. Finally, we use these new parameters to simulate proxies, treatments, and controls by drawing a set of initial values from the covariance matrix and forward simulating to match intertemporal relationships from the transformed prediction models. For further details on the data generation process, see~Appendix Section \ref{sec_datagen}.

We now compare multiple possible approaches for estimating the effects of our three synthetic treatments on a long-term outcome. To construct this outcome we select one proxy to be the outcome of interest. We consider the effect of each treatment in period $t$ on the cumulative sum of the outcome from period $t$ to $t+3$, four periods, or $t$ to $t+7$, eight periods. We can calculate the true treatment effects in the synthetic data as a function of parameters from the linear prediction models.

\begin{figure*}[ht]
\begin{subfigure}{1\textwidth}
\centering
    \centerline{\includegraphics[width=\textwidth]{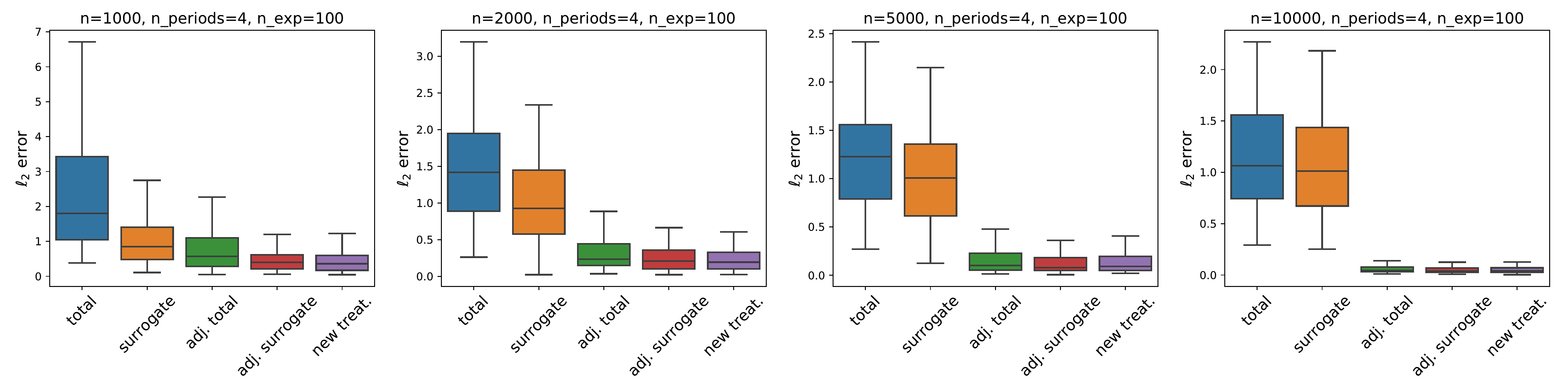}}
\end{subfigure}
\begin{subfigure}{1\textwidth}
\vspace{-.07in}
\centering
    \centerline{\includegraphics[width=\textwidth]{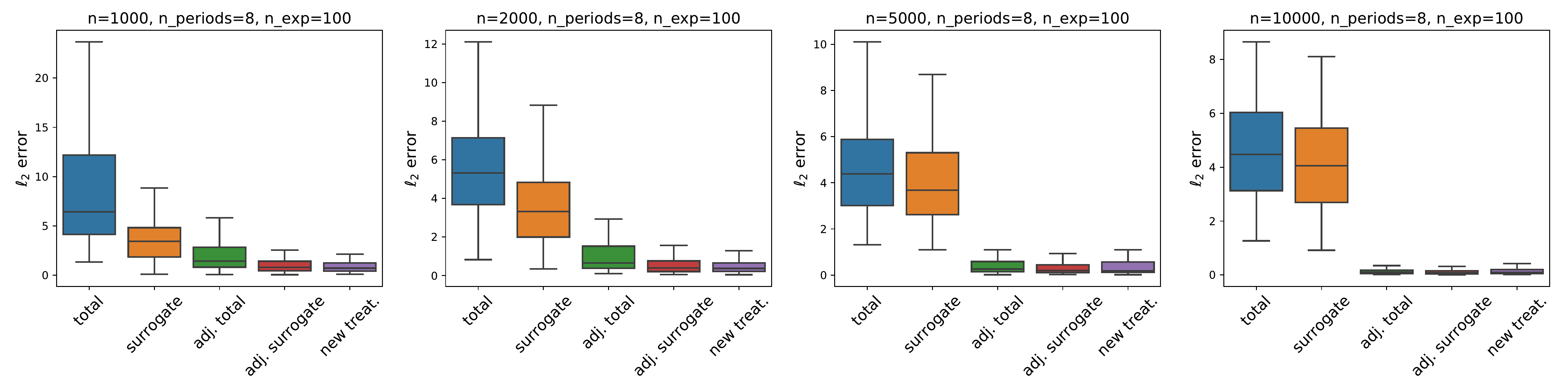}}
\vspace{-.1in}
\end{subfigure}
\caption{Experimental performance for $M=4$ periods and $M=8$ periods.}
\label{fig:experimental}
\end{figure*}

Because we construct a single, long synthetic dataset for this exercise it is possible to estimate the treatment effects on realized long-term outcomes directly, unlike the typical use case for a surrogate approach. Following a likely approach, we estimate the effect of each treatment at time $t$ on outcomes over the next 4 or 8 periods using double machine learning and controlling for invariant customer characteristics and \textit{contemporaneous and lagged} values of all proxies and other treatments. The blue, ``total" bars in each panel of Figure \ref{fig:experimental} show the distribution of the estimation error\footnote{We use the $\ell_2$ error $\|\hat{\theta}-\theta_0\|_2$.} in the estimated treatment effects from this method across 100 simulated datasets. The top row plots the estimation error when estimating the effect on four periods of outcomes, increasing the sample size of each simulation from left to right, while the bottom row shows the same for the effect on eight periods of outcome. As predicted, the auto-correlation in treatments causes this method, which does not control for \textit{future} treatments, to substantially overestimate treatment effects relative to their true values.

We then estimate the same set of treatment effects using the unadjusted surrogate approach described in Section \ref{ssec:ddsurr_noadj}. The distribution of estimation errors from this approach is represented in the orange ``surrogate" bars in each panel of Figure \ref{fig:experimental}. Since this approach still fails to control for future treatments when estimating the surrogate index, the estimated treatment effects are still substantially larger than the true effects on average. Note that the surrogate model exhibits slightly less bias than the direct ``total" approach. Intuitively, because the surrogate approach is only capturing the relationship between treatment and outcome that passes through the surrogates it picks up less of the bias resulting from future correlated treatments than the direct approach.

The third set of green ``adj. total" bars plot the distribution of estimation errors when estimating treatment effects on adjusted realized outcomes using the method of \citet{lewis2020doubledebiased}. When a dataset containing both all treatments of interest and realized long-run outcomes is available, this should be the preferred approach. This third methodology, which removes the effects of future treatments from the long-run outcome in a first step, exhibits significantly less bias than the first two methods, particularly for reasonably large samples in the right two columns.

The final two bars in each panel of Figure \ref{fig:experimental} illustrate the success of the adjusted surrogate approach described in Section \ref{ssec:ddsurr_wadj}. We recommend this approach in the case when treatments are serially-correlated, as in the synthetic data, and it is not possible to collect a single dataset that contains both long-term outcomes and all the treatments of interest. As illustrated by the red ``adj. surrogate" bars, this adjusted surrogate approach is highly accurate in predicting long-term effects with a performance comparable to that of having access to the raw long-term outcome itself. The final purple ``new treat." bars show that the approach works equally well when considering the effect of a novel treatment that appears only in the experimental sample and was not part of the dynamic adjustment. Overall, this methodology overcomes a common data limitation when considering long-term effects of novel treatments and expands the surrogate approach to consider a common, and previously problematic, pattern of serially correlated treatments.

\bibliography{bib}
\bibliographystyle{icml2021}

\appendix

\input{appendix}

\end{document}

%% file: appendix.tex
\newpage

\onecolumn

\section{Proofs of Theorems in Section \ref{ssec:ddsurr_noadj}}

\subsection{Proof of Theorem~\ref{thm:id}}
\paragraph{Surrogate index representation.} By the causal graph assumption and the IR assumption we have:
\begin{align*}
    \E_e[\bar{Y} \mid T_1, S_0] =~& \E_e[\E_e[\bar{Y}\mid S_1, T_1, S_0]\mid T_1, S_0] = \E_e[\E_e[\bar{Y}\mid S_1]\mid T_1, S_0] = \E_e[\E_o[\bar{Y}\mid S_1]\mid T_1, S_0]\\
    =~& \E_e[g_0(S_1)\mid T_1, S_0]
\end{align*}
Thus also we have that $\E_e[\bar{Y} \mid S_0]=\E_e[g_0(S_1)\mid S_0]$ and that:
\begin{align*}
    \E_e[\bar{Y} \mid T_1, S_0] - \E_e[\bar{Y} \mid S_0] = \E_e[g_0(S_1)\mid T_1, S_0] - \E_e[g_0(S_1)\mid S_0] = \E_e[\tilde{g}_0(S_1)\mid T_1, S_0]
\end{align*}

By the PLR assumption and the definition of $\tilde{\Phi}_1 := \phi(T_1,S_0) - \E[\phi(T_1,S_0)\mid S_0]$, we have that:
\begin{equation*}
    \E_e[\bar{Y}\mid T_1, S_0] - \E_e[\bar{Y}\mid S_0] = \theta_0^{\top} \tilde{\Phi}_1 \implies \E[\tilde{g}_0(S_1, S_0) - \theta_0^{\top} \tilde{\Phi}_1\mid T_1, S_0] = 0\implies  \E\left[\left(\tilde{g}_0(S_1, S_0) - \theta_0^{\top} \tilde{\Phi}_1\right) \tilde{\Phi}_1\right] = 0 
\end{equation*}
Solving for $\theta_0$, yields the surrogate index representation:
\begin{align}
    \theta_0 := \E_e[\tilde{\Phi}_1\, \tilde{\Phi}_1^\top]^{-1} \E_{e}[\tilde{g}_0(S_1, S_0) \tilde{\Phi}_1]
\end{align}

\paragraph{Surrogate score representation.} We start by the definition of the surrogate score representation and use the causal graph assumptions to derive:
\begin{align*}
    I_{ss} :=~& \E_{e}[\tilde{\Phi}_1 \tilde{\Phi}_1^\top]^{-1}
    \E_{o}\left[\frac{\Pr(e\mid S_1)}{\Pr(o\mid S_1)} \frac{\Pr(o)}{\Pr(e)}\, \bar{Y}\, \E_{e}[\tilde{\Phi}_1\mid S_{1}]\right]\\
    =~& \E_{e}[\tilde{\Phi}_1 \tilde{\Phi}_1^\top]^{-1}
    \E_{o}\left[\frac{\Pr(S_1\mid e)}{\Pr(S_1\mid o)}\, \bar{Y}\,  \E_{e}[\tilde{\Phi}_1\mid S_{1}]\right] \tag{Bayes-rule}\\
    =~& \E_{e}[\tilde{\Phi}_1 \tilde{\Phi}_1^\top]^{-1}
    \E_{o}\left[\frac{\Pr(S_1\mid e)}{\Pr(S_1\mid o)}\, \E_o[\bar{Y}\mid S_1]\, \E_e[\tilde{\Phi}_1\mid S_1]\right] \tag{tower-law}\\
     =~& \E_{e}[\tilde{\Phi}_1 \tilde{\Phi}_1^\top]^{-1}\,
    \E_{e}\left[\E_o[\bar{Y} \mid S_1]\, \E_e[\tilde{\Phi}_1\mid S_1]\right] \tag{change of measure}\\
    =~& \E_{e}[\tilde{\Phi}_1 \tilde{\Phi}_1^\top]^{-1}
    \E_{e}\left[\E_e[\bar{Y} \mid S_1]\, \E_e[\tilde{\Phi}_1\mid S_1]\right] \tag{\eqref{eqn:ass:ir} assumption}\\
    =~& \E_{e}[\tilde{\Phi}_1 \tilde{\Phi}_1^\top]^{-1}
    \E_{e}\left[\E_e[\bar{Y}\, \tilde{\Phi}_1 \mid S_1] \right] \tag{\eqref{eqn:ass:meanid} assumption}\\
    =~& \E_{e}[\tilde{\Phi}_1 \tilde{\Phi}_1^\top]^{-1}
    \E_{e}\left[\bar{Y} \tilde{\Phi}_1 \right] \tag{inverse tower-law}\\
    =~& \E_{e}[\tilde{\Phi}_1 \tilde{\Phi}_1^\top]^{-1}
    \left(\E_{e}\left[\tilde{\Phi}_1\, \phi(T_1, S_0)^\top \theta_0 \right] + \E_{e}\left[b_0(S_0)\, \tilde{\Phi}_1 \right]\right) \tag{\eqref{eqn:ass:plr} assumption}\\
    =~& \E_{e}[\tilde{\Phi}_1 \tilde{\Phi}_1^\top]^{-1}
    \E_{e}\left[\tilde{\Phi}_1\, \tilde{\Phi}_1^\top \theta_0\right] \tag{mean-zero: $\E[\bar{\Phi}_1 \mid S_0]=0$}\\
    =~& \theta_0
\end{align*}

\paragraph{Orthogonal representation.}
This follows easily by the fact that the second term in the orthogonal representation is mean zero and the first term is equal to $\theta_0$ by the surrogate-index argument.

\subsection{Proof of Theorem~\ref{thm:ortho}}

\paragraph{Orthogonal moment formulation.} First we see that we can re-write the orthogonal representation from Theorem~\ref{thm:id} as follows:
\begin{align}
\theta_0 =~& \E_{e}[\tilde{\Phi}_1 \tilde{\Phi}_1^\top]^{-1} \left(\frac{\E[\tilde{g}_0(S_1) \tilde{\Phi}_1 1\{e\}]}{\Pr(e)} + \frac{1}{\Pr(e)}\E\left[1\{o\}\frac{\Pr(e\mid S_1)}{\Pr(o\mid S_1)} (\bar{Y} - g_0(S_1)) \E_{e}[\tilde{\Phi}_1\mid S_{1}, S_{0}]\right]\right)\\
    =~& \E[\tilde{\Phi}_1\, \tilde{\Phi}_1^\top 1\{e\}]^{-1} \left(\E[\tilde{g}_0(S_1) \tilde{\Phi}_1 1\{e\}]+ \E\left[1\{o\}\frac{\Pr(e\mid S_1)}{\Pr(o\mid S_1)} (\bar{Y} - g_0(S_1)) \E_{e}[\tilde{\Phi}_1\mid S_{1}, S_{0}]\right]\right)
\end{align}
Equivalently, the solution to the orthogonal moment equation:
\begin{align}
    \E\left[1\{e\}\, (\tilde{g}_0(S_1) - \theta_0^\top \tilde{\Phi}_1) \tilde{\Phi}_1 + 1\{o\}\, \frac{\Pr(e\mid S_1)}{\Pr(o\mid S_1)} \E_{e}[\tilde{\Phi}_1\mid S_{1}] (\bar{Y} - g_0(S_1))  \right] = 0
\end{align}
By the definition of $q_0, p_0$ and $h_0$ the result follows that we can write $\theta_0$ as the solution to $m(\theta_0; f_0)=0$, where $f_0=(q_0, p_0, h_0)$ and we remind that:
\begin{align*}
    m(\theta; f):=~& m_e(\theta; f) + m_o(\theta; f)\\
    m_e(\theta; f) :=~& \E\left[1\{e\}\, (g(S_1) - h(S_0) - \theta^\top (\phi(T_1,S_0) - p(S_0)))\, (\phi(T_1,S_0) - p(S_0))\right]\\
    m_o(\theta_0; f):=~& \E\left[1\{o\}\, q(S_1) (\bar{Y} - g(S_1))  \right]
\end{align*}

\paragraph{Orthogonality.} For any nuisance $f$, let $\nu_f = f - f_0$, denote the difference with respect to its corresponding true value. Orthogonality with respect to $p, h, q$, follows since:
\begin{align*}
    D_p[m(\theta_0;f_0), \nu_p] :=~& \Pr(e)\,\left(\E_e[\tilde{\Phi}_1\, \nu_p(S_0)] - \E_{e}\left[\left(\tilde{g}_0(S_1) - \theta_0^\top \tilde{\Phi}_1\right) \, \nu_p(S_0)\right]\right) = 0\\
    D_h[m(\theta_0;f_0), \nu_h] :=~& \Pr(e)\,\E_e[-\tilde{\Phi}_1\, \nu_h(S_0)] = 0\\
    D_q[m(\theta_0;f_0), \nu_q] :=~& \Pr(o)\,\E_o[(\bar{Y} - g_0(S_1))\, \nu_q(S_1)]
    = \Pr(o)\,\E_o[\E[\bar{Y} - g_0(S_1))\mid S_1]\, \nu_q(S_1)] = 0
\end{align*}
Orthogonality with respect to $g$ is slightly more involved. First note that:
\begin{align*}
    D_g[m(\theta_0;f_0), \nu_g] =~& \Pr(e)\,\E_e[\tilde{\Phi}_1\, \nu_g(S_1)] - \Pr(o) \E_o[q_0(S_1)\, \nu_g(S_1)]
\end{align*}
Now using a sequence of derivations almost identical to the ones we invoked in the surrogate score representation in Theorem~\ref{thm:id}, we can show that:
\begin{lemma}\label{lem:transform}
For any scalar valued function $r(S_1)$,
\begin{equation}
    \E_o\left[q_0(S_1)\, r(S_1)\right] = \frac{\Pr(e)}{\Pr(o)}\E_e[\tilde{\Phi}_1\, r(S_1)]
\end{equation}
\end{lemma}
\begin{proof}
\begin{align*}
\frac{\Pr(e)}{\Pr(o)}\E_o\left[q_0(S_1)\, r(S_1)\right] :=~&    \E_{o}\left[\frac{\Pr(e\mid S_1)}{\Pr(o\mid S_1)} \frac{\Pr(o)}{\Pr(e)}\, \E_{e}[\tilde{\Phi}_1\mid S_{1}]\, r(S_1)\right] \tag{definition of $q_0$}\\
    =~&
    \E_{o}\left[\frac{\Pr(S_1\mid e)}{\Pr(S_1\mid o)}\,\E_{e}[\tilde{\Phi}_1\mid S_{1}]\, r(S_1)\right] \tag{Bayes-rule}\\
     =~& 
    \E_{e}\left[\E_e[\tilde{\Phi}_1\mid S_1]\, r(S_1)\right] \tag{change of measure}\\
    =~& 
    \E_{e}\left[\E_e[\tilde{\Phi}_1\, r(S_1)\mid S_1]\right]\\
    =~& 
    \E_{e}\left[\tilde{\Phi}_1\, r(S_1) \right] \tag{inverse tower-law}
\end{align*}
\end{proof}
Applying Lemma~\ref{lem:transform} for $r(S_1) = \nu_g(S_1)$, we have:
\begin{align*}
    D_g[m(\theta_0;f_0), \nu_g] =~& \Pr(e)\,\E_e[\tilde{\Phi}_1\, \nu_g(S_1)] - \Pr(o) \frac{\Pr(e)}{\Pr(o)}\,\E_e[\tilde{\Phi}_1\, \nu_g(S_1)] = 0
\end{align*}

\paragraph{Double robustness.} Observe that by re-writing the orthogonal moment equation, and applying the definition of $g_0(S_1)$ and Lemma~\ref{lem:transform} we have:
\begin{align*}
    \E_e[\tilde{\Phi}_1 \tilde{\Phi}_1^\top] \hat{\theta} =~& \E_e[\hat{g}(S_1)\, \tilde{\Phi}_1] + \frac{\Pr(o)}{\Pr(e)}\E_{o}\left[(\bar{Y} - \hat{g}(S_1)) \hat{q}(S_1)\right]\\
    =~& \E_e[\hat{g}(S_1)\, \tilde{\Phi}_1] + \frac{\Pr(o)}{\Pr(e)}\E_{o}\left[\E\left[\bar{Y} - \hat{g}(S_1)\mid S_1\right]\, \hat{q}(S_1)\right] \tag{tower-law and causal graph}\\
    =~& \E_e[\hat{g}(S_1)\, \tilde{\Phi}_1] + \frac{\Pr(o)}{\Pr(e)}\E_{o}\left[(g_0(S_1) - \hat{g}(S_1))\, \hat{q}(S_1)\right] \tag{definition of $g_0$}\\
    =~& \frac{\Pr(o)}{\Pr(e)}\E_o[\hat{g}(S_1)\, q_0(S_1)] + \frac{\Pr(o)}{\Pr(e)}\E_{o}\left[(g_0(S_1) - \hat{g}(S_1)) \hat{q}(S_1)\right] \tag{Lemma~\ref{lem:transform}}\\
    =~& \frac{\Pr(o)}{\Pr(e)} \E_o[g_0(S_1)\, q_0(S_1)] + \frac{\Pr(o)}{\Pr(e)}\,\E_{o}\left[(g_0(S_1) - \hat{g}(S_1)) (\hat{q}(S_1) - q_0(S_1))\right]\\
    =~& \E_e[\tilde{\Phi}_1\, g_0(S_1)] + \frac{\Pr(o)}{\Pr(e)}\,\E_{o}\left[(g_0(S_1) - \hat{g}(S_1)) (\hat{q}(S_1) - q_0(S_1))\right] \tag{Lemma~\ref{lem:transform}}\\
    =~& \E_e[\tilde{\Phi}_1 \tilde{\Phi}_1^\top]\theta_0 + \E_{o}\left[(g_0(S_1) - \hat{g}(S_1)) (\hat{q}(S_1) - q_0(S_1))\right]
\end{align*}
Thus we get the desired property, i.e.:
\begin{align}
    \E_e[\tilde{\Phi}_1 \tilde{\Phi}_1^\top] \left(\hat{\theta} - \theta_0\right) = \frac{\Pr(o)}{\Pr(e)} \E_{o}\left[(g_0(S_1) - \hat{g}(S_1)) (\hat{q}(S_1) - q_0(S_1))\right] 
\end{align}
which implies that:
\begin{align}
    \left\|\E_e[\tilde{\Phi}_1 \tilde{\Phi}_1^\top] \left(\hat{\theta} - \theta_0\right)\right\|_2 \leq \frac{\Pr(o)}{\Pr(e)}\sqrt{\E_{o}\left[(g_0(S_1) - \hat{g}(S_1))^2\right]} \cdot \sqrt{\E_o\left[(\hat{q}(S_1) - q_0(S_1))^2\right]} 
\end{align}

\vsdelete{
\subsection{Proof of Theorem~\ref{thm:estimation}}
Employing the general theorem in \cite{chernozhukov2018double} and since our estimator is the solution to the empirical analogue of a Neyman orthogonal moment equation:
\begin{equation}
    m(\theta; f_0) = \E[\psi(Z; \theta, f_0)] = 0
\end{equation}
Thus as long as the nuisance functions $f_0$ are estimated on a separate sample and they satisfy an MSE rate of $o_p(n^{-1/4})$, then the final parameter estimate is asymptotically normal, i.e. $\sqrt{n}\, (\hat{\theta}- \theta_0) \to_d N(0, V)$ with $V = J^{-1} \Sigma J^{-\top}$, where $J=\nabla_{\theta} m(\theta_0; f_0)$ and $\Sigma = \E[\psi(Z; \theta_0, f_0)\, \psi(Z; \theta_0, f_0)^\top]$. We can also easily leverage the double robustness property in Theorem~\ref{thm:ortho} to further relax the requirement on the nuisance functions $g, q$, so as to only require that the product of the two mean squared errors is of $o_p(n^{-1/2})$ as opposed to requiring the stronger property that each individual error is of $o_p(n^{-1/4})$.
\paragraph{Asymptotic variance characterization and estimation.} 
We now analyze and simplify the variance term. First observe that in our setting:
\begin{align}
    J :=~& \E[\tilde{\Phi}_1\, \tilde{\Phi}_1^\top 1\{e\}] = \Pr(e) \E_e[\tilde{\Phi}_1\, \tilde{\Phi}_1^\top] =: \Pr(e)\, J_e
\end{align}
Moreover, if we let:
\begin{align}
    \psi(Z; \theta, f) :=~& 1\{e\}\, \psi_e(S_1, T_1, S_{0}; \theta, g, p, h) + 1\{o\}\, \psi_o(\bar{Y}, S_1, S_{0}; g, q)\\
    :=~& 1\{e\}\, \psi_e(Z_e; \theta, g, p, h) + 1\{o\}\, \psi_o(Z_o; \theta, g, q)\\
    \psi_e(Z_e; \theta, g, p, h) :=~& (g(S_1) - h(S_0) - \theta_0^\top \left(\phi(T_1, S_0) - p(S_0)\right)) \left(\phi(T_1, S_0) - p(S_0)\right)\\
    \psi_o(Z_o; g, q) :=~& q(S_1, S_0)\, (\bar{Y} - g(S_1))
\end{align}
Then we also have that:
\begin{align}
    \Sigma :=~& \E[\psi(Z; \theta_0, f_0)\, \psi(Z; \theta_0, f_0)^{\top}]\\
    =~& \Pr(e)\,\E_e[\psi_e(Z_e; \theta_0, g_0, p_0, h_0)\, \psi_e(Z_e; \theta_0, g_0, p_0, h_0)^{\top}] + \Pr(o)\,\E_o[\psi_o(Z_o; g_0, q_0)\, \psi_o(Z_o; g_0, q_0)^{\top}]\\
    =~& \Pr(e)^2\,\Sigma_e + \Pr(e)^2\Sigma_o \tag{definition of $\Sigma_e, \Sigma_o$}
\end{align}
Thus the variance $V$ can be re-written as desired:
\begin{align}
    V = J^{-1} \Sigma J^{-1} = \frac{1}{\Pr(e)^2}  J_e^{-1}\Sigma J_e^{-1} = J_e^{-1} \left(\Sigma_e + \Sigma_o\right) J_e^{-1}
\end{align}
}

\section{Proofs of Theorems in Section \ref{ssec:ddsurr_wadj}}

\subsection{Proof of Theorem~\ref{thm:non-param-id}}
\begin{proof}
Assumptions~\ref{ass:cond-ex} and \ref{ass:dyn-surrogacy} can be directly read from the SWIG in Figure~\ref{fig:swig}.

For the second part of the theorem, given the assumptions, we can write:
\begin{align*}
    \E_e\left[\bar{Y}^{(t_1, 0)}\right] =~& \E_e\left[\E_e\left[\bar{Y}^{(t_1, 0)} \mid S_0\right]\right] \tag{tower law}\\
    =~& \E_e\left[\E_e\left[\bar{Y}^{(t_1, \underline{0})} \mid S_0, T_1=t_1\right]\right] \tag{dynExog + overlap}\\
    =~& \E_e\left[\E_e\left[\bar{Y}^{(T_1, \underline{0})} \mid S_0, T_1=t_1\right]\right] \tag{consistency}\\
    =~& \E_e\left[\E_e\left[\E_e\left[\bar{Y}^{(T_1, \underline{0})} \mid S_1, S_0, T_1\right]\mid S_0, T_1=t_1\right]\right] \tag{tower law}\\
    =~& \E_e\left[\E_e\left[\E_e\left[\bar{Y}^{(T_1, \underline{0})} \mid S_1\right]\mid S_0, T_1=t_1\right]\right] \tag{dynSurr}\\
    =~& \E_e\left[\E_e\left[\E_e\left[\bar{Y}^{(\underline{0}_2)} \mid S_1\right]\mid S_0, T_1=t_1\right]\right] \tag{consistency}\\
    =~& \E_e\left[\E_e\left[\E_o\left[\bar{Y}^{(\underline{0}_2)} \mid S_1\right]\mid S_0, T_1=t_1\right]\right] \tag{dynIR}
\end{align*}
Moreover, the quantity $\E_o\left[\bar{Y}^{(\underline{0}_2)} \mid S_1\right]$ can also be non-parametrically identified via a recursive formula as follows.
Define: $f_{t,j}(s_{t-1}) = \E_o\left[Y_j^{(\bar{T}_{t}, \underline{0}_{t+1})} \mid S_{t-1}=s_{t-1}, T_t=0\right]$, then for any $j > t \geq 1$, we have the recursion:
\begin{align*}
    f_{t,j}(s_{t-1}) =~& \E_o\left[Y_j^{(\bar{T}_{t}, \underline{0}_{t+1})} \mid S_{t-1}=s_{t-1}, T_t=0\right]\\
    =~& \E_o\left[\E_o\left[Y_j^{(\bar{T}_{t}, \underline{0}_{t+1})}\mid S_{t}, S_{t-1}, T_t\right] \mid S_{t-1}=s_{t-1}, T_t=0\right] \tag{tower law}\\
    =~& \E_o\left[\E_o\left[Y_j^{(\bar{T}_{t},\underline{0}_{t+1})}\mid S_{t}\right] \mid S_{t-1}=s_{t-1}, T_t=0\right] \tag{dynSurr}\\
    =~& \E_o\left[\E_o\left[Y_j^{(\bar{T}_{t}, \underline{0}_{t+1})}\mid S_{t}, T_{t+1}=0\right] \mid S_{t-1}=s_{t-1}, T_t=0\right] \tag{dynExog + overlap}\\
    =~& \E_o\left[\E_o\left[Y_j^{(\bar{T}_{t+1}, \underline{0}_{t+2})}\mid S_{t}, T_{t+1}=0\right] \mid S_{t-1}=s_{t-1}, T_t=0\right] \tag{consistency}\\
    =~& \E_o\left[f_{t+1, j}(S_{t})\mid S_{t-1}=s_{t-1}, T_t=0\right] + \E_o[\\
\end{align*}
Moreover, note that for any $j\geq 2$:
\begin{align*}
    f_{j,j}(s_{j-1}) =~& \E_o\left[Y_j^{(\bar{T}_{j}, \underline{0}_{j+1})} \mid S_{j-1}=s_{j-1}, T_j=0\right]\\
    =~& \E_o\left[Y_j^{(\bar{T}_{j})} \mid S_{j-1}=s_{j-1}, T_j=0\right] = \E[Y_j \mid S_{j-1}=s_{j-1}, T_j=0] \tag{base case identification}
\end{align*}
Thus for any $j\geq 2$, we have that $f_{j,j}(s_{j-1})$ is identified via the above equation and that by induction, if $f_{t+1,j}$ has been identified, then $f_{t,j}$ is identified in terms of $f_{t+1,j}$, via the recurvise equation:
\begin{align}
    f_{t,j}(s_{t-1}) =~& \E_o\left[f_{t+1, j}(S_{t})\mid S_{t-1}=s_{t-1}, T_t=0\right] \tag{recursive identification}
\end{align}
Thus $f_{t,j}$ are identified for any $M\geq j\geq t \geq 2$.
Finally note that:
\begin{align*}
    \E_o\left[\bar{Y}^{(\underline{0}_2)} \mid S_1\right] =~& \sum_{j=1}^M \E_o\left[Y_j^{(\underline{0}_2)} \mid S_1\right]\\
    =~& \E_o[Y_1\mid S_1] + \sum_{j=2}^M \E_o\left[Y_j^{(\underline{0}_2)} \mid S_1\right] \tag{consistency}\\
    =~& \E_o[Y_1\mid S_1] + \sum_{j=2}^M \E_o\left[Y_j^{(\underline{0}_2)} \mid S_1, T_2=0\right] \tag{dynExog + overlap}\\
    =~& \E_o[Y_1\mid S_1] + \sum_{j=1}^M f_{2,j}(S_1)
\end{align*}
Thus since all $f_{2,j}$ are identified by the recursive argument, we also have that the quantity $\E_o\left[\bar{Y}^{(\underline{0}_2)} \mid S_1\right]$ is non-parametrically identified. This concludes the proof of the theorem.
\end{proof}

\begin{proof}
We remind that our goals is to show that any quantity $\E_e\left[Y_m^{(\tau, \underline{0}_2)}\right]$ for any $m\in [M]$ and any $\tau\in \mcT_1$ can be represented as:
\begin{align*}
    R_m(\tau) := \E_e\left[f_{1,m}(\tau, S_0) + W_0(\tau)\, (f_{2,m}(S_1) - f_{1,m}(\tau, S_0))\right] + \E_o\left[\sum_{t=2}^m \bar{W}_t\, \left(f_{t+1,m}(S_{t}) - f_{t,m}(S_{t-1})\right)\right]
\end{align*}
Moreover, we want to show that the latter is a valid representation if either i) for all $t\in [m]$, $f_{t,m}$ are correct, or ii) for all $t\in [m]$, $W_{t}$ are correct.

First note that the proof of Theorem~\ref{thm:non-param-id} shows that $\E_e\left[Y_m^{(\tau,\underline{0}_2)}\right]=\E_e\left[f_{1,m}(S_0)\right]$. Now, suppose that the regression functions $f_{t,m}$ are correct. Then we have:
\begin{align*}
    \E_e\left[W_0(\tau) (f_{2,m}(S_1) - f_{1,m}(\tau, S_0))\right] =~& \E_e\left[W_0(\tau) \E_e\left[f_{2,m}(S_1) - f_{1,m}(\tau, S_0)\mid T_1, S_0\right]\right]\\
    =~& \E_e\left[W_0(\tau) \E_e\left[f_{2,m}(S_1) - f_{1,m}(\tau, S_0)\mid T_1=\tau, S_0\right]\right] = 0\\
    \E_o\left[\bar{W}_t \left(f_{t+1,m}(S_{t}) - f_{t,m}(S_{t-1})\right)\right] =~& \E_o\left[\bar{W}_t\, \E_o\left[f_{t+1,m}(S_{t}) - f_{t,m}(S_{t-1})\mid \bar{T}_t, \bar{S}_{t-1}\right]\right]\\
    =~& \E_o\left[\bar{W}_t\, \E_o\left[f_{t+1,m}(S_{t}) - f_{t,m}(S_{t-1})\mid \bar{T}_t=0, \bar{S}_{t-1}\right]\right]\\
    =~& \E_o\left[\bar{W}_t\, \E_o\left[f_{t+1,m}(S_{t}) - f_{t,m}(S_{t-1})\mid T_t=0, S_{t-1}\right]\right] = 0
\end{align*}
Thus all these difference terms vanish and the representation takes the form:
\begin{align*}
    R_m(\tau) = \E_e\left[f_{1,m}(\tau, S_0)\right] = \E_e\left[Y^{(\tau, \underline{0}_2)}\right]
\end{align*}

Now suppose that all the models that enter the propensity weights are correct. Then we have:
\begin{align*}
    \E_e[W_0(\tau) f_{1,m}(\tau, S_0)] =~& \E_e[\E_e[W_0(\tau)\mid S_0] f_{1,m}(\tau, S_0)] = \E_e\left[ \E_e\left[\frac{1\{T_1=\tau\}}{\Pr_e(T_1=\tau\mid S_0)}\mid S_0\right] f_{1,m}(\tau, S_0)\right]\\
    =~& \E_e\left[f_{1,m}(\tau, S_0)\right]\\
    \E_o[\bar{W}_2 f_{t,m}(S_{t-1})] =~& \E_o\left[W_1(\tau) \frac{1\{T_t=0\}}{\Pr(T_2=0\mid S_{1})} f_{2,m}(S_1)\right]\\
    =~& \E_o\left[W_1(\tau) \E_o\left[\frac{1\{T_t=0\}}{\Pr(T_2=0\mid S_{1})}\mid S_1\right] f_{2,m}(S_1)\right]\\
    =~& \E_o\left[\frac{\Pr(S_1\mid e)}{\Pr(S_1\mid o)} \E_e\left[W_0(\tau)\mid S_1\right] f_{2,m}(\tau, S_1)\right]\\
    =~& \E_e\left[ \E_e\left[W_0(\tau)\mid S_1\right] f_{2,m}(\tau, S_1)\right]\\
    =~& \E_e\left[W_0(\tau) f_{2,m}(\tau, S_1)\right]\\
    \E_o[\bar{W}_t f_{t,m}(S_{t-1})] =~& \E_o\left[\bar{W}_{t-1} \frac{1\{T_t=0\}}{\Pr(T_t=0\mid S_{t-1}, T_{t-1}=0)} f_{t,m}(S_{t-1})\right]\\
    =~& \E_o\left[\bar{W}_{t-1} \E_o\left[\frac{1\{T_t=0\}}{\Pr(T_t=0\mid S_{t-1}, T_{t-1}=0)}\mid \bar{S}_{t-1}, \bar{T}_{t-1}\right] f_{t,m}(S_{t-1})\right]\\
    =~& \E_o\left[\bar{W}_{t-1} \E_o\left[\frac{1\{T_t=0\}}{\Pr(T_t=0\mid S_{t-1}, T_{t-1}=0)}\mid S_{t-1}, {T}_{t-1}\right] f_{t,m}(S_{t-1})\right]\\
    =~& \E_o\left[\bar{W}_{t-1} \E_o\left[\frac{1\{T_t=0\}}{\Pr(T_t=0\mid S_{t-1}, T_{t-1}=0)}\mid S_{t-1}, {T}_{t-1}=0\right] f_{t,m}(S_{t-1})\right]\\
    =~& \E_o\left[\bar{W}_{t-1} f_{t,m}(S_{t-1})\right]
\end{align*}
Thus all the negative terms for each $t$, cancel with the positive terms from the term for $t-1$ and the representation simplifies to:
\begin{align*}
    R_m(\tau) = \E_o\left[ \bar{W}_m Y_m \right]
\end{align*}

We now show that the latter is also equal to $\E_e\left[Y_m^{(\tau, \underline{0}_2)}\right]$.
\begin{align*}
    \E_o\left[ \bar{W}_m Y_m \right] =~& \E_o\left[ W_1(\tau) \Pi_{t=2}^{m} W_t Y_m\right]\\
    =~& \E_o\left[ W_1(\tau) \Pi_{t=2}^{m} W_t Y_m^{(\bar{T}_m)}\right]\\
    =~& \E_o\left[ W_1(\tau) \Pi_{t=2}^{m-1} W_t \E_o\left[W_m\, Y_m^{(\bar{T}_m)}\mid \bar{S}_{m-1}, \bar{T}_{m-1}\right]\right]\\
    =~& \E_o\left[ W_1(\tau) \Pi_{t=2}^{m-1} W_t \E_o\left[\frac{1\{T_m=0\}}{\Pr(T_m=0\mid S_{m-1}, T_{m-1})}\, Y_m^{(\bar{T}_{m-1}, 0)}\mid \bar{S}_{m-1}, \bar{T}_{m-1}=0\right]\right]\\
    =~& \E_o\left[ W_1(\tau) \Pi_{t=2}^{m-1} W_t \E_o\left[\frac{1\{T_m=0\}}{\Pr(T_m=0\mid S_{m-1}, T_{m-1})}\mid \bar{S}_{m-1}, \bar{T}_{m-1}=0\right]\, \E_o\left[Y_m^{(\bar{T}_{m-1}, 0)}\mid \bar{S}_{m-1}, \bar{T}_{m-1}=0\right]\right]\\
    =~& \E_o\left[ W_1(\tau) \Pi_{t=2}^{m-1} W_t \E_o\left[\frac{1\{T_m=0\}}{\Pr(T_m=0\mid S_{m-1}, T_{m-1})}\mid {S}_{m-1}, {T}_{m-1}=0\right]\, \E_o\left[Y_m^{(\bar{T}_{m-1}, 0)}\mid \bar{S}_{m-1}, \bar{T}_{m-1}=0\right]\right]\\
    =~& \E_o\left[ W_1(\tau) \Pi_{t=2}^{m-1} W_t \E_o\left[Y_m^{(\bar{T}_{m-1}, 0)}\mid \bar{S}_{m-1}, \bar{T}_{m-1}=0\right]\right]\\
    =~& \E_o\left[ W_1(\tau) \Pi_{t=2}^{m-1} W_t Y_m^{(\bar{T}_{m-1}, 0)}\right]
\end{align*}
Repeating the above process, we can remove the term $\Pi_{t=1}^{m-1} W_t$ from the above expression, every time fixing each $T_t$ to zero in the $Y_m$ counterfactual, i.e. that:
\begin{align*}
    R_m(\tau) = \E_o\left[ \bar{W}_m Y_m \right] =~& \E_o\left[ W_1(\tau) Y_m^{(\underline{0}_2)}\right]\\
    =~& \E_o\left[ W_1(\tau) \E_o\left[Y_m^{(\underline{0}_2)}\mid S_1\right]\right]\\
    =~& \E_o\left[ W_1(\tau) \E_e\left[Y_m^{(\underline{0}_2)}\mid S_1\right]\right]\\
    =~& \E_o\left[\frac{\Pr(S_1\mid e)}{\Pr(S_1\mid o)} \E_e\left[W_0(\tau)\mid S_1\right]\E_e\left[Y_m^{(\underline{0}_2)}\mid S_1\right]\right]\\
    =~& \E_e\left[\E_e\left[W_0(\tau)\mid S_1\right]\E_e\left[Y_m^{(\underline{0}_2)}\mid S_1\right]\right]\\
    =~& \E_e\left[\E_e\left[\frac{1\{T_1=\tau\}}{\Pr(T_1=\tau\mid S_0)}\mid S_1\right]\E_e\left[Y_m^{(\underline{0}_2)}\mid S_1\right]\right]\\
    =~& \E_e\left[\E_e\left[\frac{1\{T_1=\tau\}}{\Pr(T_1=\tau\mid S_0)}\, Y_m^{(\underline{0}_2)}\mid S_1\right]\right]\\
    =~& \E_e\left[\frac{1\{T_1=\tau\}}{\Pr(T_1=\tau\mid S_0)}\, Y_m^{(\underline{0}_2)}\right]\\
    =~& \E_e\left[\frac{1\{T_1=\tau\}}{\Pr(T_1=\tau\mid S_0)}\, Y_m^{(\tau, \underline{0}_2)}\right]\\
    =~& \E_e\left[\E_e\left[\frac{1\{T_1=\tau\}}{\Pr(T_1=\tau\mid S_0)}\, Y_m^{(\tau, \underline{0}_2)}\mid S_0\right]\right]\\
    =~& \E_e\left[\E_e\left[\frac{1\{T_1=\tau\}}{\Pr(T_1=\tau\mid S_0)}\mid S_0\right]\, \E_e\left[Y_m^{(\tau, \underline{0}_2)}\mid S_0\right]\right]\\
    =~& \E_e\left[\E_e\left[Y_m^{(\tau, \underline{0}_2)}\mid S_0\right]\right]\\
    =~& \E_e\left[Y_m^{(\tau, \underline{0}_2)}\right]
\end{align*}
\end{proof}

\section{Proofs from Section~\ref{sec:dsurr_semipar}}

\subsection{Proof of Lemma~\ref{lem:cntf-char-rho}}

\begin{proof}
First, note that for any $m\in [1,M]$, by the definition of the observed $Y_m$, we have that:
\begin{align}
    Y_m \equiv Y_m^{(\bar{T}_j)}
\end{align}
and that for any treatment sequence $\bar{\tau}_m$, via a telescoping sum argument, we can write:
\begin{align}
Y_m^{(\tau)} - Y_m^{(\bar{\tau}_{t-1}, \underline{0}_t)} = \sum_{j=t}^m Y_m^{(\bar{\tau}_j, \underline{0}_{j+1})} - Y_m^{(\bar{\tau}_{j-1}, \underline{0}_j)}
\end{align}
Thus, applying linearity of expectation, the tower law of expectations and the definition of the blip functions, we have:
\begin{align*}
    \E\left[Y_m - Y_m^{(\bar{T}_{t-1}, \underline{0}_t)} \mid S_{t-1}, T_{t}\right] =~& 
    \E\left[Y_m^{(T)} - Y_m^{(\bar{T}_{t-1}, \underline{0}_t)} \mid S_{t-1}, T_{t}\right]\\
    =~& 
    \E\left[\sum_{j=t}^m Y_m^{(\bar{T}_j, \underline{0}_{j+1})} - Y_m^{(\bar{T}_{j-1}, \underline{0}_j)}\mid S_{t-1}, T_{t}\right]\\
    =~& 
    \sum_{j=t}^m \E\left[Y_m^{(\bar{T}_j, \underline{0}_{j+1})} - Y_m^{(\bar{T}_{j-1}, \underline{0}_j)}\mid S_{t-1}, T_{t}\right]\\
    =~& 
    \sum_{j=t}^m \E\left[\E\left[Y_m^{(\bar{T}_j, \underline{0}_{j+1})} - Y_m^{(\bar{T}_{j-1}, \underline{0}_j)} \mid S_{j-1}, S_{t-1}, T_t\right]\mid S_{t-1}, T_{t}\right]\\
    =~& 
    \sum_{j=t}^m \E\left[\E\left[Y_m^{(\bar{T}_j, \underline{0}_{j+1})} - Y_m^{(\bar{T}_{j-1}, \underline{0}_j)} \mid S_{j-1}, T_{j}\right]\mid S_{t-1}, T_{t}\right] \tag{dynSurr + dynExog}\\
    =~& 
    \sum_{j=t}^m \E\left[\gamma_{j,m}(T_j, S_{j-1})\right]\\
    =~& 
    \E\left[\sum_{j=t}^m \gamma_{j,m}(T_j, S_{j-1})\right]
\end{align*}
By re-arranging we conclude the desired property:
\begin{align*}
    \E\left[Y_m^{(\bar{T}_{t-1}, \underline{0}_t)} \mid S_{t-1}, T_{t}\right] = \E\left[Y_m - \sum_{j=t}^m \gamma_{j,m}(T_j, S_{j-1})\mid S_{t-1}, T_{t}\right]
\end{align*}
\end{proof}

\subsection{Proof of Lemma~\ref{lem:moment-restrictions}}
\begin{proof}
For simplicity of notation, for any $f\in \mcF$, let:
\begin{align}
    \bar{f}(T_t, S_{t-1}) = f(T_t, S_{t-1}) - \E[f(T_t, S_{t-1})\mid S_{t-1}]
\end{align}
and observe that by the definition of $\bar{f}$, we crucially have that for any $f\in \mcF$:
\begin{align}
\E\left[\bar{f}(T_t, S_{t-1})\mid S_{t-1}\right] = 0.
\end{align}

By Lemma~\ref{lem:cntf-char-rho}, we have that at the true $\psi^*$:
\begin{align}
    \E\left[ H_t(\psi^*)\, \bar{f}(T_t, S_{t-1})\mid S_{t-1}\right]=~& \E\left[ \E\left[H_t(\psi^*) \mid S_{t-1}, T_t\right]\, \bar{f}(T_t, S_{t-1})\mid S_{t-1}\right] \tag{tower law}\\
    =~& \E\left[ \E\left[Y^{(\bar{T}_{t-1}, \underline{0}_t)} \mid S_{t-1}, T_t\right]\, \bar{f}(T_t, S_{t-1}) \mid S_{t-1}\right] \tag{Equation~\eqref{eqn:cntf-rho}}\\
    =~& \E\left[ \E\left[Y^{(\bar{T}_{t-1}, \underline{0}_t)} \mid S_{t-1}\right]\, \bar{f}(T_t, S_{t-1}) \mid S_{t-1}\right] \tag{dynExog}\\
    =~& \E\left[Y^{(\bar{T}_{t-1}, \underline{0}_t)} \mid S_{t-1}\right]\, \E\left[\bar{f}(T_t, S_{t-1})\mid S_{t-1}\right] \tag{tower law}\\
    =~& 0 \nonumber
\end{align}
\end{proof}

\subsection{Proof of Theorem~\ref{thm:dyn-id}}

Observe that under the causal graph assumption in Figure~\ref{fig:simplified_cg}, then we have that the sequential conditional exogeneity Assumption~\ref{ass:cond-ex} is satified. Moreover, by the linear blip model Assumption~\ref{ass:linear-snmm} and by sequential conditional exogeneity:
\begin{align*}
    \tau(t_1, t_0) :=~& \E_e\left[\bar{Y}^{(t_1, \underline{0})} - \bar{Y}^{(t_0, \underline{0})}\right] = \E_e\left[\E_e\left[\bar{Y}^{(t_1, \underline{0})} - \bar{Y}^{(\underline{0})}\mid S_0\right] - \E_e\left[\bar{Y}^{(t_0, \underline{0})} - \bar{Y}^{(\underline{0})}\mid S_0\right]\right]\\
    =~& \E_e\left[\E_e\left[\bar{Y}^{(t_1, \underline{0})} - \bar{Y}^{(\underline{0})}\mid S_0, T_1=t_1\right] - \E_e\left[\bar{Y}^{(t_0, \underline{0})} - \bar{Y}^{(\underline{0})}\mid S_0, T_1=t_1\right]\right]\\
    =~& \sum_{j=1}^M \theta_{e,1,j}^\top\,\E\left[\phi_{e,1}(t_1, S_0) - \phi_{e,1}(t_0, S_0)\right] = \theta_0^\top\, \E\left[\phi_{e,1}(t_1, S_0) - \phi_{e,1}(t_0, S_0)\right]
\end{align*}
where:
\begin{align*}
\theta_0 := \sum_{j=1}^M \theta_{e,1,j}.
\end{align*}
Which proves Equation~\eqref{eqn:dyn-target-param}. Thus it suffices to identify $\theta_0$, in order to identify the treatment effect of interest.

Next we prove an intermediate lemma that essentially argues that the parameter $\theta_0$ can be identified in a manner almost identical to the one in the non-dynamic case, albeit with a small change of the variable $\bar{Y}$ being replaced by $\bar{Y}^{o, \adj}$, where for any environment $d\in\{e, o\}$, we define the adjusted long-term outcome as:
\begin{align*}
    Y_j^{d,\adj} :=~& Y_j - \sum_{q=2}^{j} \gamma_{d, q, j}({T}_q, {S}_{q-1}) & \bar{Y}^{d,\adj} :=~& \sum_{j=1}^M Y_j^{d,\adj}
\end{align*}
\begin{lemma}
Suppose that the data generating processes in the experimental and observational setting satisfy sequential conditional exogeneity Assumption~\ref{ass:cond-ex}, the linear blip function Assumption~\ref{ass:linear-snmm} and the dynamic invariance Assumption~\ref{ass:dynIR}. Then the following conditions hold:
\begin{align}
    \E_e\left[\bar{Y}^{e,\adj}\mid T_1, S_0\right] =~& \E[\bar{Y}^{(T_1,\underline{0})}\mid T_1, S_0] =\theta_0^\top\phi_{1}(T_1, S_0) + b_0(S_0) \tag{DynPLR}\label{eqn:ass:dynplr}\\
    \E_e\left[\bar{Y}^{e,\adj}\mid T_1, S_0\right] =~& \E_o\left[\bar{Y}^{o,\adj}\mid T_1, S_0\right] \tag{DynIR} \label{eqn:ass:dynIR}\\
    \bar{Y}^{o,\adj} \indep_{\mean}& (T_1, S_0) \mid S_1, o ~~\text{and}~~~ \bar{Y}^{e,\adj} \indep_{\mean}(T_1, S_0) \mid S_1, e \tag{DynMeanID} \label{eqn:ass:dyn-meanid}
\end{align}
\end{lemma}

\begin{proof}
First, note that by the definition of the blip functions:
\begin{align*}
    \E_e\left[\bar{Y}^{(T_1, \underline{0})}\mid T_1, S_0\right] = \theta_0^\top\phi_{t}(T_1, S_0) + \E\left[\bar{Y}^{(\underline{0})}\mid S_0, T_1\right] = \theta_0^\top\phi_{t}(T_1, S_0) + \E\left[\bar{Y}^{(\underline{0})}\mid S_0\right] =:  \theta_0^\top\phi_{t}(T_1, S_0) + b_0(S_0)
\end{align*}
Moreover, note that by the causal graph assumption the random variable $\bar{Y}^{(T_1, 0)}$, satisfies the surrogate condition:
\begin{align}\label{eqn:cntf-indep}
    \bar{Y}^{(T_1, 0)} \indep (T_1, S_0) \mid S_1
\end{align}
This conditional independence can be easily verified from the single world intervention graph (SWIG) of the intervention $(T_1, \underline{0})$ as depicted in Figure~\ref{fig:swig1}.
\begin{figure}[htpb]
    \centering
    \includegraphics[scale=.5]{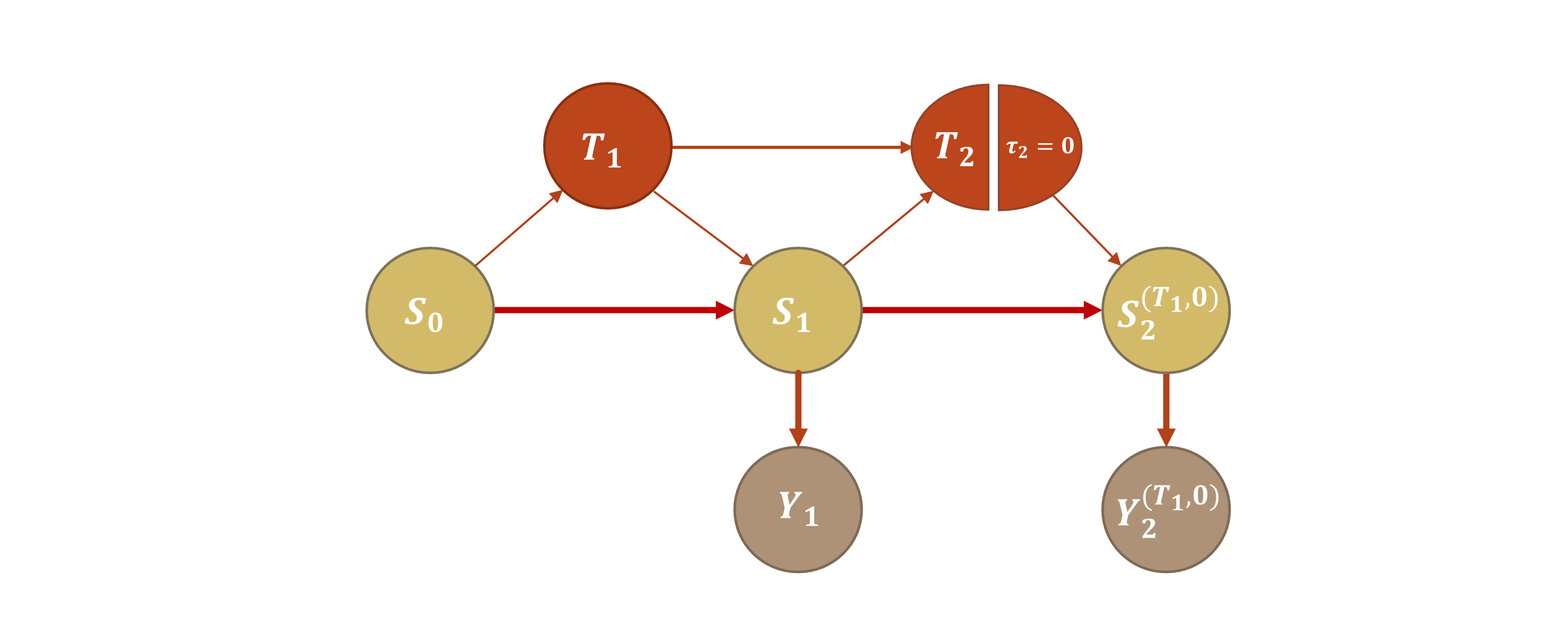}
    \caption{Single World Intervention Graph (SWIG), for the intervention that zeros out all treatments beyond the first period.}
    \label{fig:swig1}
\end{figure}
Thus if we define:
\begin{align}
    g_{\adj}^*(S_1) := \E_e\left[\bar{Y}^{(T_1,\underline{0})}\mid S_1\right]
\end{align}
then we can write:
\begin{align*}
    \E_e\left[\bar{Y}^{(T_1, \underline{0})}\mid T_1, S_0\right] = \E_e\left[\E_e\left[\bar{Y}^{(T_1, \underline{0})}\mid S_1, T_1, S_0\right]\mid T_1, S_0\right] = \E_e\left[\E_e\left[\bar{Y}^{(T_1, \underline{0})}\mid S_1\right]\mid T_1, S_0\right] = \E_e[g_{\adj}^*(S_1)\mid T_1, S_0]
\end{align*}
Observe that by consistency, the counterfactual outcome $Y^{(T_1, \underline{0})}$ can be written as a counterfactual outcome of interventions in periods $\{2,\ldots,M\}$:
\begin{align*}
    \bar{Y}^{(T_1, \underline{0})} = \sum_{j=1}^m Y_j^{(T_1, \underline{0}_{2})} = Y_1 + \sum_{j=2}^m Y_j^{(\underline{0}_2)} \equiv \bar{Y}^{(\underline{0}_2)}
\end{align*}

\paragraph{Proof of Equation~\eqref{eqn:ass:dynIR}.} Now observe that by the dynamic invariance Assumption~\ref{ass:dynIR}, we have that:
\begin{align*}
    g_{\adj}^*(S_1) :=~& \E_e\left[\bar{Y}^{(\underline{0}_2)}\mid S_1\right] = \E_o\left[\bar{Y}^{(\underline{0}_2)}\mid S_1\right]
\end{align*}
For any environment $d\in \{o,e\}$, if we had access to long-term outcomes, then the quantity $\E_d[\bar{Y}^{(0)}\mid S_1]$, is identified via the estimand given by Lemma~\ref{lem:cntf-char-rho}:
\begin{align*}
    \E_d\left[\bar{Y}^{(\underline{0}_2)}\mid S_1, T_2\right] =~& \sum_{j=1}^M \E_d\left[Y_j^{({T}_{1}, \underline{0}_2)}\mid {S}_{1}, T_2\right]\\
    =~& \sum_{j=1}^M \E_d\left[Y_j - \sum_{q=2}^{j} \gamma_{d, q, j}(T_q, {S}_{q-1}) \mid {S}_{1}, T_2\right]\\
    =~& \sum_{j=1}^M \E_d\left[Y_j^{d,\adj} \mid {S}_{1}, T_2\right] = \E_d\left[\bar{Y}^{d,\adj}\mid S_1, T_2\right]
\end{align*}
Therefore:
\begin{align*}
    g_{\adj}^*(S_1) = \E_o\left[\bar{Y}^{(\underline{0}_2)}\mid S_1\right] = \E_o\left[\bar{Y}^{o,\adj}\mid S_1\right] = \E_e\left[\bar{Y}^{(\underline{0}_2)}\mid S_1\right] = \E_e\left[\bar{Y}^{e,\adj}\mid S_1\right]  
\end{align*}
which proves Equation~\eqref{eqn:ass:dynIR}.

\paragraph{Proof of Equation~\eqref{eqn:ass:dyn-meanid}.} Furthermore, we note that $\bar{Y}^{o,\adj}$ and $\bar{Y}^{e,\adj}$ satisfy the mean conditional independence assumption, i.e. for any $d\in \{e,o\}$:
\begin{align*}
    \bar{Y}^{d,\adj} \indep_{\mean} (S_0, T_1) \mid S_1, d
\end{align*}
since we first note that by the causal graph assumption $\{Y_j, \underline{S}_1, \underline{T}_2\} \indep (S_0, T_1) \mid S_1, T_2, d$ and therefore:
\begin{align*}
    \E_d[\bar{Y}^{d,\adj}\mid S_1, T_1, S_0] =~& \E_d\left[\E_d\left[\bar{Y}^{d,\adj}\mid  S_1, T_2, T_1, S_0\right] \mid  S_1, T_1, S_0\right]\\
    =~& \E_d\left[\E_d\left[\bar{Y}^{d,\adj}\mid  S_1, T_2\right] \mid  S_1, T_1, S_0\right]\\
    =~& \E_d\left[\E_d\left[\bar{Y}^{(\underline{0}_2)}\mid  S_1, T_2\right] \mid  S_1, T_1, S_0\right]\\
    =~& \E_d\left[\E_d\left[\bar{Y}^{(\underline{0}_2)}\mid  S_1, T_2, T_1, S_0\right] \mid  S_1, T_1, S_0\right]\\
    =~& \E_d\left[\E_d\left[\bar{Y}^{(\underline{0}_2)}\mid  S_1, T_1, S_0\right] \mid  S_1, T_1, S_0\right]\\
    =~& \E_d\left[\bar{Y}^{(\underline{0}_2)}\mid  S_1, T_1, S_0\right] =\E_d\left[\bar{Y}^{(\underline{0}_2)}\mid  S_1\right]= \E_d\left[\bar{Y}^{d,\adj}\mid S_1\right]
\end{align*}
where in second-to-last equation we used the conditional independence from Equation~\eqref{eqn:cntf-indep}. This proves Equation~\eqref{eqn:ass:dyn-meanid}.

\paragraph{Proof of Equation~\eqref{eqn:ass:dynplr}.} Finally, note that the above sequence of inequalities also implies that for any environment $d\in\{e, o\}$, we have that $\E_d[\bar{Y}^{d,\adj}\mid S_1, T_1, S_0]=\E_d\left[\bar{Y}^{(\underline{0}_2)}\mid  S_1, T_1, S_0\right]$. Thus we can also write:
\begin{align*}
    \E_d[\bar{Y}^{d,\adj}\mid T_1, S_0] =~& \E_d\left[\E_d[\bar{Y}^{d,\adj}\mid S_1, T_1, S_0]\mid T_1, S_0\right]\\
    =~& \E_d\left[\E_d\left[\bar{Y}^{(\underline{0}_2)}\mid  S_1, T_1, S_0\right]\mid T_1, S_0\right]\\
    =~& \E_d\left[\bar{Y}^{(\underline{0}_2)}\mid T_1, S_0\right]\\
    =~& \E_d\left[\bar{Y}^{(T_1, \underline{0})}\mid T_1, S_0\right]\\
    =~& \theta_0^\top\phi_1(T_1, S_0) + b_0(S_0)
\end{align*}
which proves Equation~\eqref{eqn:ass:dynplr} and completes the proof of the lemma.
\end{proof}

Thus we see that the setting satisfies all properties \eqref{eqn:ass:plr}, \eqref{eqn:ass:ir}, \eqref{eqn:ass:meanid}, defined in the previous section, but with $\bar{Y}^{d,\adj}$ in the place of $\bar{Y}$. Hence, Theorem~\ref{thm:id} and Theorem~\ref{thm:ortho} hold verbatim, by simply replacing $\bar{Y}$ with $\bar{Y}^{d,\adj}$. We provide these derivations here for completeness.

\paragraph{Surrogate index representation.} By Equations~\eqref{eqn:ass:dyn-meanid} and \eqref{eqn:ass:dynIR} we have:
\begin{align*}
    \E_e[\bar{Y}^{e,\adj} \mid T_1, S_0] =~& \E_e[\E_e[\bar{Y}^{e,\adj}\mid S_1, T_1, S_0]\mid T_1, S_0] = \E_e[\E_e[\bar{Y}^{e,\adj}\mid S_1]\mid T_1, S_0] = \E_e[\E_o[\bar{Y}^{o,\adj}\mid S_1]\mid T_1, S_0]\\
    =~& \E_e[g_{\adj}^*(S_1)\mid T_1, S_0]
\end{align*}
with $g_{\adj}^*$ as defined in Equation~\eqref{eqn:defns-main-thm}.
Thus also we have that $\E_e[\bar{Y}^{e,\adj} \mid S_0]=\E_e[g_{\adj}^*(S_1)\mid S_0]$ and that:
\begin{align*}
    \E_e[\bar{Y}^{e,\adj} \mid T_1, S_0] - \E_e[\bar{Y}^{e,\adj} \mid S_0] = \E_e[g_{\adj}^*(S_1)\mid T_1, S_0] - \E_e[g_{\adj}^*(S_1)\mid S_0] = \E_e[\tilde{g}_{\adj}^*(S_1, S_0)\mid T_1, S_0]
\end{align*}

By Equation~\eqref{eqn:ass:dynplr} and the definition of $\tilde{\Phi}_1 := \phi(T_1,S_0) - \E[\phi(T_1,S_0)\mid S_0]$, we have that:
\begin{align*}
    \E_e[\bar{Y}^{e,\adj}\mid T_1, S_0] - \E_e[\bar{Y}^{e,\adj}\mid S_0] = \theta_0^{\top} \tilde{\Phi}_1 \implies~& \E_e[\tilde{g}_{\adj}^*(S_1,S_0) - \theta_0^{\top} \tilde{\Phi}_1\mid T_1, S_0] = 0\\
    \implies~&  \E_e\left[\left(\tilde{g}_{\adj}^*(S_1, S_0) - \theta_0^{\top} \tilde{\Phi}_1\right) \tilde{\Phi}_1\right] = 0 
\end{align*}
Solving for $\theta_0$, yields the surrogate index representation:
\begin{align}
    \theta_0 := \E_e[\tilde{\Phi}_1\, \tilde{\Phi}_1^\top]^{-1} \E_{e}[\tilde{\Phi}_1\tilde{g}_{\adj}^*(S_1, S_0)]
\end{align}

\paragraph{Surrogate score representation.} We start by the definition of the surrogate score representation and use Equations~\eqref{eqn:ass:dyn-meanid}, \eqref{eqn:ass:dynIR} and \eqref{eqn:ass:dynplr} to derive:
\begin{align*}
    I_{ss} :=~& \E_{e}[\tilde{\Phi}_1 \tilde{\Phi}_1^\top]^{-1}
    \E_{o}\left[\frac{\Pr(e\mid S_1)}{\Pr(o\mid S_1)} \frac{\Pr(o)}{\Pr(e)}\, \bar{Y}^{o,\adj}\, \E_{e}[\tilde{\Phi}_1\mid S_{1}]\right]\\
    =~& \E_{e}[\tilde{\Phi}_1 \tilde{\Phi}_1^\top]^{-1}
    \E_{o}\left[\frac{\Pr(S_1\mid e)}{\Pr(S_1\mid o)}\, \bar{Y}^{o,\adj}\,  \E_{e}[\tilde{\Phi}_1\mid S_{1}]\right] \tag{Bayes-rule}\\
    =~& \E_{e}[\tilde{\Phi}_1 \tilde{\Phi}_1^\top]^{-1}
    \E_{o}\left[\frac{\Pr(S_1\mid e)}{\Pr(S_1\mid o)}\, \E_o[\bar{Y}^{o,\adj}\mid S_1]\, \E_e[\tilde{\Phi}_1\mid S_1, S_{0}]\right] \tag{tower-law}\\
     =~& \E_{e}[\tilde{\Phi}_1 \tilde{\Phi}_1^\top]^{-1}\,
    \E_{e}\left[\E_o[\bar{Y}^{o,\adj} \mid S_1]\, \E_e[\tilde{\Phi}_1\mid S_1]\right] \tag{change of measure}\\
    =~& \E_{e}[\tilde{\Phi}_1 \tilde{\Phi}_1^\top]^{-1}
    \E_{e}\left[\E_e[\bar{Y}^{e,\adj} \mid S_1]\, \E_e[\tilde{\Phi}_1\mid S_1]\right] \tag{\eqref{eqn:ass:dynIR} assumption}\\
    =~& \E_{e}[\tilde{\Phi}_1 \tilde{\Phi}_1^\top]^{-1}
    \E_{e}\left[\E_e[\bar{Y}^{e,\adj}\, \tilde{\Phi}_1 \mid S_1] \right] \tag{\eqref{eqn:ass:dyn-meanid} assumption}\\
    =~& \E_{e}[\tilde{\Phi}_1 \tilde{\Phi}_1^\top]^{-1}
    \E_{e}\left[\bar{Y}^{e,\adj} \tilde{\Phi}_1 \right] \tag{inverse tower-law}\\
    =~& \E_{e}[\tilde{\Phi}_1 \tilde{\Phi}_1^\top]^{-1}
    \E_{e}\left[\tilde{\Phi}_1 \E_e\left[\bar{Y}^{e,\adj} \mid T_1, S_0\right]\right] \tag{tower-law}\\
    =~& \E_{e}[\tilde{\Phi}_1 \tilde{\Phi}_1^\top]^{-1}
    \left(\E_{e}\left[\tilde{\Phi}_1\, \phi(T_1, S_0)^\top \theta_0 \right] + \E_{e}\left[b_0(S_0)\, \tilde{\Phi}_1 \right]\right) \tag{\eqref{eqn:ass:dynplr} assumption}\\
    =~& \E_{e}[\tilde{\Phi}_1 \tilde{\Phi}_1^\top]^{-1}
    \E_{e}\left[\tilde{\Phi}_1\, \tilde{\Phi}_1^\top \theta_0\right] \tag{mean-zero: $\E[\bar{\Phi}_1 \mid S_0]=0$}\\
    =~& \theta_0
\end{align*}

\paragraph{Orthogonal representation.}
This follows easily by the fact that the second term in the orthogonal representation is mean zero and the first term is equal to $\theta_0$ by the surrogate-index argument.

\subsection{Proof of Theorem~\ref{thm:dyn-ortho}}

First we define for any functional $L(f)$ the Frechet derivative as:
\begin{equation}
    D_f L(f)[\nu] = \frac{\partial}{\partial t} L(f + t\, \nu)\mid_{t=0} 
\end{equation}
Similarly, we can define higher order derivatives, denoted as $D_{g,f} L(f, g)[\mu,\nu]$.

\paragraph{Notation.} We first define some quantities that will be useful throughout the proof. Let:
\begin{align*}
    \tilde{\Phi}_1 :=~& \Phi_1 - \E_e[\Phi_1\mid S_0] & 
    \tilde{\Phi}_{o,t',t} :=~& \Phi_{o,t'} - \E_o[\Phi_{o,t'}\mid S_{t-1}] &
    \tilde{Y}_t :=~& \bar{Y}_t - \E_o[\bar{Y}_t\mid S_{t-1}] \\
    g_{\adj}^*(S_1) :=~& \E_o[Y^{o,\adj}\mid S_1] &
    \tilde{g}_{\adj}^*(S_1, S_0) :=~& g_{\adj}(S_1) - \E_e[g_{\adj}(S_1)\mid S_0] &
\end{align*}

\paragraph{Orthogonal moment formulation.} Observe that we can re-write the moment $m(\theta^*; f^*)=0$ as:
\begin{align*}
    \E\left[1\{e\}\, ({g}_{\adj}^*(S_1) - h^*(S_0) - \theta_0^\top \tilde{\Phi}_1) \tilde{\Phi}_1 + 1\{o\}\, \frac{\Pr(e\mid S_1)}{\Pr(o\mid S_1)} \E_{e}[\tilde{\Phi}_1\mid S_{1}] (\bar{Y}^{o,\adj} - g_{\adj}^*(S_1))  \right] =~& 0\\
    \forall t\in [2,M]: \E\left[1\{o\} \left(\tilde{Y}_{t} - \sum_{\tau=t}^{M} \theta_{o,\tau}^\top \tilde{\Phi}_{o,\tau, t}\right)\, \tilde{\Phi}_{o,t,t}\right] =~& 0
\end{align*}
The second set of equations was shown to hold for the true parameters in \cite{lewis2020doubledebiased} and can be easily verify from Lemma~\ref{lem:moment-restrictions}. Now we also verify that the first moment equation holds. First we see that we can re-write the orthogonal representation from Theorem~\ref{thm:dyn-id} as follows:
\begin{align}
\theta_0 =~& \E_{e}[\tilde{\Phi}_1 \tilde{\Phi}_1^\top]^{-1} \left(\frac{\E[\tilde{g}_{\adj}^*(S_1, S_0) \tilde{\Phi}_1 1\{e\}]}{\Pr(e)} + \frac{1}{\Pr(e)}\E\left[1\{o\}\frac{\Pr(e\mid S_1)}{\Pr(o\mid S_1)} (\bar{Y}^{o,\adj} - g_{\adj}^*(S_1)) \E_{e}[\tilde{\Phi}_1\mid S_{1}, S_{0}]\right]\right)\\
    =~& \E[\tilde{\Phi}_1\, \tilde{\Phi}_1^\top 1\{e\}]^{-1} \left(\E[\tilde{g}_{\adj}^*(S_1, S_0) \tilde{\Phi}_1 1\{e\}]+ \E\left[1\{o\}\frac{\Pr(e\mid S_1)}{\Pr(o\mid S_1)} (\bar{Y}^{o,\adj} - g_{\adj}^*(S_1)) \E_{e}[\tilde{\Phi}_1\mid S_{1}, S_{0}]\right]\right)
\end{align}
Equivalently, the solution to the orthogonal moment equation:
\begin{align}
    \E\left[1\{e\}\, ({g}_{\adj}^*(S_1) - h^*(S_0) - \theta_0^\top \tilde{\Phi}_1) \tilde{\Phi}_1 + 1\{o\}\, \frac{\Pr(e\mid S_1)}{\Pr(o\mid S_1)} \E_{e}[\tilde{\Phi}_1\mid S_{1}] (\bar{Y}^{o,\adj} - g_{\adj}^*(S_1))  \right] = 0
\end{align}

\paragraph{Orthogonality.} For any nuisance $f$, let $\nu_f = f - f_0$, denote the difference with respect to its corresponding true value. More note that the directional derivative with respect to $f$, decomposes into the sum of the directional derivatives with respect to each component of $f$. Thus it suffices to check orthogonality for each of the nuisances $\{g, q, h, p_{e,1}, \{g_t, p_{e,t}, b_{o, t}, p_{o,t,\tau}\}_{2\leq \tau\leq t\leq M}\}$. We provide a proof for each such nuisance below.

Since $\{g, q, h, p_{e,1}, \{g_t, p_{e,t}\}_{2\leq t\leq M}\}$ only appear in moment $m_1$, we check orthogonality only of that moment with respect to these components:
\begin{align*}
    D_{p_{e,1}}m_1(\theta^*;f^*)[\nu_{p_{e,1}}] :=~& \Pr(e)\left(\E_e[\tilde{\Phi}_1\, \theta_0^\top \nu_{p_{e,1}}(S_0)] - \E_{e}\left[\left(\tilde{g}_{\adj}^*(S_1, S_0) - \theta_0^\top \tilde{\Phi}_1\right) \, \nu_{p_{e,1}}(S_0)\right]\right) = 0\\
    D_{p_{e,t}}m_1(\theta^*;f^*)[\nu_{p_{e,t}}] :=~& \Pr(e)\, \E_e[\tilde{\Phi}_1\, \theta_{o,t}^\top \nu_{p_{e,1}}(S_0)] = 0\\
    D_h m_1(\theta^*;f^*)[\nu_h] :=~& \Pr(e)\,\E_e[-\tilde{\Phi}_1\, \nu_h(S_0)] = 0\\
    D_q m_1(\theta^*;f^*)[\nu_q] :=~& \Pr(o)\,\E_o[(\bar{Y}^{o,\adj} - g_{\adj}^*(S_1))\, \nu_q(S_1)]
    = \Pr(o)\,\E_o[\E[\bar{Y}^{o,\adj} - g_{\adj}^*(S_1)\mid S_1]\, \nu_q(S_1)] = 0\\
    D_g m_1(\theta^*;f^*)[\nu_g] =~& \Pr(e)\,\E_e[\tilde{\Phi}_1\, \nu_g(S_1)] - \Pr(o) \E_o[q^*(S_1)\, \nu_g(S_1)] = 0 \tag{by Lemma~\ref{lem:transform}}\\
    D_{g_t} m_1(\theta^*;f^*)[\nu_{g_t}] =~& -\Pr(e)\,\E_e[\tilde{\Phi}_1\,\theta_{o,t}^\top\, \nu_{g_t}(S_1)] + \Pr(o) \E_o[q^*(S_1)\, \theta_{o,t}^\top \nu_{g_t}(S_1)] = 0 \tag{by Lemma~\ref{lem:transform}}
\end{align*}

Since for each $t \in [2, M]$, the components $\{b_{o, t}, p_{o,t,t}, p_{o,\tau,t}\}_{t< \tau\leq M}$ only appear in moments $m_{t}$, we check orthogonality only of $m_t$ with respect to them:
\begin{align*}
    D_{b_{o,t}} m_t(\theta^*;f^*)[\nu_{b_{o,1}}] =~& \Pr(o)\E_o\left[\tilde{\Phi}_{o,t,t}\nu_{b_{o,1}}(S_{t-1})\right] = 0\\
    D_{p_{o,\tau, t}} m_t(\theta^*;f^*)[\nu_{p_{o,\tau, t}}] =~& \Pr(o)\E_o\left[\tilde{\Phi}_{o,t,t}\, \theta_{o,\tau}^\top\, \nu_{p_{o,\tau, t}}(S_{t-1})\right] = 0\\
    D_{p_{o,t,t}} m_t(\theta^*;f^*)[\nu_{p_{o,t,t}}] =~& \Pr(o)\,\left(\E_o[\tilde{\Phi}_{o,t,t}\, \theta_{o,t}^\top \nu_{p_{o,t,t}}(S_{t-1})] - \E_{o}\left[\left(\tilde{Y}_t - \sum_{\tau=t}^M\theta_{o,\tau}^\top \tilde{\Phi}_{o,\tau,t}\right)\, \nu_{p_{o,t,t}}(S_{t-1})\right]\right) = 0
\end{align*}

\begin{lemma}\label{lem:second-order-frechet-mt}
The second order Frechet derivative of the moment $m_t$ for any $t\geq 2$ satisfies:
\begin{align*}
   \forall f, \tilde{f}: D_{ff} m_t(\theta^*; \tilde{f})[\nu_f, \nu_f] :=~&  \E\left[1\{o\}\,\nu_{p_{o,t,t}}(S_{t-1})\, \nu_{{b}_{o,t}}(S_{t-1})\right]
    - \sum_{j=t+1}^m \E\left[1\{o\}\,\nu_{p_{o,t,t}}(S_{t-1})\, \nu_{{p}_{o,j,t}}(S_{t-1})^\top\right]\, \theta_{o,j}^*\\
   ~& - 2\, \E\left[1\{o\}\,\nu_{{p}_{o,t,t}}(S_{t-1})\, \nu_{{p}_{o,t,t}}(S_{t-1})^\top\right]\, \theta_{o,t}^*
\end{align*}
\end{lemma}
\begin{proof}
Note that by the definition of the Frechet derivative and the chain rule of differentiation:
\begin{align*}
    D_{ff} m_t(\theta^*; \tilde{f})[\nu_f, \nu_f] =~& \sum_{j=t}^{m} D_{p_{o,j,t}, b_{o,t}} m_t(\theta^*; \tilde{f})[\nu_{p_{o,j,t}}, \nu_{b_{o,t}}]
    + \sum_{j=t}^{m}\sum_{j'=j}^m D_{p_{o,j,t}, p_{o,j',t}} m_t(\theta^*; \tilde{f})[\nu_{p_{o,j,t}}, \nu_{p_{o,j',t}}]
\end{align*}
However, note that by the definition of $m_t$, it only contains quadratic nuisance terms of the form $p_{o,t,t}\,p_{o,j,t}$ for $j\geq t$ and $p_{o,t,t}\, b_{o,t}$. Thus we have that all other second derivative terms will be zero and we can write:
\begin{align*}
    D_{ff} m_t(\theta^*; \tilde{f})[\nu_f, \nu_f]  =~&  D_{p_{o,t,t}, b_{o,t}} m_t(\theta^*; \tilde{f})[\nu_{p_{o,t,t}}, \nu_{b_{o,t}}]
    + \sum_{j=t}^{m} D_{p_{o,t,t}, p_{o,j,t}} m_t(\theta^*; \tilde{f})[\nu_{p_{o,t,t}}, \nu_{p_{o,j,t}}]
\end{align*}
By simple calculus each of these terms can be shown to take the form given in the lemma.
\end{proof}

\begin{lemma}\label{lem:second-order-frechet-m1}
The second order Frechet derivative of the moment $m_1$ satisfies:
\begin{align*}
   \forall f, \tilde{f}: D_{ff} m_1(\theta^*; \tilde{f})[\nu_f, \nu_f] :=~&  - \E\left[1\{e\}\, \nu_{p_{e,1}}(S_0)\, \nu_{g}(S_1)\right] + \E\left[1\{e\}\,\nu_{p_{e,1}}(S_0)\, \nu_{h}(S_0)\right]\\
   ~& + \sum_{t=2}^M \left(\E\left[1\{e\}\,\nu_{p_{e,1}}(S_0)\, \nu_{g_t}(S_1)^\top\right]\,\theta_{o,t}^* - \E\left[1\{e\}\,\nu_{p_{e,1}}(S_0)\, \nu_{p_{e,t}}(S_0)^\top\right]\,\theta_{o,t}^*\right)\\
   ~& - 2\, \E\left[1\{e\}\,\nu_{{p}_{e,1}}(S_{0})\, \nu_{{p}_{e,1}}(S_{0})^\top\right]\, \theta_{0}^*\\
   ~& - \E\left[1\{o\}\, \nu_{q}(S_1)\, \nu_{g}(S_1)\right]
   + \sum_{t=2}^M \E\left[1\{o\}\, \nu_{q}(S_1)\, \nu_{g_t}(S_1)^\top\right]\, \theta_{o,t}^*
\end{align*}
\end{lemma}
\begin{proof}
Note that by the definition of the Frechet derivative and the chain rule of differentiation:
\begin{align*}
    D_{ff} m_1(\theta^*; \tilde{f})[\nu_f, \nu_f] =~& \sum_{f_1, f_2 \in \{g, q, h, p_{e,1}, \{g_t, p_{e,t}\}_{2\leq \tau\leq t\leq M}\}} D_{f_1, f_2} m_1(\theta^*; \tilde{f})[\nu_{f_1}, \nu_{f_2}]
\end{align*}
However, note that by the definition of $m_t$, it only contains quadratic nuisance terms of the form $p_{e,1}\, f_2$ for $f_2\in \{g, h, p_{e,1}, \{g_t, p_{e,t}\}_{2\leq t\leq M}\}$ and of the form $q\, f_2$ for $f_2\in \{g, \{g_t\}_{2\leq \tau\leq t\leq M}\}$. Thus we have that all other second derivative terms will be zero and we can write:
\begin{align*}
    D_{ff} m_1(\theta^*; \tilde{f})[\nu_f, \nu_f]  =~&  D_{p_{e,1},g}  m_1(\theta^*; \tilde{f})[\nu_{p_{e,1}}, \nu_{g}]
    + D_{p_{e,1},h}  m_1(\theta^*; \tilde{f})[\nu_{p_{e,1}}, \nu_{g}]\\
    ~&~ + \sum_{t=2}^M \left(D_{p_{e,1},g_t}  m_1(\theta^*; \tilde{f})[\nu_{p_{e,1}}, \nu_{g_t}] + D_{p_{e,1},p_{e,t}}  m_1(\theta^*; \tilde{f})[\nu_{p_{e,1}}, \nu_{p_{e,t}}]\right)\\
    ~&~ + D_{p_{e,1},p_{e,1}}  m_1(\theta^*; \tilde{f})[\nu_{p_{e,1}}, \nu_{p_{e,1}}]\\
    ~&~ + D_{q,g}  m_1(\theta^*; \tilde{f})[\nu_{q}, \nu_{g}] + \sum_{t=2}^M D_{q, g_t}  m_1(\theta^*; \tilde{f})[\nu_{q}, \nu_{g_t}]
\end{align*}
By simple calculus each of these terms can be shown to take the form given in the lemma.
\end{proof}

\subsection{Proof of Theorem~\ref{thm:dyn-estimation}}\label{app:normality}

\begin{proof} 
Finally, let $\E_S[\cdot]$ denote the empirical average over the samples in $S$ and let:
\begin{align}
    m_S(\theta; f) :=~& \E_S\left[\psi(Z; \theta, f)\right] & m(\theta; f) :=~& \E\left[\psi(Z; \theta, f)\right]
\end{align}
Observe that the estimator from Equation~\eqref{eqn:z-estimator-alg} can be equivalently viewed as the solution to an cross-fitted plug-in empirical version of the following vector of moment conditions: 
\begin{align}
    \frac{1}{2}\sum_{O\in \{S,S'\}} m_S(\hat{\theta}, \hat{f}_O) = 0
\end{align}
where each $\hat{f}_O$ is trained on samples outside of set $O$. Moreover, the true parameter $\theta^*$ satisfies the population moment conditions at the true nuisance parameters:
\begin{align}
    m(\theta^*, f^*) = 0
\end{align}
Furthermore, by Theorem~\ref{thm:dyn-ortho} the moment vector $m(\theta, f)$ satisfies the property of Neyman orthogonality with respect to $f$. Moreover, by Lemma~\ref{lem:second-order-frechet-m1}, Lemma~\ref{lem:second-order-frechet-mt} and Assumption~\ref{ass:nuisance-rate}, the second order term of $m(\hat{\theta}, \hat{f}_O)$ in a second-order Taylor expansion around $f^*$ is $o_p(n^{-1/2})$, for every $O\in \{S, S'\}$.

Moreover, the Jacobian $J:=\nabla_{\theta} m(\theta^*, f^*)$ of the moment vector $m$ at the true values $\theta^*, f^*$ is a block upper triangular matrix whose block values are of the form:
\begin{align}
\forall 2\leq t\leq j \leq m: J_{t, j} =~& \Pr(o)\E_o[\Cov_o(\Phi_{o,t}, \Phi_{o,j}\mid S_{t-1})]\\
\forall 2\leq j \leq m: J_{1,j} =~& \Pr(o)\E_o[q^*(S_1)\, \Phi_{o,t}^\top] \\
J_{1,1}=~& \Pr(e) \E_e[\Cov_e(\Phi_{1},\Phi_1\mid S_0)]
\end{align}
Thus by our strict average overlap assumption, its diagonal block values satisfy that $J_{t, t}\succeq \lambda I$.  Hence, the minimum eigenvalue of $J$ is at least $\lambda$.

Thus our setting and our estimator satisfy all the assumptions required to apply Theorem~3.1 of \cite{chernozhukov2018double} to get the following result: if we let:
\begin{align}
    \Sigma = \E\left[\psi(Z;\theta^*, f^*)\, \psi(Z;\theta^*, f^*)^\top\right]
\end{align}
and $V=J^{-1} \Sigma (J^{-1})'$, we have that:
\begin{align}
    \sqrt{n} V^{-1/2} (\hat{\theta} - \theta^*) = \frac{1}{\sqrt{n}}\sum_{i=1}^n \bar{\psi}(Z_i) + o_p(1) \to_d N(0, I_{d\cdot m})
\end{align}
where:
\begin{align}
    \bar{\psi}(\cdot) = -V^{-1/2} J^{-1}\, \psi(\cdot; \theta^*, f^*)
\end{align}

The second part of the theorem on the construction of confidence intervals follows then directly by Corollary 3.1 of \cite{chernozhukov2018double}.
\end{proof}

\subsection{Asymptotic variance characterization}
We further analyze and decompose the variance $V$ above. In particular, observe that:
\begin{align*}
    J =~& - \begin{bmatrix}
    A \Pr(e) & B \Pr(o)\\
    0 & C\, \Pr(o)
    \end{bmatrix}\\
    A =~& \E_e[\tilde{\Phi}_1 \tilde{\Phi}_1^\top]\\
    B =~& \begin{bmatrix} B_2 & \ldots & B_M \end{bmatrix}\\
    B_t =~& \begin{bmatrix}
    \E_o[q_0(S_1,S_0)\, T_t^\top] & \ldots & \E_o[q_0(S_1,S_0)\, T_2^\top]
    \end{bmatrix}
\end{align*}
And:
\begin{align*}
    C =~& \begin{bmatrix} C_2 & 0 & \ldots & 0\\
    0 & C_3 & \ldots & 0\\
    \ldots & \ldots & \ldots & \ldots \\
    0 & 0 & \ldots & C_M
    \end{bmatrix} 
    &
    C_t =~& \begin{bmatrix}
    \E_o[\tilde{\Phi}_{t, t}\, \tilde{\Phi}_{t,t}^{\top}] & 0 & \ldots & 0\\
    \E_o[\tilde{\Phi}_{t, t-1}\, \tilde{\Phi}_{t-1, t-1}^\top] & \E_o[\tilde{\Phi}_{t-1, t-1}\, \tilde{\Phi}_{t-1, t-1}^\top] & \ldots & 0\\
    \ldots & \ldots & \ldots & 0\\
    \E_o[\tilde{\Phi}_{t, 2}\, \tilde{\Phi}_{2, 2}^\top] & \E_o[\tilde{\Phi}_{t-1, 2}\, \tilde{\Phi}_{2, 2}^\top] & \ldots & \E_o[\tilde{\Phi}_{2, 2}\, \tilde{\Phi}_{2, 2}] 
    \end{bmatrix}
\end{align*}

Observe that we can write:
\begin{align}
    J =~& - \Pr(e) \begin{bmatrix}
    A & B \frac{\Pr(o)}{\Pr(e)}\\
    0 & C \frac{\Pr(o)}{\Pr(e)}
    \end{bmatrix}
    &
    J^{-1} =~& - \frac{1}{\Pr(e)} \begin{bmatrix}
    A^{-1} & -A^{-1} B C^{-1}\\
    0 & C^{-1} \frac{\Pr(e)}{\Pr(o)}
    \end{bmatrix}
\end{align}
Moreover, we can write:
\begin{equation}
    \Sigma = \Pr(e) \begin{bmatrix}
    \E_e[\psi_{0, e}(Z;\theta_*, f_0)\, \psi_{0, e}(Z;\theta_*, f_0)^\top] & 0\\
    0 & 0
    \end{bmatrix}
    + \Pr(o) \E_o[\psi(Z;\theta_*, f_0)\, \psi(Z;\theta_*, f_0)^\top]
\end{equation}

Leading to:
\begin{equation}
    V = V_e + V_o
\end{equation}
where:
\begin{equation}
    [V_e]_{1:k, 1:k} = \frac{1}{\Pr(e)} A^{-1} \E_e[\psi_{0, e}(Z;\theta_*, f_0)\, \psi_{0, e}(Z;\theta_*, f_0)^\top] A^{-1}
\end{equation}
and
\begin{align}
    V_o =~& \frac{\Pr(o)}{\Pr(e)^2} \begin{bmatrix}
    A & B \frac{\Pr(o)}{\Pr(e)}\\
    0 & C \frac{\Pr(o)}{\Pr(e)}
    \end{bmatrix}^{-1}  \E_o[\psi(Z;\theta_*, f_0)\, \psi(Z;\theta_*, f_0)^\top] \begin{bmatrix}
    A & B \frac{\Pr(o)}{\Pr(e)}\\
    0 & C \frac{\Pr(o)}{\Pr(e)}
    \end{bmatrix}^{-\top}\\
    =~& \frac{\Pr(o)}{\Pr(e)^2} \begin{bmatrix}
    A^{-1} & -A^{-1} B C_{-1}\\
    0 & C^{-1} \frac{\Pr(e)}{\Pr(o)}
    \end{bmatrix} \E_o\left[\psi(Z;\theta_*, f_0)\, \psi(Z;\theta_*, f_0)^\top\right] \begin{bmatrix}
    A^{-1} & 0\\
    -C^{-\top} B^\top A^{-\top} & C^{-1} \frac{\Pr(e)}{\Pr(o)}
    \end{bmatrix}
\end{align}

The $V_e$ part of the variance is the variance if one ignores the uncertainty in the surrogate index model and simply estimates uncertainty as if the surrogate index was the target outcome. So this is the uncertainty in estimating the causal effect of the treatment on the surrogate index. The variance $V_o$ is the influence of the uncertainty of estimating the surrogate index on the final treatment effect.

We can further expand and simplify the variance $V_o$ in particular, we can write the top left $k\times k$ diagonal block in the following simplified form:
\begin{align}
    [V_o]_{1:k, 1:k} =~& \frac{\Pr(o)}{\Pr(e)^2} A^{-1} \left(\E_o[\Psi_{0,o} \Psi_{0,o}^\top] - \sum_{t=2}^M B_t C_t^{-1} \left(\E_o[\Psi_{t} \Psi_{0,o}^\top] -\sum_{t'=2}^{M} \E_o[\Psi_{t} \Psi_{t'}^\top] C_{t'}^{-\top} B_{t'}^\top \right)\right) A^{-1}
\end{align}
which nicely also decomposes into the part that we had without the dynamic effect estimation and the extra part that stems from the estimation of the dynamic effects.

\section{Description of Data Generating Processes}\label{sec_datagen}

 \subsection{Synthetic Data}

 For fully synthetic data we simply generate data based on the linear data generating process presented in Equation~\eqref{eqn:dgp}, with Gaussian exogenous shocks and randomly initialized parameter matrices.

 \subsection{Semi-Synthetic Data}\label{sec:datagen}
As mentioned in Section~\ref{app:semidata}, we generate the semi-synthetic data by leveraging the correlation matrix and some pre-trained models from a real world dataset. In this section, we describe in details how we simulate the dataset step by step.

The real world dataset contains approximately 10k customers. For each customer we collect a time series of their monthly investments, proxies, and revenue trajectory, along with a set of fixed customer characteristics. We extract meaningful information from the real data and then simulate a new synthetic data as follows:
\paragraph{Generate data for initial period.} From the real world dataset we filter one month data on some period $t$ and derive the normalized covariance matrix. In order to not expose the real correlation among variables, we decompose this matrix and recreate eigenvalues and eigenvectors ourselves. For eigenvalues, we keep the top 4 eigenvalues and fit a discontinuous linear regression on the true eigenvalues curve. For eigenvectors, we keep the corresponding 4 vectors and impute random remaining vectors. We then combine these new eigenvalues and eigenvectors to create a new covariance matrix. From this perturbed covariance matrix we draw a sample on a multi-variate gaussian distribution, which we use as the data for the initial period $t_0$ for each customer.
\paragraph{Learn intertemporal auto and cross-correlations.} We first train linear models (e.g. LassoCV) on each proxy and investment to predict each of these outcomes using 6 lagged periods of each investment, 6 lagged periods of each proxy, and a set of time-invariant demographics. For both proxies and investments, we find a large amounts of auto-correlation and small cross-correlation effects, along with some effects from customer characteristics including customer size and level of engagement. Based on these insights, we then draw new coefficients for each model, with decaying trend on lagged variables in different scale, and randomly draw coefficients for demographics.    
\paragraph{Simulate the residual distribution corresponds to each model.} From the pre-trained models, we see how unobserved heterogeneity behaves on each investment and proxy. we find that majority of our residuals follows as a normal distribution, however, we have some unpredictable large outliers on both sides. In order to capture those surprising behaviours, we fit a mixture of different models. For a residual distribution from a single model, we fit mixture gaussian model (n\_component=2) on the 5\% and 95\% band of the residuals and for the two tails, we fit log normal distributions.
\paragraph{Build the panel dataset in a feed-forward manner.} With the initial data set, proxy and investment parametric model coefficients and the residuals, we simulate the data following the equations below:
    \begin{align}
    T_{i,t} =~& \sum_{j=1}^{6}\kappa_{j} T_{i,t-j} + \sum_{j=1}^{6}\alpha_{j} S_{i,t-j} + \lambda D_{i,t} + \eta_{i,t} \\
    S_{i,t} =~& \theta T_{i,t} + \sum_{j=1}^{6}\gamma_{j} S_{i,t-j}+ \beta D_{i,t} + \epsilon_{i,t} 
    \end{align}
where $T$ represents treatment, $S$ represents surrogates(proxies), and $D$ represents demographics. First of all, we use the investment model to predict the current period investment, and add random residuals drawing from our fitted residual mixture distribution. Then we use the proxy model to predict current period proxy from all the controls, and add residuals as well. Other than that, we also add the effect of current period investment to the proxy. Moving to the next period, all the predictive outcomes and controls in the current period will become controls for next period, we could then repeat the process mentioned above again to get the next period prediction. Figure~\ref{fig:simplified_cg} also shows how this forward-feeding iteration works. After repeating this process $m$ times, we have a semi-synthetic panel data ready to run experiment.

\newpage

\section{Coverage Results of Asymptotic Normal Based Intervals}

\subsection{Synthetic Data}

\begin{figure}[h]
    \centering
    \begin{subfigure}{.45\textwidth}
    \centering
    \begin{tabular}{cc|c|c|c|}
    & $n_t=$ & $2000$ & $5000$ & $10000$\\
    \hline
    $n_t=$ & $2$ & $96$ & $93.5$ & $96$ \\
    & $4$ & $79$ & $76$ & $86$ \\
    & $8$  & $53.3$ & $57.5$ & $62$ \\
    \end{tabular}
    \caption{Dynamic Adjusted Surrogate}
    \end{subfigure}
    ~~
    \begin{subfigure}{.45\textwidth}
    \centering
    \begin{tabular}{cc|c|c|c|}
    & $n_t=$ & $2000$ & $5000$ & $10000$\\
    \hline
    $n_t=$ & $2$ & $96.5$ & $93$ & $97$ \\
    & $4$ & $84.5$ & $83.5$ & $87$\\
    & $8$  & $75$ & $71.5$ & $81$\\
    \end{tabular}
    \caption{De-biased Dynamic Adjusted Surrogate}
    \end{subfigure}
    \caption{Synthetic data, high-dimensional, lasso models. Coverage levels averaged across $n_{exp}=100$ experiments and across the $k=2$ treatments. $n$ is number of samples and $n_t$ number of periods of long-term outcome. Target coverage level is $99\%$.}
    \label{fig:synth_lasso_cov}
\end{figure}

\begin{figure}[h]
    \centering
    \begin{subfigure}{.45\textwidth}
    \centering
    \begin{tabular}{cc|c|c|c|}
    & $n_t=$ & $2000$ & $5000$ & $10000$\\
    \hline
    $n_t=$ & $2$ & $98.5$ & $99$ & $98.5$ \\
    & $4$ & $89$ & $91$ & $89.5$ \\
    & $8$  & $80$ & $75$ & $80$ \\
    \end{tabular}
    \caption{Dynamic Adjusted Surrogate}
    \end{subfigure}
    ~~
    \begin{subfigure}{.45\textwidth}
    \centering
    \begin{tabular}{cc|c|c|c|}
    & $n_t=$ & $2000$ & $5000$ & $10000$\\
    \hline
    $n_t=$ & $2$ & $98.5$ & $99$ & $98.5$ \\
    & $4$ & $95.5$ & $94$ & $93$\\
    & $8$  & $90.5$ & $90.5$ & $91$\\
    \end{tabular}
    \caption{De-biased Dynamic Adjusted Surrogate}
    \end{subfigure}
    \caption{Synthetic data, low-dimensional, linear regression models. Coverage levels averaged across $n_{exp}=100$ experiments and across the $k=2$ treatments. $n$ is number of samples and $n_t$ number of periods of long-term outcome. Target coverage level is $99\%$.}
    \label{fig:synth_lr_cov}
\end{figure}

\subsection{Semi-Synthetic Data}

\begin{figure}[h]
    \centering
    \begin{subfigure}{.45\textwidth}
    \centering
    \begin{tabular}{cc|c|c|c|}
    & $n_t=$ & $2000$ & $5000$ & $10000$\\
    \hline
    $n_t=$ & $2$ & $53$ & $42$ & $34$ \\
    & $4$ & $49$ & $42$ & $36$ \\
    & $8$  & $42$ & $26$ & \\
    \end{tabular}
    \caption{Dynamic Adjusted Surrogate}
    \end{subfigure}
    ~~
    \begin{subfigure}{.45\textwidth}
    \centering
    \begin{tabular}{cc|c|c|c|}
    & $n_t=$ & $2000$ & $5000$ & $10000$\\
    \hline
    $n_t=$ & $2$ & $65$ & $60$ & $52$ \\
    & $4$ & $76$ & $71$ & $70$\\
    & $8$  & $87$ & $84$ & \\
    \end{tabular}
    \caption{De-biased Dynamic Adjusted Surrogate}
    \end{subfigure}
    \caption{Semi-Synthetic data, lasso models. Coverage levels averaged across $n_{exp}=100$ experiments and across the $k=3$ treatments. $n$ is number of samples and $n_t$ number of periods of long-term outcome. Target coverage level is $99\%$.}
    \label{fig:semisynth_lasso_cov}
\end{figure}

\section{MSE Results on Further Synthetic and Semi-Synthetic Data}

We present results on the $\ell_2$ error of the recovered coefficients, for synthetic and semi-synthetic data. The algorithms that whose performance we present is as follows:
\begin{enumerate}
    \item total: estimating the effect with hypothetical access to the long-term outcome and no dynamic adjustment
    \item total: estimating the effect with hypothetical access to the long-term outcome and no dynamic adjustment, but first projecting to the surrogates and then estimating the effect on the surrogate index.
    \item adj. total: estimating the effect with hypothetical access to the long-term outcome, applying dynamic adjustment
    \item adj. surrogate: estimating the effect with hypothetical access to the long-term outcome, applying dynamic adjustment, after first projected the adjusted outcome on the surrogates and then estimating the effect on the dynamically adjusted surrogate index
    \item new treat: our dynamically adjusted surrogate index algorithm that uses a separate long-term dataset to estimate the surrogate index and then applying it to the short term dataset.
    \item deb new treat: the debiased (fully orthogonal) version of our dynamically adjusted surrogate index algorithm, based on orthogonal score.
\end{enumerate}

\newpage

\subsection{Synthetic Data}

\begin{figure}[H]
    \centering
    \includegraphics[width=1\textwidth]{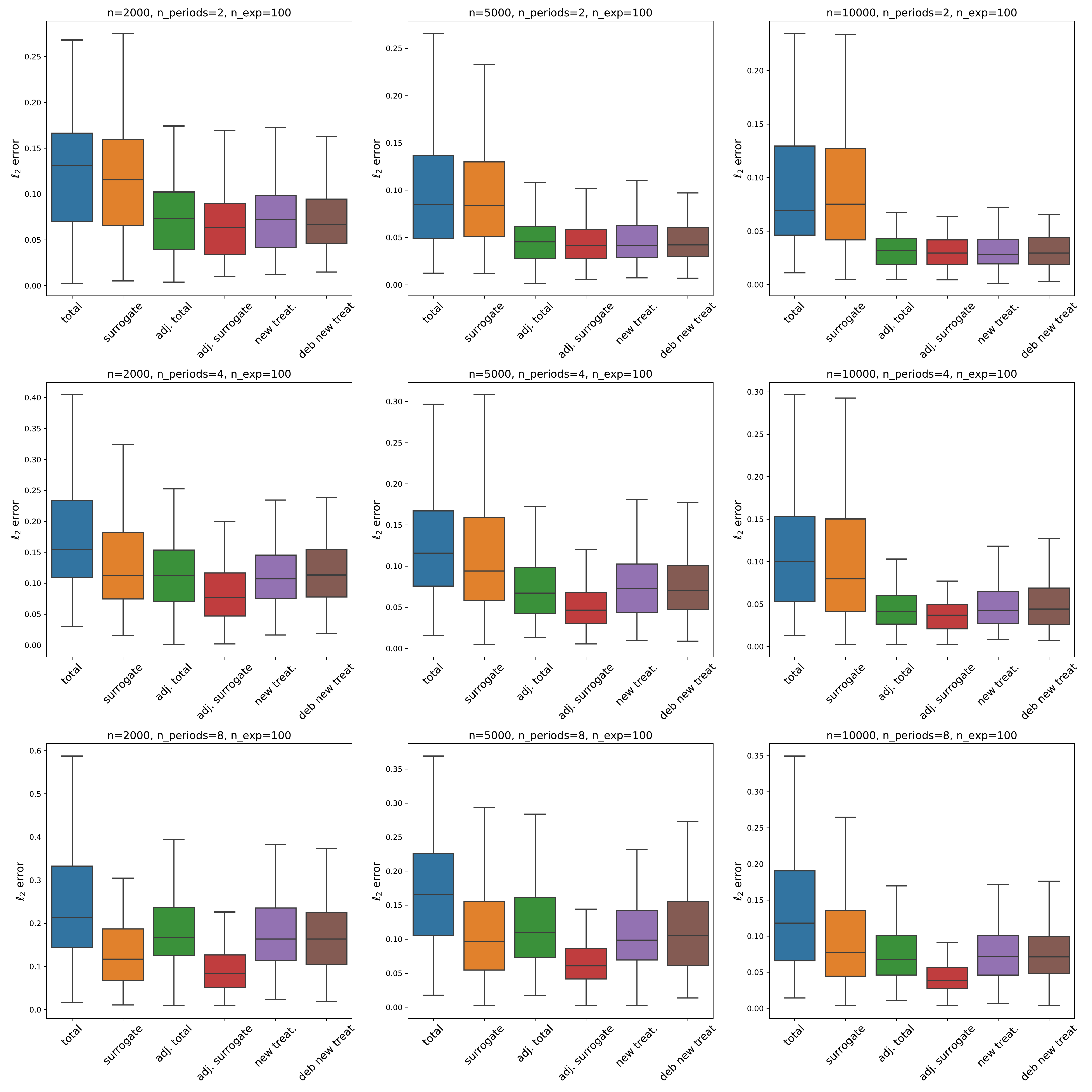}
    \caption{Synthetic data, high-dimensional, lasso models. $n$ is number of samples and $n_{periods}$ number of periods of long-term outcome. $\ell_2$ error of recovered coefficients.}
    \label{fig:synth_lasso_mse}
\end{figure}

\newpage

\subsection{Semi-Synthetic Data}

\begin{figure}[H]
    \centering
    \begin{subfigure}{1\textwidth}
    \includegraphics[width=1\textwidth]{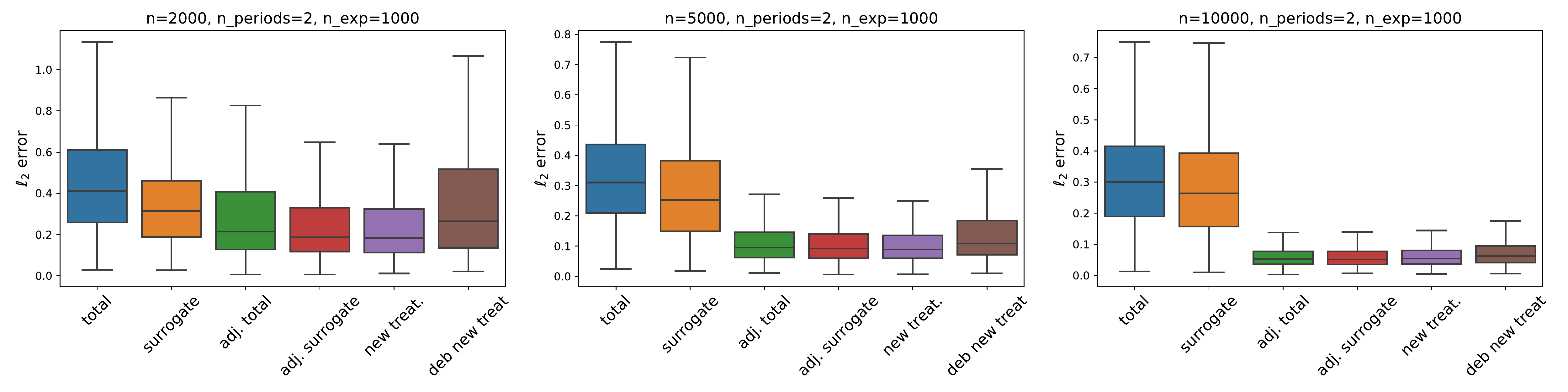}
    \end{subfigure}
    \begin{subfigure}{1\textwidth}
    \includegraphics[width=1\textwidth]{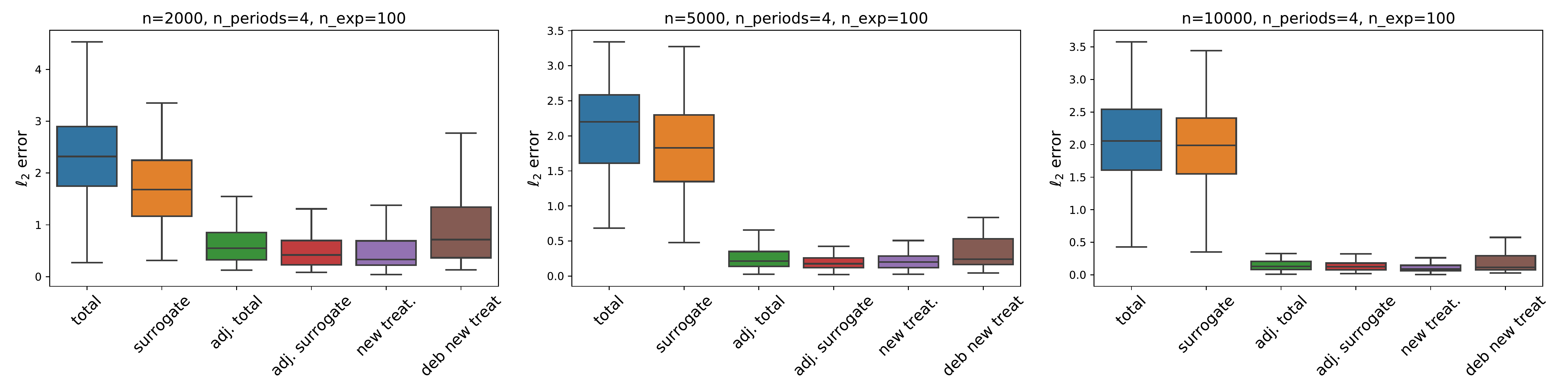}
    \end{subfigure}
    \begin{subfigure}{1\textwidth}
    \includegraphics[width=1\textwidth]{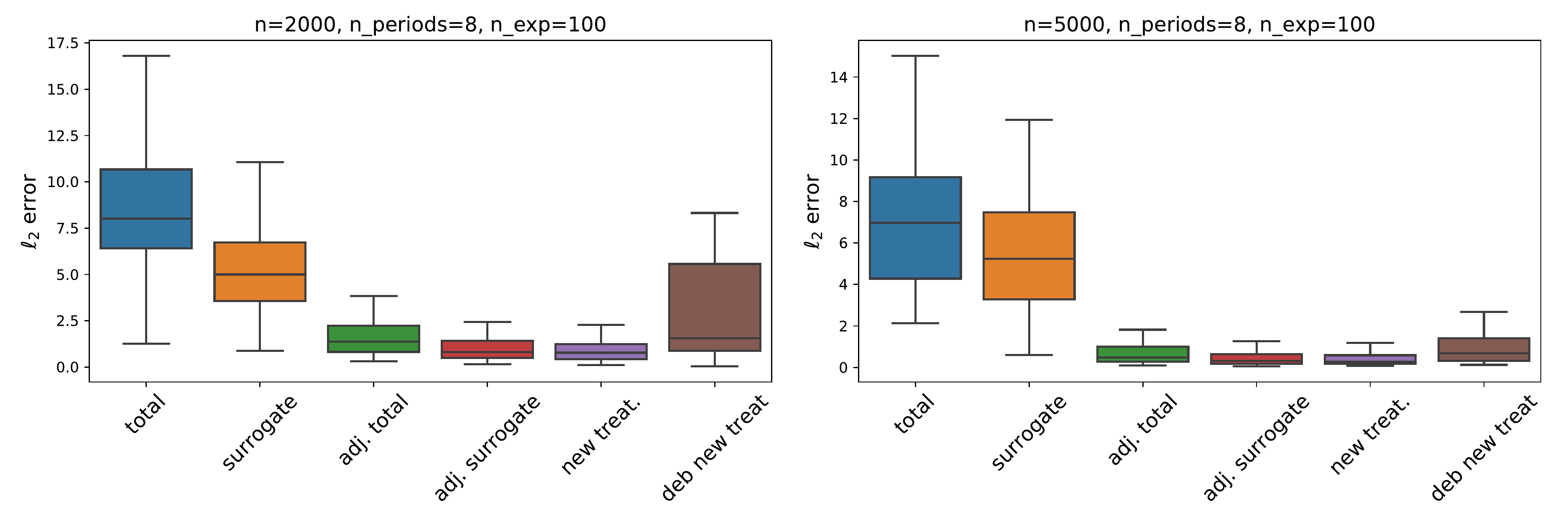}
    \end{subfigure}
    \caption{Semi-Synthetic data, high-dimensional, lasso models. $n$ is number of samples and $n_{periods}$ number of periods of long-term outcome. $\ell_2$ error of recovered coefficients.}
    \label{fig:semisynth_lasso_mse}
\end{figure}

%% file: main.bbl
\begin{thebibliography}{40}
\providecommand{\natexlab}[1]{#1}
\providecommand{\url}[1]{\texttt{#1}}
\expandafter\ifx\csname urlstyle\endcsname\relax
  \providecommand{\doi}[1]{doi: #1}\else
  \providecommand{\doi}{doi: \begingroup \urlstyle{rm}\Url}\fi

\bibitem[Ai \& Chen(2003)Ai and Chen]{Ai2003}
Ai, C. and Chen, X.
\newblock Efficient estimation of models with conditional moment restrictions
  containing unknown functions.
\newblock \emph{Econometrica}, 71\penalty0 (6):\penalty0 1795--1843, 2003.

\bibitem[Athey et~al.(2020)Athey, Chetty, Imbens, and
  Kang]{athey2020estimating}
Athey, S., Chetty, R., Imbens, G., and Kang, H.
\newblock Estimating treatment effects using multiple surrogates: The role of
  the surrogate score and the surrogate index, 2020.

\bibitem[Begg \& Leung(2000)Begg and Leung]{begg2000use}
Begg, C.~B. and Leung, D.~H.
\newblock On the use of surrogate end points in randomized trials.
\newblock \emph{Journal of the Royal Statistical Society: Series A (Statistics
  in Society)}, 163\penalty0 (1):\penalty0 15--28, 2000.

\bibitem[Bodory et~al.()Bodory, Huber, and Laff{\'e}rs]{bodoryevaluating}
Bodory, H., Huber, M., and Laff{\'e}rs, L.
\newblock Evaluating (weighted) dynamic treatment effects by double machine
  learning.

\bibitem[Chakraborty \& Moodie(2013)Chakraborty and Moodie]{Chakraborty2013}
Chakraborty, B. and Moodie, E. E.~M.
\newblock \emph{Semi-parametric Estimation of Optimal DTRs by Modeling
  Contrasts of Conditional Mean Outcomes}, pp.\  53--78.
\newblock Springer New York, New York, NY, 2013.
\newblock ISBN 978-1-4614-7428-9.
\newblock \doi{10.1007/978-1-4614-7428-9_4}.
\newblock URL \url{https://doi.org/10.1007/978-1-4614-7428-9_4}.

\bibitem[Chatterjee \& Bose(2005)Chatterjee and Bose]{Chatterjee2005}
Chatterjee, S. and Bose, A.
\newblock {Generalized bootstrap for estimating equations}.
\newblock \emph{The Annals of Statistics}, 33\penalty0 (1):\penalty0 414 --
  436, 2005.
\newblock \doi{10.1214/009053604000000904}.
\newblock URL \url{https://doi.org/10.1214/009053604000000904}.

\bibitem[Chernozhukov et~al.(2013)Chernozhukov, Chetverikov, and
  Kato]{Chernozhukov2013}
Chernozhukov, V., Chetverikov, D., and Kato, K.
\newblock {Gaussian approximations and multiplier bootstrap for maxima of sums
  of high-dimensional random vectors}.
\newblock \emph{The Annals of Statistics}, 41\penalty0 (6):\penalty0 2786 --
  2819, 2013.
\newblock \doi{10.1214/13-AOS1161}.
\newblock URL \url{https://doi.org/10.1214/13-AOS1161}.

\bibitem[Chernozhukov et~al.(2014)Chernozhukov, Chetverikov, and
  Kato]{Chernozhukov2014}
Chernozhukov, V., Chetverikov, D., and Kato, K.
\newblock {Anti-concentration and honest, adaptive confidence bands}.
\newblock \emph{The Annals of Statistics}, 42\penalty0 (5):\penalty0 1787 --
  1818, 2014.
\newblock \doi{10.1214/14-AOS1235}.
\newblock URL \url{https://doi.org/10.1214/14-AOS1235}.

\bibitem[{Chernozhukov} et~al.(2016){Chernozhukov}, {Escanciano}, {Ichimura},
  {Newey}, and {Robins}]{Chernozhukov2016locally}
{Chernozhukov}, V., {Escanciano}, J.~C., {Ichimura}, H., {Newey}, W.~K., and
  {Robins}, J.~M.
\newblock {Locally Robust Semiparametric Estimation}.
\newblock \emph{arXiv e-prints}, art. arXiv:1608.00033, July 2016.

\bibitem[Chernozhukov et~al.(2018{\natexlab{a}})Chernozhukov, Chetverikov,
  Demirer, Duflo, Hansen, Newey, and Robins]{chernozhukov2018}
Chernozhukov, V., Chetverikov, D., Demirer, M., Duflo, E., Hansen, C., Newey,
  W., and Robins, J.
\newblock {Double/debiased machine learning for treatment and structural
  parameters}.
\newblock \emph{The Econometrics Journal}, 21\penalty0 (1):\penalty0 C1--C68,
  01 2018{\natexlab{a}}.
\newblock ISSN 1368-4221.
\newblock \doi{10.1111/ectj.12097}.
\newblock URL \url{https://doi.org/10.1111/ectj.12097}.

\bibitem[Chernozhukov et~al.(2018{\natexlab{b}})Chernozhukov, Chetverikov,
  Demirer, Duflo, Hansen, Newey, and Robins]{chernozhukov2018double}
Chernozhukov, V., Chetverikov, D., Demirer, M., Duflo, E., Hansen, C., Newey,
  W., and Robins, J.
\newblock Double/debiased machine learning for treatment and structural
  parameters.
\newblock \emph{The Econometrics Journal}, 21\penalty0 (1):\penalty0 C1--C68,
  2018{\natexlab{b}}.

\bibitem[Frangakis \& Rubin(2002)Frangakis and Rubin]{frangakis2002principal}
Frangakis, C.~E. and Rubin, D.~B.
\newblock Principal stratification in causal inference.
\newblock \emph{Biometrics}, 58\penalty0 (1):\penalty0 21--29, 2002.

\bibitem[Freedman et~al.(1992)Freedman, Graubard, and
  Schatzkin]{freedman1992statistical}
Freedman, L.~S., Graubard, B.~I., and Schatzkin, A.
\newblock Statistical validation of intermediate endpoints for chronic
  diseases.
\newblock \emph{Statistics in medicine}, 11\penalty0 (2):\penalty0 167--178,
  1992.

\bibitem[Hern{\'a}n \& Robins(2010)Hern{\'a}n and Robins]{hernan2010causal}
Hern{\'a}n, M.~A. and Robins, J.~M.
\newblock Causal inference, 2010.

\bibitem[Kallus \& Uehara(2019{\natexlab{a}})Kallus and
  Uehara]{kallus2019double}
Kallus, N. and Uehara, M.
\newblock Double reinforcement learning for efficient off-policy evaluation in
  markov decision processes, 2019{\natexlab{a}}.

\bibitem[Kallus \& Uehara(2019{\natexlab{b}})Kallus and
  Uehara]{kallus2019efficiently}
Kallus, N. and Uehara, M.
\newblock Efficiently breaking the curse of horizon in off-policy evaluation
  with double reinforcement learning, 2019{\natexlab{b}}.

\bibitem[Lewis \& Syrgkanis(2020)Lewis and Syrgkanis]{lewis2020doubledebiased}
Lewis, G. and Syrgkanis, V.
\newblock Double/debiased machine learning for dynamic treatment effects, 2020.

\bibitem[Lok \& DeGruttola(2012)Lok and DeGruttola]{lok2012impact}
Lok, J.~J. and DeGruttola, V.
\newblock Impact of time to start treatment following infection with
  application to initiating haart in hiv-positive patients.
\newblock \emph{Biometrics}, 68\penalty0 (3):\penalty0 745--754, 2012.

\bibitem[Neyman(1979)]{Neyman:1979}
Neyman, J.
\newblock $c(\alpha)$ tests and their use.
\newblock \emph{Sankhya}, pp.\  1--21, 1979.

\bibitem[Nie et~al.(2019)Nie, Brunskill, and Wager]{nie2019learning}
Nie, X., Brunskill, E., and Wager, S.
\newblock Learning when-to-treat policies, 2019.

\bibitem[Petersen et~al.(2014)Petersen, Schwab, Gruber, Blaser, Schomaker, and
  van~der Laan]{petersen2014targeted}
Petersen, M., Schwab, J., Gruber, S., Blaser, N., Schomaker, M., and van~der
  Laan, M.
\newblock Targeted maximum likelihood estimation for dynamic and static
  longitudinal marginal structural working models.
\newblock \emph{Journal of causal inference}, 2\penalty0 (2):\penalty0
  147--185, 2014.

\bibitem[Prentice(1989)]{prentice1989surrogate}
Prentice, R.~L.
\newblock Surrogate endpoints in clinical trials: definition and operational
  criteria.
\newblock \emph{Statistics in medicine}, 8\penalty0 (4):\penalty0 431--440,
  1989.

\bibitem[Robins(1986)]{robins1986new}
Robins, J.
\newblock A new approach to causal inference in mortality studies with a
  sustained exposure period-application to control of the healthy worker
  survivor effect.
\newblock \emph{Mathematical modelling}, 7\penalty0 (9-12):\penalty0
  1393--1512, 1986.

\bibitem[Robins et~al.(2000{\natexlab{a}})Robins, Rotnitzky, and Van~der
  Laan]{robins2000comment}
Robins, J., Rotnitzky, A., and Van~der Laan, M.
\newblock Comment on “on profile likelihood” by sa murphy and aw van der
  vaart.
\newblock \emph{Journal of the American Statistical Association--Theory and
  Methods}, 450:\penalty0 431--435, 2000{\natexlab{a}}.

\bibitem[Robins(1994)]{robins1994correcting}
Robins, J.~M.
\newblock Correcting for non-compliance in randomized trials using structural
  nested mean models.
\newblock \emph{Communications in Statistics-Theory and methods}, 23\penalty0
  (8):\penalty0 2379--2412, 1994.

\bibitem[Robins(2004)]{Robins2004}
Robins, J.~M.
\newblock \emph{Optimal Structural Nested Models for Optimal Sequential
  Decisions}, pp.\  189--326.
\newblock Springer New York, New York, NY, 2004.
\newblock ISBN 978-1-4419-9076-1.
\newblock \doi{10.1007/978-1-4419-9076-1_11}.
\newblock URL \url{https://doi.org/10.1007/978-1-4419-9076-1_11}.

\bibitem[Robins \& Ritov(1997)Robins and Ritov]{robins1997toward}
Robins, J.~M. and Ritov, Y.
\newblock Toward a curse of dimensionality appropriate (coda) asymptotic theory
  for semi-parametric models.
\newblock \emph{Statistics in medicine}, 16\penalty0 (3):\penalty0 285--319,
  1997.

\bibitem[Robins et~al.(1992)Robins, Blevins, Ritter, and Wulfsohn]{robins1992g}
Robins, J.~M., Blevins, D., Ritter, G., and Wulfsohn, M.
\newblock G-estimation of the effect of prophylaxis therapy for pneumocystis
  carinii pneumonia on the survival of aids patients.
\newblock \emph{Epidemiology}, pp.\  319--336, 1992.

\bibitem[Robins et~al.(2000{\natexlab{b}})Robins, Hernan, and
  Brumback]{robins2000marginal}
Robins, J.~M., Hernan, M.~A., and Brumback, B.
\newblock Marginal structural models and causal inference in epidemiology,
  2000{\natexlab{b}}.

\bibitem[Robinson(1988)]{robinson:88}
Robinson, P.~M.
\newblock Root-n-consistent semiparametric regression.
\newblock \emph{Econometrica: Journal of the Econometric Society}, pp.\
  931--954, 1988.

\bibitem[Rotnitzky et~al.(2012)Rotnitzky, Lei, Sued, and
  Robins]{rotnitzky2012improved}
Rotnitzky, A., Lei, Q., Sued, M., and Robins, J.~M.
\newblock Improved double-robust estimation in missing data and causal
  inference models.
\newblock \emph{Biometrika}, 99\penalty0 (2):\penalty0 439--456, 2012.

\bibitem[Scharfstein et~al.(1999)Scharfstein, Rotnitzky, and
  Robins]{scharfstein1999adjusting}
Scharfstein, D.~O., Rotnitzky, A., and Robins, J.~M.
\newblock Adjusting for nonignorable drop-out using semiparametric nonresponse
  models.
\newblock \emph{Journal of the American Statistical Association}, 94\penalty0
  (448):\penalty0 1096--1120, 1999.

\bibitem[Singh et~al.(2020)Singh, Xu, and Gretton]{singh2020kernel}
Singh, R., Xu, L., and Gretton, A.
\newblock Kernel methods for policy evaluation: Treatment effects, mediation
  analysis, and off-policy planning.
\newblock \emph{arXiv preprint arXiv:2010.04855}, 2020.

\bibitem[Spokoiny \& Zhilova(2015)Spokoiny and Zhilova]{Spokoiny2015}
Spokoiny, V. and Zhilova, M.
\newblock {Bootstrap confidence sets under model misspecification}.
\newblock \emph{The Annals of Statistics}, 43\penalty0 (6):\penalty0 2653 --
  2675, 2015.
\newblock \doi{10.1214/15-AOS1355}.
\newblock URL \url{https://doi.org/10.1214/15-AOS1355}.

\bibitem[Thomas \& Brunskill(2016)Thomas and Brunskill]{Thomas2016}
Thomas, P.~S. and Brunskill, E.
\newblock Data-efficient off-policy policy evaluation for reinforcement
  learning.
\newblock In \emph{Proceedings of the 33rd International Conference on
  International Conference on Machine Learning - Volume 48}, ICML'16, pp.\
  2139--2148. JMLR.org, 2016.

\bibitem[Tran et~al.(2019)Tran, Yiannoutsos, Wools-Kaloustian, Siika, Van
  Der~Laan, and Petersen]{tran2019double}
Tran, L., Yiannoutsos, C., Wools-Kaloustian, K., Siika, A., Van Der~Laan, M.,
  and Petersen, M.
\newblock Double robust efficient estimators of longitudinal treatment effects:
  Comparative performance in simulations and a case study.
\newblock \emph{The international journal of biostatistics}, 15\penalty0 (2),
  2019.

\bibitem[van~der Laan \& Gruber(2011)van~der Laan and Gruber]{van2011targeted}
van~der Laan, M.~J. and Gruber, S.
\newblock Targeted minimum loss based estimation of an intervention specific
  mean outcome.
\newblock 2011.

\bibitem[Vansteelandt \& Sjolander(2016)Vansteelandt and
  Sjolander]{Vansteelandt2016}
Vansteelandt, S. and Sjolander, A.
\newblock Revisiting g-estimation of the effect of a time-varying exposure
  subject to time-varying confounding.
\newblock \emph{Epidemiologic Methods}, 5\penalty0 (1):\penalty0 37 -- 56,
  2016.
\newblock \doi{https://doi.org/10.1515/em-2015-0005}.
\newblock URL
  \url{https://www.degruyter.com/view/journals/em/5/1/article-p37.xml}.

\bibitem[Vansteelandt et~al.(2014)Vansteelandt, Joffe,
  et~al.]{vansteelandt2014structural}
Vansteelandt, S., Joffe, M., et~al.
\newblock Structural nested models and g-estimation: the partially realized
  promise.
\newblock \emph{Statistical Science}, 29\penalty0 (4):\penalty0 707--731, 2014.

\bibitem[Zhilova(2020)]{Zhilova2020}
Zhilova, M.
\newblock {Nonclassical Berry-Esseen inequalities and accuracy of the
  bootstrap}.
\newblock \emph{The Annals of Statistics}, 48\penalty0 (4):\penalty0 1922 --
  1939, 2020.
\newblock \doi{10.1214/18-AOS1802}.
\newblock URL \url{https://doi.org/10.1214/18-AOS1802}.

\end{thebibliography}
